\documentclass[10pt,english]{article}
\usepackage[T1]{fontenc}
\usepackage[latin1]{inputenc}
\usepackage{graphicx}    
\usepackage{amsmath,amssymb,amsfonts,amsthm}
\usepackage{enumerate,subfigure,array}
\usepackage{float,mathrsfs}
\usepackage[babel=true]{csquotes}
\usepackage{amsthm}
\usepackage{enumerate,subfigure,array}
\newcolumntype{g}{>{$}c<{$}}
\usepackage{bbm}

\newtheorem{theorem}{\textbf{Theorem}}[section]

\newtheorem{proposition}{\textbf{Proposition}}[section]

\newtheorem{remark}{\textbf{Remark}}[section]

\newcommand{\cF}{\mathcal{F}}

\newcommand{\R}{\mathbb{R}}

\newcommand{\E}{\mathbb{E}}

\newcommand{\N}{\mathbb{N}}

\newcommand{\dd}{\textup{d}}
\newcommand{\epsibis}{\nu}
\newcommand{\phis}{\varphi_\textup{s}}
\newcommand{\phic}{\varphi_\textup{c}}
\newcommand{\ios}{\iota_\textup{s}}
\newcommand{\ioc}{\iota_\textup{c}}
\newcommand{\phicos}{\phi_\textup{s}}
\newcommand{\phicoc}{\phi_\textup{c}}

\newcommand{\midprice}{M}
\newcommand{\resfact}{\gamma}
\newcommand{\adjlag}{L_\text{adj}}

\newcommand{\regwin}{\Delta_\textup{RW}}
\newcommand{\tp}{^\top} 
\newcommand{\indi}[1]{\mathbbm{1}_{\{#1\}}}
\newcommand{\cov}{\mathscr{C} }

\begin{document}

\title{Extension and calibration \\ of a Hawkes-based optimal execution model
}
\author{Aur\'elien Alfonsi, Pierre Blanc\thanks{Universit\'e Paris-Est, CERMICS, Projet MATHRISK
    ENPC-INRIA-UMLV, 6 et 8 avenue Blaise Pascal, 77455 Marne La Vall\'ee, Cedex
    2, France, e-mails : alfonsi@cermics.enpc.fr, blancp@cermics.enpc.fr. P. Blanc is grateful to Fondation Natixis for his Ph.D. grant. This research also benefited
    from the support of the ``Chaire Risques Financiers'', Fondation du Risque.
}} 

\date{\today}
\maketitle

 \begin{abstract}
We provide some theoretical extensions and a calibration protocol for our former dynamic optimal execution model. The Hawkes parameters and the propagator are estimated independently on financial data from stocks of the CAC40. Interestingly, the propagator exhibits a smoothly decaying form with one or two dominant time scales, but only so after a few seconds that the market needs to adjust after a large trade. Motivated by our estimation results, we derive the optimal execution strategy for a multi-exponential Hawkes kernel and backtest it on the data for round trips. We find that the strategy is profitable on average when trading at the midprice, which is in accordance with violated martingale conditions. However, in most cases, these profits vanish  when we take bid-ask costs into account. 



{\bf AMS (2010)}: 91G99, 91B24, 91B26, 60G55, 49J15.

\end{abstract}



\section{Introduction}

In the last fifteen years, the literature in quantitative finance has been enriched by many studies on optimal execution problems. The principle is as follows: one considers a particular trader who wants to liquidate a quantity $x_0$ of assets on the time interval~$[0,T]$. Thus, if $X_t$ is the position at time~$t$, one has $X_0=x_0$ and $X_{T+}=0$: $x_0>0$ (resp. $x_0<0$) corresponds to to a sell (resp. buy) program. The trader uses an execution strategy of minimal expected cost, which should take into account the fact that large trades have an impact on the market price. The works of Bertsimas and Lo~\cite{BL} and Almgren and Chriss~\cite{AC} are pioneers in this area. They have been followed by several authors who suggested extensions of their framework, such as Obizhaeva and Wang~\cite{OW} who considered a model that includes transient price impact. This feature allows to reproduce the mean-reversion that is observed in intra-day prices. On average, when a large trade impacts the market price, a fraction of this impact vanishes over time. 

In Alfonsi and Blanc~\cite{AB_DynHawkes}, we introduce a model where other liquidity takers trade the same asset as the large trader, and share the same price impact profile as her.
In this model, the volumes of incoming trades is described by a c\`adl\`ag (right continuous left limits) pure jump process $N_t$, and the market price $P_t$ at time~$t$ is given by
\begin{equation}\label{Price_1} P_t=\sum_{\tau < t} \Delta N_\tau \times \left[\frac \nu q +\frac{1 - \nu}q e^{-\rho (t-\tau)} \right],
\end{equation}
where the times $\tau$ are the jump times of the process~$N$ and $\Delta N_\tau=N_\tau-N_{\tau-}$ is the signed volume of the order at time~$\tau$. Thus,  $q>0$ is a measure of market liquidity, $\nu \in [0,1]$ the proportion of permanent impact, and $\rho>0$ the resilience speed of the transient part of the price.
In~\cite{AB_DynHawkes}, the order flow is modeled by a two-dimensional Hawkes process, which allows self and mutual excitation between buy and sell orders. An interesting feature of this model is that it accounts for herding behavior and meta-orders splitting, see Bacry and Muzy~\cite{BacryMuzyQF}. Namely, let $N^+$ and $N^-$ be two nondecreasing c\`adl\`ag pure jump processes that describe respectively the volumes of incoming buy and sell orders. We have $N=N^+-N^-$, and we proposed in~\cite{AB_DynHawkes} the following model for the respective jump intensities of $N^\pm$
\begin{align}\label{kappa_expo}
\kappa^+_t  &= \kappa_{\infty} + \sum_{\tau < t}  \left[ \indi{\Delta N_\tau>0} \phis\left( \frac{\Delta N_\tau}{m_1}\right) + \ \indi{\Delta N_\tau<0} \phic\left(- \frac{\Delta N_\tau}{m_1}\right) \right]e^{-\beta (t-\tau)}, \\
\kappa^-_t  &= \kappa_{\infty} + \sum_{\tau < t}  \left[\indi{\Delta N_\tau<0}\phis\left(- \frac{\Delta N_\tau}{m_1}\right) + \  \indi{\Delta N_\tau>0}\phic\left( \frac{\Delta N_\tau}{m_1}\right) \right]e^{-\beta (t-\tau)},
\end{align}
where $\kappa_\infty \geq 0$ is the  common baseline intensity of $N^+$ and $N^-$, $\beta$ is the resilience speed of the intensity and $\phis, \phic: \mathbb{R}^+ \rightarrow \mathbb{R}^+$ are measurable positive functions that encode intensity feedback. We assume that the sizes of orders are independent variables distributed according to a square integrable probability law $\mu$ on $\R_+$, and $m_1=\int_0^\infty x \mu(dx)$ is the average amplitude of the jumps of $N$. This price model is called MIH, as Mixed-Impact Hawkes.
In this model, we provide a closed-form solution for the optimal liquidation strategy, and determine a set of conditions on $\nu,\rho,\beta,\phis,\phic$ that exclude Price Manipulation Strategies (as defined in~\cite{HS}) from the model. These are referred to as the MIHM (Mixed Impact Hawkes Martingale) conditions.

One of the benefits of the framework introduced in~\cite{AB_DynHawkes} is that it is possible to calibrate the model on financial data, without effectively trading (which can be costly). One only has to observe the order flow and price process of the market, and to estimate the price impact of trades issued by other participants, which is expected to be similar to the impact that the liquidating trader would have.
The aim of the present paper is to conduct such a calibration on real stock data. This enables us to evaluate the realism of the theoretical price model of~\cite{AB_DynHawkes}, as well as the performance of the optimal strategy in a practical context. 
Since our main goal is to confront the model to market data, we test the validity of our calibration protocol on simulations and we leave its mathematical justification for further research.

Many studies have explored the estimation of Hawkes parameters in various contexts (see for instance Bacry et al.~\cite{BDM_NP}, Bouchaud and Hardiman~\cite{HBB}, Reynaud-Bouret~\cite{R-B}, Lemonnier and Vayatis~\cite{LV}). The present paper focuses on marked Hawkes processes used to model price jumps triggered by transactions in financial markets, where the marks of the jumps are either the traded volumes or the price jumps. As opposed to most Hawkes models in finance, price moves which do not correspond to trades are treated separately through the propagator function. Propagator price models have been studied extensively in theoretical frameworks such as Gatheral~\cite{Gatheral}, Alfonsi et al.~\cite{ASS} and Gatheral et al.~\cite{GSS}, Bouchaud et al.~\cite{BGPW} and Farmer et al.~\cite{FGLM}. However, to the best of our knowledge, very few empirical studies have described the form of the propagator curve, or only asymptotically. Here, we suggest an estimation protocol for the propagator and discuss the quality of fit of exponential and multi-exponential decays. We also describe the  behavior of the curve on the first seconds, where it is found to have an increasing part.

The paper is structured as follows. First, we present the model in Section~\ref{sec_model}. It extends the one considered in~\cite{AB_DynHawkes} to general decay kernels, while preserving most of its properties. Then, in Section~\ref{sec_cal} we describe our dataset and our calibration method. In particular, we explain how we slightly modify the original model to be in accordance with practical considerations. Section~\ref{section:calib_results} validates our calibration procedure with simulations and discusses the calibration results on real stock data. Eventually, we test in Section~\ref{section:eval_model} the relevance of the optimal execution strategy described in Section~\ref{sec_model} and discuss whether it may constitute Price Manipulation Strategies, i.e. round trips that are profitable in average.

\section{Model settings}\label{sec_model}

In view of its estimation to market data, we make the model of~\cite{AB_DynHawkes} more general by adding further parameters. First, even if it is appealing to see the price as the pure result of past trades, equation~\eqref{Price_1} is probably too restrictive and one should add some noise. Besides, we know that adding a martingale to the price process does not change the main results on the model, see Remark 2.6 in~\cite{AB_DynHawkes}. Second, we chose the resilience on the price and on the intensity to be exponential, and one may like to consider a priori more general decay functions. Thus, we consider the following  propagator model for the price:
\begin{equation}\label{Propag_model}
P_t= \frac{1}{q} \sum_{\tau < t} \Delta N_\tau G(t-\tau) + \sigma W_t. 
\end{equation}
The process $W$ is a Brownian motion independent of $N$ that takes into account the non-deterministic noise in limit orders and cancellations. The parameter $\sigma>0$ tunes the volatility of this noise. The function $G : \R_+ \mapsto \R$ is the propagator function of the market, that encodes the average evolution of the price between two market orders, which takes form through limit orders and cancellations. As before, $q>0$ describes the market liquidity and allows to normalize $G$ such that $G(0)=1$.  The propagator model is the same as the one considered by Alfonsi et al~\cite{ASS} and Gatheral et al.~\cite{GSS} and generalizes~\eqref{Price_1}. Similar models have been considered for instance by Bouchaud et al.~\cite{BGPW} and Gatheral~\cite{Gatheral}. In the same way, we consider a general decay function $K:\R_+ \mapsto \R_+$ for the intensities of $N^+$ and $N^-$. Namely, we assume that the jump intensities of  $N^+$ and $N^-$ are respectively given by
\begin{align}
\kappa^+_t  &= \kappa_{\infty} + \sum_{\tau < t}  \left[ \indi{\Delta N_\tau>0} \phis\left( \frac{\Delta N_\tau}{m_1}\right) + \ \indi{\Delta N_\tau<0} \phic\left(- \frac{\Delta N_\tau}{m_1}\right) \right]K (t-\tau), \label{kappa_gen} \\
\kappa^-_t  &= \kappa_{\infty} + \sum_{\tau < t}  \left[\indi{\Delta N_\tau<0}\phis\left(- \frac{\Delta N_\tau}{m_1}\right) + \  \indi{\Delta N_\tau>0}\phic\left( \frac{\Delta N_\tau}{m_1}\right) \right]K (t-\tau). \nonumber
\end{align}
with $K(0)=1$. We also introduce the average self-excitation $\ios$ and the average cross-excitation $\ioc$ $$\ios = \int_0^\infty \phis(v/m_1) \mu(\dd v) \text{ and } \ioc =\int_0^\infty \phic(v/m_1) \mu(\dd v).$$ Therefore, the model presented in~\cite{AB_DynHawkes} corresponds to the exponential decay functions $G(t)=e^{-\rho t}$ and $K(t)=e^{-\beta t}$. By estimating more general functions $G(t)$ and $K(t)$ in the sequel, we are able to assess the relevance of the exponential decay assumption. 

\subsection{Markovian specification of the model}\label{section:markov_spec}

Considering general decay kernels is very natural from a modeling point of view. Unfortunately, it generally leads to drop the Markov property of the price process, which is important in the context of optimal execution. Still, for completely monotone decay kernels, it is possible to get back Markovian dynamics for the price. This has already been studied in Alfonsi and Schied~\cite{AS_SICON} for the price propagator model. Considering completely monotone kernels amounts to assume the existence of probability measures $\tilde{\lambda}(\dd \rho)$ and $\tilde{w}(\dd \rho)$ on $\R_+^*$ such that 
\begin{equation}\label{comp_monotone_kernels}
G(t)=\nu+ (1-\nu) \int_{\R_+} e^{-\rho t} \tilde{\lambda}(\dd \rho), \ K(t)=\int_{\R_+} e^{-\rho t} \tilde{w}(\dd \rho).
\end{equation}
Here, for the sake of simplicity, we consider probability measures with finite support. We can then assume without loss of generality\footnote{Note that $G$ and $K$ may still include different decay speeds: one only has to include all the speeds in the $\rho_i$'s and to set some weights $\lambda_i, w_i$ to zero if necessary.}
that
\begin{equation}\label{expo_mixture}
G(u) = \nu + \sum_{i=1}^p \lambda_i \exp(-\rho_i u),
\quad
K(u) = \sum_{i=1}^p w_i \exp(-\rho_i u),
\end{equation}
with $0<\rho_1<\dots <\rho_p$, $\nu, \lambda_i, w_i \ge 0$ such that $\nu + \sum_{i=1}^p \lambda_i = 1$ and $\sum_{i=1}^p w_i=1$. For $i\in \{1,\dots, p\}$, we introduce the following processes
\begin{align}
\dd D_t^i &= -\rho_i \ D_t^i \ \dd t \ + \ \frac{\lambda_i} q \ \dd N_t, \label{def_Di} \\
\textup{d}{\kappa_t^+}^{(i)}  &= \ - \rho_i \ ({\kappa_t^+}^{(i)} - \kappa_{\infty}/p) \ \textup{d}t
\ + \ w_i [\phis(\textup{d}N^+_t/m_1) + \ \phic(\textup{d}N^-_t/m_1)],
\label{def_kpi}  \\
\textup{d}{\kappa_t^-}^{(i)}  &= \ - \rho_i \ ({\kappa_t^-}^{(i)} - \kappa_{\infty}/p) \ \textup{d}t
\ + \ w_i[\phic(\textup{d}N^+_t/m_1) + \ \phis(\textup{d}N^-_t/m_1)].\label{def_kmi}
\end{align}
We also define the process $\dd S_t = \frac\nu q \ \dd N_t$ that describes the permanent impact component of the price. Then, it is easy to check from~\eqref{expo_mixture}, \eqref{Propag_model} and~\eqref{kappa_gen} that
\begin{equation}
P_t=S_t+\sum_{i=1}^p D^i_t +\sigma W_t, \ \kappa_t^\pm=\sum_{i=1}^p{\kappa_t^\pm}^{(i)}, \label{price_kappa_mkv}
\end{equation}
and the process $(P,S,D^i,{\kappa^\pm}^{(i)})$ satisfies the Markov property. 

\begin{remark}\label{rem_mkv}
In the general setting~\eqref{Propag_model} and~\eqref{kappa_gen}, we implicitly assume that the stationarity conditions $(\ios+\ioc)\int_0^\infty K(s)ds<1, \ G$ integrable are satisfied, so that the sums are well-defined. This is no longer required in the Markovian case since the law of $(P_t,S_t,D^i_t,{\kappa^\pm_t}^{(i)};t\ge 0)$ is determined by the initial condition $(P_0,S_0,D^i_0,{\kappa^\pm_0}^{(i)})$. In the particular case $D^i_0=0$ for all~$i$, and only in this case, we have $P_t=P_0+\frac{1}{q} \sum_{0<\tau < t} \Delta N_\tau G(t-\tau) + \sigma W_t$. Thus, if $|D^i_0| [G(t)-G(\infty)] \ll P_t$ for all $i\in \{1,\cdots,p\}$ and all $t\geq t_0$, then the approximation $P_t \approx P_0+\frac{1}{q} \sum_{0<\tau < t} \Delta N_\tau G(t-\tau) + \sigma W_t$ is reasonable for $t\geq t_0$. 
\end{remark}

Besides the Markov property, the particular form~\eqref{expo_mixture} enables us to calculate explicitly the auto-covariance function of the number of jumps as explained by Hawkes in~\cite{H71}, Section~3. This auto-covariance structure is of empirical interest, and serves as a starting point for our calibration procedure, see Section~\ref{section:calib_method_int}. 
The total intensity $\Sigma_t = \kappa^+_t + \kappa^-_t$ has the dynamics
\begin{equation}
\Sigma_t = 2 \kappa_\infty + \iota \int_{-\infty}^t K(t-s) \ \dd J_s,
\nonumber
\end{equation}
where $\iota = \ios+\ioc $ is the average jump size of $\Sigma_t$, and $\dd J_t= [(\phis+\phic)(\textup{d}N^+_t/m_1) + \ (\phis+\phic)(\textup{d}N^-_t/m_1)]/\iota$ has jumps normalized to unity. 
We assume that the stationarity condition $\iota \int_0^\infty K(s) ds <1$ holds, see Theorem~1 in~\cite{BMH2}, and that the intensity process $(\kappa^+_t, \kappa^-_t)$ in its stationary state. We consider the symmetric auto-covariance function $\cov$ of the infinitesimal increments of $J$. It is defined for $\tau > 0$ by
\begin{equation}
\cov(\tau) \ = \ \underset{h\rightarrow0^+}{\lim} \frac1{h^2} \E[ (J_{t+h}-J_t)(J_{t-\tau+h}-J_{t-\tau})]- 4\overline{\kappa}^2
\ = \ \underset{h\rightarrow0^+}{\lim} \frac1h \E[ \Sigma_t \  (J_{t-\tau+h}-J_{t-\tau})]- 4\overline{\kappa}^2,
\label{eqn:autocov_Hawkes_def}
\end{equation}
where $\overline{\kappa} = \kappa_\infty / (1-\iota/\beta)$ is the common stationary mean of $\kappa^+$ and $\kappa^-$. As derived in~\cite{H71}, one gets the self-consistent equation on $\cov$: for $\tau>0$,
\begin{equation}
\cov(\tau) = 2 \overline{\kappa} \iota K(\tau) + \iota \int_{-\infty}^\tau K(\tau-u) \cov(u) \dd u.
\label{eqn:autocov_Hawkes_consist_eq}
\end{equation}

\begin{proposition}\label{prop_autocov} Let us assume that $K$ satisfies~\eqref{expo_mixture} with $w_1,\dots,w_p>0$ and the stationarity condition $\iota \sum_{i=1}^p \frac{w_i}{\rho_i} <1$. Then, the autocovariance function is given by
\begin{equation}
\cov(\tau) = \sum_{j=1}^p a_j \exp(-b_j |\tau|).  \ \tau \in \R^*.
\label{eqn:multi_autocorrel}
\end{equation}
The coefficients $a_1, \cdots, a_p$ and $b_1, \cdots, b_p$ are positive and determined as follows: $b_1<\dots<b_p$ are the distinct roots of the polynomial functions $P(X) = \prod_{i=1}^p (\rho_i - X) - \iota \sum_{i=1}^p w_i \prod_{k\neq i} (\rho_k-X)$ and $(a_1b_1, \cdots, a_pb_p)\tp=\overline{\kappa} \ B^{-1} \ (1,\cdots,1)\tp$, where $B$ is the Cauchy matrix $B_{i,j}=\frac{1}{\rho_i^2-b_j^2}$. 
\end{proposition}
\begin{proof}
 Equation~\eqref{eqn:autocov_Hawkes_consist_eq} then yields for $\tau>0$
\begin{align}
\sum_{j=1}^p a_j \exp(-b_j \tau) 
&= 2 \iota \sum_{i=1}^p w_i \left[ \overline{\kappa} - \sum_{j=1}^p \frac{a_j b_j}{(\rho_i+b_j)(\rho_i-b_j)} \right] \exp(-\rho_i \tau) 
\ + \ \iota \sum_{j=1}^p a_j \left[ \sum_{i=1}^p \frac{w_i}{\rho_i-b_j} \right] \exp(-b_j \tau).
\nonumber 
\end{align}
Therefore, \eqref{eqn:autocov_Hawkes_consist_eq} holds if we have
\begin{align*}
\forall j,  \ \iota  \left[ \sum_{i=1}^p \frac{w_i}{\rho_i-b_j} \right]=1, \ \ \forall i,   \sum_{j=1}^p \frac{a_j b_j}{\rho_i^2-b_j^2}=\overline\kappa.
\end{align*}
The first equation gives precisely $P(b_j)=0$. Since $P(0)>0$ from the stationarity condition and $P(\rho_l)=-\iota w_l \prod_{k \not = l}(\rho_k-\rho_l)$ has the same sign as $(-1)^l$,  we have by the intermediate value theorem that $0<b_1<\rho_1<b_2<\rho_2<\dots<\rho_{p-1}<b_p<\rho_{p}$. These coefficients are distincts and therefore the Cauchy matrix $B$ is invertible. Let $v=B^{-1} \ (1,\cdots,1)\tp$: $v_i$ is the $i^th$ row sum of $B^{-1}$. By Theorem~2 in~\cite{Schechter}, $v_i= -A(b_i^2)/B'(b_i^2)$, where $A(x)=\prod_i (x-\rho_i^2)$, $B(x)=\prod_i (x-b_i^2)$. 
This gives in particular $v_i>0$ and thus $a_i>0$. Last, it is easy to check~\eqref{eqn:multi_autocorrel} is the unique function satisfying~\eqref{eqn:autocov_Hawkes_consist_eq}.
\end{proof}
In the mono-exponential case $p=1$, Proposition~\ref{prop_autocov} gives $\iota = \rho - b, \ a b = (\rho+b)(\rho-b) \overline{\kappa}$,
which yields
\begin{equation}
\cov(\tau) = \frac{\iota(2\rho-\iota)}{2\rho} \times \frac{2 \kappa_\infty}{(1-\iota/\rho)^2} \times \exp(-(\rho-\iota)|\tau|),
\nonumber
\end{equation}
as found by Hawkes in~\cite{H71}.

\subsection{Trading strategies and a generalized no-arbitrage condition}\label{subsec_noarb}

We now specify the trading rules in our model. We denote by $(\mathcal{F}_t)$ the natural filtration generated by the process $(P,S,D^i,{\kappa^\pm}^{(i)})$. As in~\cite{AB_DynHawkes}, we consider a particular trader called ``strategic trader'' and denote by $X_t$ the number of assets she holds at time~$t$.  We assume that the strategy $X$ is $(\mathcal{F}_t)$-adapted, c\`agl\`ad, square integrable and with bounded variations. The c\`agl\`ad (left continuous - right limits) assumption means that the strategic trader is able to react instantly to the flow of trades. For simplicity and tractability, we assume that the trades of the strategic trader affect the price in the same fashion as other trades, but leave unchanged the flow of orders~$N$. To be more precise, we now assume that
$$ dS_t= \frac\nu q  ( \dd N_t \ + \  \dd X_t), \  \  \dd D_t^i = -\rho_i \ D_t^i \ \dd t \ + \ \frac{\lambda_i} q ( \dd N_t \ + \  \dd X_t), $$
but the intensities ${\kappa_t^+}^{(i)}$ and ${\kappa_t^-}^{(i)}$ remain as defined by~\eqref{def_kpi} and~\eqref{def_kmi}. The price as well as the  intensities  $\kappa^+_t$ and $\kappa^-_t$ of buy and sell orders are still defined by~\eqref{price_kappa_mkv}. Last, the cost of the trade $\Delta X_t=X_{t+}-X_t$ at time $t$ is assumed to be given by
$$  \frac{P_t+P_{t+}}{2} \Delta X_t =  P_t \Delta X_t+ \frac{1}{2q} (\Delta X_t)^2.$$
This yields the following cost for a liquidation strategy $X$ on $[0,T]$ (i.e. such that $X_{T+}=0$)
\begin{equation}\label{costX} C(X)=\int_{[0,T)} P_u \ \textup{d}X_u 
\ + \ \frac1{2q} \underset{\tau \in \mathcal{D}_X \cap [0,T)}{\sum} (\Delta X_\tau)^2
 \ - \ P_T X_T \ + \ \frac1{2q} \ X_T^2, 
\end{equation}
where $\mathcal{D}_X$ is the (countable) set of discontinuities of~$X$. 

When considering high-frequency trading, a standard approach is to define arbitrages as strategies that can make money on average, with no specific exogenous signal. Roughly speaking, one may expect that by repeating such strategies one obtains a classical almost sure arbitrage. Thus, Huberman and Stanzl have proposed in~\cite{HS} the following definition of a {\it Price Manipulation Strategy}: this is a strategy $X$ such that $X_0=X_{T+}=0$ and $\E[C(X)]<0$. 
\begin{theorem}\label{thm_MIHMext} The model excludes Price Manipulation Strategies if, and only if $P_t$ is a $(\cF_t)$-martingale when $X_t=0$ for any~$t$. In this case, the optimal strategy is the one given by Theorem 2 (see also Section 1.3) in Alfonsi and Schied~\cite{AS_SICON}.

Besides, under the specification~\eqref{def_kpi}, \eqref{def_kmi} and~\eqref{price_kappa_mkv} of the order flow~$N=N^+-N^-$, the model does not admit PMS if, and only if,
\begin{equation}\label{cond_MIHM_ext}
\forall i \in \{1,\dots,p\}, \quad  (\ios-\ioc) w_i = \lambda_i \rho_i, 
\quad  \frac{m_1}q ({\kappa_0^+}^{(i)}-{\kappa_0^-}^{(i)}) - \rho_i D_0^i=0,
\end{equation}
and $\phis\left(y/m_1\right)-\phic\left(y/m_1\right)=(\ios-\ios)y/m_1$ for all $y\ge 0$ such that $\forall \epsilon>0, \ \mu((y-\epsilon,y+\epsilon))>0$. 
\end{theorem}
This theorem extends Theorem~2.1 and Proposition~5.1 of~\cite{AB_DynHawkes} to completely monotone kernels $G$ and $K$. Its proof relies on the same arguments that we recall briefly in Appendix~\ref{appendix:proof_MIHMext}. An interesting consequence of~\eqref{cond_MIHM_ext} is the connection made between the price propagator and the decay kernel of the intensity. For general completely monotone functions~\eqref{comp_monotone_kernels}, this yields in particular the following condition:
\begin{equation}
\forall \rho>0, \quad (\ios-\ioc) \ \tilde{w}(\dd \rho) \ = \ (1-\nu) \ \rho\  \tilde{\lambda}(\dd \rho).
\label{eqn:MIHM_cond_measure}
\end{equation}
Thus, to exclude PMS,  $\tilde{w}(\dd \rho)$ has to be proportional to $\rho\  \tilde{\lambda}(\dd \rho)$ and therefore the decay speed of $K$ should be  higher than that of $G$, whatever their functional form (as soon as they are completely monotone). Besides, we can make the two  following comments. 

First, by dividing both sides of equation~\eqref{eqn:MIHM_cond_measure} by $\rho$, integrating on $(0,+\infty)$ and using Fubini's theorem, one gets the necessary (but not sufficient) martingale price condition
\begin{align}
1-\nu \ &= \ (\ios-\ioc) \int_0^\infty \frac{\tilde{w}(\dd \rho)}\rho = (\ios-\ioc) \int_0^\infty \left( \int_0^\infty \exp(-\rho t) \ \dd t \right) \tilde{w}(\dd \rho)
\nonumber \\
&= (\ios-\ioc) \int_0^\infty \left( \int_0^\infty \exp(-\rho t) \ \tilde{w}(\dd \rho) \right) \dd t
\nonumber \\
&=  (\ios-\ioc) \int_0^\infty K(t) \ \dd t
\ =: \ \text{DBR}.
\label{eqn:DBR_transient}
\end{align}
This equation means that the proportion of transient impact should be equal to the \textit{directional branching ratio}, which we define as
\begin{equation}
\text{DBR} \ = \ (\ios-\ioc) \int_0^\infty K(t) \ \dd t \ = \ \frac{\ios-\ioc}{\ios+\ioc} \times \text{BR},
\label{eqn:DBR_BR}
\end{equation}
where $\text{BR}$ is the usual branching ratio for Hawkes-based models that count positively price changes of both signs (see for instance Hardiman and Bouchaud~\cite{HB_Branching}). This result is intuitive since the DBR represents the average number of \enquote{children trades of the same sign} for each trade, which, to obtain a diffusive price process, should be equal to the proportion of price impact that vanishes over time. Although it is only a necessary condition, equation~\eqref{eqn:DBR_transient} gives a quite general numerical criterion to assess empirically whether an observed price process is compatible with the martingale property, or rather persistent ($\text{DBR}>1-\nu$) or mean-reverting ($\text{DBR}<1-\nu$).

Second, the power-law kernels
\begin{equation}
G(u) = \nu + (1-\nu) (1+c_G \times t)^{-a},
\quad
K(u) = (1+c_K \times t)^{-(1+\epsilon)}
\nonumber
\end{equation}
are particular cases of~\eqref{comp_monotone_kernels}, with
\begin{equation}
\tilde{\lambda}(\dd \rho) = \frac{\rho^{a-1} \exp(-\rho/c_G)}{\Gamma(a) \ c_G^a} \ \dd \rho,
\quad
\tilde{w}(\dd \rho) = \frac{\rho^\epsilon \exp(-\rho/c_K)}{\Gamma(1+\epsilon) \ c_K^{1+\epsilon}} \ \dd \rho.
\nonumber
\end{equation}
Equation~\eqref{eqn:MIHM_cond_measure} then yields
\begin{equation}
a=\epsilon, \quad c_G = c_K = c, \quad \frac{\ios-\ioc}{\epsilon c} = 1-\nu.
\nonumber
\end{equation}
Let us recall that if $K$ is a power-law, one must have $\epsilon>0$ to obtain integrability, which is a necessary condition for the Hawkes process to be stationary. Also, in that case, the process can only have long-memory (i.e. non-integrable auto-covariance) if the Hawkes norm is equal to one\footnote{We refer to Hardiman et al.~\cite{HBB} for a test of this property on market data, and to Jaisson and Rosenbaum~\cite{JaissonRosenbaum} for a study of Hawkes processes with an Hawkes norm close to one.} and if $\epsilon \in (0,1/2)$, see Br\'emaud and Massouli\'e, Theorem 1 in~\cite{BMH}. In that case, the auto-covariance decays asymptotically as $t^{-(1-2\epsilon)}$. We thus reach exactly the same conclusion as Bouchaud et al.~\cite{BGPW}, who give the diffusive price condition $\beta = (1-\gamma)/2$, where $\gamma$ is the decay exponent of the auto-correlation of trade signs, and $\beta = a$ is the decay exponent of the propagator. Note that we used a totally different approach (absence of Price Manipulation Strategies), and that equation~\eqref{eqn:MIHM_cond_measure} is a possible generalization of their result to a wider class of kernels, within the Hawkes framework.

The calibration results presented in Section~\ref{section:calib_results} allow us to confront real stock data to the martingale price condition obtained above. In particular, it is easy to check whether the proportion of transient impact $1-\nu = \sum \lambda_i$ is smaller, equal or greater than the directional branching ratio $\text{DBR}$.
 Although we do not expect the condition to be exactly satisfied in practice, we find it interesting to evaluate how much (and which way) real data deviate from the theoretical equilibrium.

\subsection{The optimal execution strategy}\label{section:opt_strat}

In~\cite{AB_DynHawkes}, we obtained an explicit characterization of the optimal execution strategy that minimzes $\E[C(X)]$ among strategies such that $X_0 \in \R$ and $X_{T+}=0$ when $G(t)=e^{-\rho t}$ and $K(t)=e^{-\beta t}$. It is of interest to generalize this result to multi-exponential kernels~\eqref{expo_mixture}. This is in principle possible. In fact, the model is still Markovian and Affine with respect to the state variable  $(X_t,P_t,S_t,D^i_t,{\kappa^\pm}^{(i)}_t)$, and the cost is still quadratic. As in~\cite{AB_DynHawkes}, one should first guess the quadratic form of the cost function, then derive necessary conditions on its coefficients, and last run a verification argument.  However, we know from Alfonsi and Schied~\cite{AS_SICON} that the optimal strategy without the flow of trades (i.e. $N\equiv 0$) is already quite involved and is characterized through a matrix Riccati equation. In our context, the system of ordinary differential equations that characterize the cost function would be much more intricated, and one would presumably have to solve it with numerical methods, which are less efficient than closed formulas for high-frequency trading. However, in the particular case where the propagator is kept exponential
\begin{equation}\label{expo_mixture_forK}
G(u) = \nu +  (1-\nu) \exp(-\rho u), \quad K(u) = \sum_{i=1}^p w_i \exp(-\beta_i u),
\end{equation}
with $0<\beta_1<\dots<\beta_p$ and $w_1,\dots,w_p>0$, it is still possible to derive explicitly the optimal execution strategy. In fact, we can handle the same arguments as in~\cite{AB_DynHawkes} and obtain the following result, proved in Appendix~\ref{appendix:proof_opt_strat}. 
\begin{theorem}\label{thm_opt_strat}Let $\alpha_i=w_i(\ios-\ioc)$ and H, the  square matrix of order~$p$ defined by
\begin{equation}\label{def_H}
1\le i,j \le p, \  H_{i,j} \ = \ \indi{i=j}\beta_i -\alpha_j.
\end{equation}
%
We also define the two continuous matrix functions $\zeta, \omega$ by\footnote{When $M$ is invertible, $\zeta(M)=M^{-1}[I_p-\exp(-M)]$ and $\omega(M)=M^{-2}[\exp(-M)-I_p+M]$. } 
\begin{equation}\label{eqn:def_omega_zeta} \zeta(M)=\sum_{k\ge 0} (-1)^k \frac{M^k}{(k+1)!} \text{ and } \omega(M)=\sum_{k\ge 0} (-1)^k \frac{M^k}{(k+2)!}.
\end{equation}
Then, the strategy~$X^*$ that minimizes the expected cost $\E[C(X)]$ satisfies a.s. and $dt$-a.e on $(0,T)$,
\begin{align}
(1-\epsibis)X^*_t=&-[1+\rho(T-t)]D_t+ \frac{m_1}{2\rho}[2+\rho(T-t)]  \label{opt_strat} \\
&\times \delta_t\tp \left\{ I_p
+ \frac{\rho (T-t)}{2+\rho (T-t)} \times [\zeta((T-t) H)+ \nu \rho (T-t) \ \omega((T-t) H)]
\right\} \ . \ (1,\cdots,1)\tp, \nonumber
\end{align}
where $\delta_t^i = {\kappa_t^+}^{(i)}-{\kappa_t^-}^{(i)}$ for $i \in \{1,\cdots,p\}$ are intensity imbalances. Moreover, the optimal strategy is fully characterized by equation~\eqref{opt_strat}. 
\end{theorem}
Though restricted to~\eqref{expo_mixture_forK}, we believe that this extension of the result of~\cite{AB_DynHawkes} may be relevant for applications. In fact, on our dataset, there is not much gain to use the multi-exponential price propagator rather than the mono-exponential one, see Figure~\ref{fig:LR_propag_BNPP}. Instead, for the decay kernel of the intensity, considering an exponential mixture allows to produce a richer variety of autocovariance functions, see Figure~\ref{fig:Int_BNPP}.

\section{Calibration method}\label{sec_cal}

\subsection{Description of the dataset}

We consider tick-by-tick data provided by the French investment bank Natixis, to which we are grateful. The data contains all the changes in prices and volumes of the best bid and best ask, for two actively traded French stocks: BNP Paribas and Total.

The data is selected between 11a.m. and 1p.m., for every trading day between January and September 2012 and 2013. We exclude the three last months of the year, where activity decreases on average, along with the months where the tick size deviates from $0.005$ euros. The two-hour window around noon is chosen to obtain a rather stable and uniform behavior of market activity, see e.g. Lehalle and Laruelle~\cite{LehalleLaruelle}, p.~112.
This way, for each stock separately and with minimal data treatment, we can reasonably assume that each two-hour window of trading is a realization of the same random price process.

In the initial dataset, for each stock separately, each line corresponds either to an update in price and/or quantity at one of the best queues (triggered by a market event such as a market order, a limit order or a cancellation), or a new trade executed  for a given volume at a given price. The time stamps for these updates are precise to the millisecond. We reduce this data by aggregating the events happening on the same millisecond: we only keep track of the best prices at the beginning and at the end of each time stamp, which yields the aggregated price impact of the events that happened \enquote{simultaneously}, i.e. on the same millisecond. Similarly, we sum all the volumes that were executed on the same time stamp. We obtain a simplified sequence of market events, among which a minority is associated to a traded volume and/or to a price change.

A correspondence should be clarified between the theoretical items of the models of~\cite{AB_DynHawkes} and Section~\ref{sec_model}, and actual financial data. Different possibilities may be relevant, but our choices are the following:

\begin{itemize}
\item We define the \enquote{market price} $P_t$ as the midpoint price, i.e. the average of the best bid price and the best ask price at any time $t$.

\item We only consider time stamps where the midpoint price jumps. In other words, we ignore the trades and cancellations that do not empty either the best bid or the best ask, as well as the passive limit orders that do not define a new best price. For the stocks that we consider, this gives an average latency of one to four seconds between two consecutive time stamps. This is in agreement with the time scale that is thought of in the theoretical model of~\cite{AB_DynHawkes}, which is not of ultra-high frequency.

\item We express the time in hours, and note $T = 2$ the length of the window that we consider for each trading day. Throughout the paper, we note $\tau \in (0,T)$ the time stamps which correspond to midpoint jumps triggered by trades, i.e. by limit orders that cross the spread or by market orders. These correspond to the jumps of the process $N$ of the theoretical model: they are marked by both a price jump $\Delta \midprice_\tau$ (of one or several half-ticks), and an executed volume $\Delta V_\tau>0$ expressed in number of shares. The time stamps of other price jumps  are noted $\theta \in (0,T)$. They are triggered by cancellations and passive limit orders, with no executed volume, and they are assumed to enforce on average the deterministic resilience effect as in~\cite{BGPW}. Between two trades, the deviation of the price from this deterministic average is considered as a noise process, modeled using an arithmetic Brownian motion.

\end{itemize}

Some key statistics for these items are given in Table~\ref{table:BNPP_stats_2012-2013} for BNP Paribas and Total.

\begin{table}[!h]
\center
\begin{tabular}{|c||c|c||c|c|}\hline
	\text{Stock}				& 	\multicolumn{2}{c||}{BNP Paribas}	& 	\multicolumn{2}{c|}{Total}		\\ \hline
	\text{Year}				& 	2012			& 	2013		& 	2012		& 	2013		\\ \hline
	Average midprice			&  	$32.4$		& 	$44.9$ 	&  	$38.2$	& 	$39.0$ 	\\
	Tick size				&  	$0.005$		& 	$0.005$ 	&  	$0.005$	& 	$0.005$ 	\\
	Number of mid. changes per hour& 	$1909$ 		& 	$1699$ 	& 	$1209$ 	& 	$929$ 	\\
	Proportion due to transactions	& 	$10.0\%$ 		& 	$7.9\%$ 	& 	$7.6\%$ 	& 	$6.9\%$  	\\ 
	$m_1$				& 	$776$ 		& 	$636$		& 	$978$ 	& 	$963$		\\
	$m_2/m_1^2$			& 	$3.38$ 		& 	$4.69$ 	& 	$4.30$ 	& 	$6.72$  	\\ 
	Average size of the first queue	&  	$1398$		& 	$1136$ 	&  	$1710$	& 	$1779$	\\ \hline
\end{tabular}
\caption{Table of statistics for the stocks BNP Paribas and Total on the periods January-September 2012-2013, between 11 a.m. and 1 p.m. January 2012 is excluded for BNP Paribas because the tick size dropped below $0.005$. We give the proportion of midpoint changes which are triggered by trades, the remaining proportion being triggered by cancellations or passive limit orders. $m_1$ is the average volume of transactions that trigger price moves, and $m_2$ is the average squared volume for these transactions. The greater the ratio $m_2/m_1^2$, the more variance in the distribution of traded volumes.}
\label{table:BNPP_stats_2012-2013}
\end{table}

\subsection{Overview of the calibration process}

One specificity of the price model given by equation~\eqref{Propag_model} is that it is composed of two separate components:
\begin{itemize}
\item The point process $N$ for the trades that trigger the price moves, for the which time stamps $\tau$ and the marks (the price jumps  $\Delta \midprice_\tau$ and the executed volumes $\Delta V_\tau$) are modeled and estimated jointly,
\item The propagator model, which conditionally to the midpoint jumps due to trades, is a continuous-time linear regression model with a Gaussian noise process $\sigma W_t$.
\end{itemize}
Therefore, the trades are modeled using marked Hawkes processes, and conditionally to them, the price is Gaussian. This segmentation has at least three advantages. First, the calibration process is simpler since the two parts can be estimated independently, which significantly reduces the dimension of the problem. Second, the estimation results on each side are robust to the choices made in the other. For instance, if one wants to modify the Hawkes modeling for the trades, then our results for the propagator are still valid, and vice versa. Eventually, the results of Section~\ref{subsec_noarb} include some theoretical links between the Hawkes parameters and the propagator, and it seems more rigorous to confront these links to our calibration results when the two parts are estimated independently.

Our calibration protocol as a whole being somewhat sophisticated, we test its validity and robustness by running it on simulated data. In Sections~\ref{section:calib_results} and~\ref{section:eval_model}, we give the results of our analysis for these simulations as well as for real financial data.

\subsection{Estimation of the propagator}\label{section:estim_propag}

\subsubsection{Framework}\label{section:calib_res_framework}

In this section we explain how the propagator model introduced in Section~\ref{sec_model} can be adapted for practical applications, in particular for its calibration. This requires to consider the two following points:
\begin{itemize}
\item In practice, the price impact of transactions is not proportional to their volumes. It is typically of a few ticks, while the volumes span a wider range of values. Therefore, one must choose between \enquote{price resilience} and \enquote{volume resilience} as in Alfonsi et al.~\cite{AFS}. The first choice corresponds to modeling the mean-reversion property of market prices, the second describes how liquidity \enquote{regenerates} after a trade, and the two are only equivalent for linear price impact.
\item The evolution of the price between two transactions is very noisy, and the propagator model only explains a part of its variance. Therefore, we need to control the variance of the estimation to obtain satisfying calibration results.
\end{itemize}

For the first point, we choose to model price resilience, which is easier to measure in practice and has been considered more often in the literature. This boils down to replacing $\Delta N_\tau/q$ by the midprice jumps $\Delta \midprice_\tau$ in equation~\eqref{Propag_model}.
For the second point, an intuitive possibility consists in restraining the propagator regression to a finite time window $\regwin>0$, and to assume that the model predicts the price increment $P_t-P_{t-\regwin}$ for $t \geq \regwin$ instead of $P_t-P_0$. If the noise is an additive Brownian term $\sigma W_t$, this fixes the variance of the predicted variables to $\sigma^2 \regwin$ instead of $\sigma^2 t, \ t \in [\regwin,T]$. We obtain the modified price model
\begin{equation}
P_t \ = \  P_{t-\regwin} 
\ + \
\underset{t-\regwin < \tau \leq t}{\sum} \Delta \midprice_\tau \ G(t-\tau)
\ + \ \sigma (W_t-W_{t-\regwin}).
\label{eqn:practical_price_model}
\end{equation}
Of course, $\regwin$ must be such that $G(\regwin)-G(\infty)$ is small compared to $G(\infty)$ for the model to be a meaningful approximation of the original model~\eqref{Propag_model}, see Remark~\ref{rem_mkv}. This condition also allows to avoid bias in the estimation of the propagator $G$. We fix $\regwin=0.5$ hours (30 minutes) throughout the sequel of this paper, basing ourselves on preliminary observations that we do not detail here. Note that within the range $\regwin \in [0.1,1]$, the choice of this parameter has little impact on the results. One can verify \textit{a posteriori} that our estimations of $G$ are compatible with $G(0.5)-G(\infty) \ll G(\infty)$.

The predicted price increment between $t-\regwin$ and $t$ is given by
\begin{equation}
\hat P_t - P_{t-\regwin} \ = \ 
\underset{t-\regwin \leq \tau \leq t}{\sum} \Delta \midprice_\tau \ G(t-\tau)
\label{eq:estimated_price_model}
\end{equation}
where $P_{t-\regwin}$ is the real midpoint price at time $t-\regwin$, taken directly from the data.
Equation~\eqref{eqn:practical_price_model} becomes
\begin{equation}
P_t \ = \ \hat P_t + \sigma (W_t - W_{t-\regwin}).
\label{eqn:def_regr_model}
\end{equation}
Conditionally to $P_{t-\regwin}$ and to the process $\midprice$, one has $P_t \sim \mathcal{N}(\hat P_t, \sigma^2 \regwin)$. In this setting, the Maximum Likelihood Estimator of $G$ is equivalent to the Least Squares Estimator. We thus minimize numerically on the parameters of $G$ the quadratic error
\begin{equation}
\mathcal{E}(G) = \underset{\regwin<\theta<T}{\sum} [\hat P_\theta(G) - P_\theta]^2,
\label{eqn:quad_error}
\end{equation}
where the $\theta$'s are the occurrences of price jumps due to cancellations or passive limit orders.
To get a better understanding of the shape of the propagator, we first estimate $G$ in an \enquote{unconstrained} manner, i.e. as the linear interpolation of a discrete set of points.
Thus, we model $G$ as
\begin{equation}
G(t) = g_l \mathbf{1}_{[t_l,\regwin[}(t) + \overset{l-1}{\underset{i=0}{\sum}} \frac{(t_{i+1}-t) g_i + (t-t_i)g_{i+1}}{t_{i+1}-t_i} 
 \ \mathbf{1}_{[t_i,t_{i+1}[}(t),
\nonumber
\end{equation}
where $t_1,\cdots,t_l$ are fixed a priori and $(g_1,\cdots,g_l)$ is the parameter to estimate. We see that the resulting curve, which is given is Section~\ref{section:calib_results} for stock data, has an increasing short-range part, and switches to a decreasing mode after a few seconds. One has $G(0)=1$, but $G$ reaches a point above unity before it enters its decreasing regime. Let us recall that in an idealized model without bid-ask spread, Alfonsi et al.~\cite{ASS}  and Gatheral et al.~\cite{GSS} show that $G$ has to be decreasing and convex around zero to exclude PMS and some market instability. This is not the case on our dataset.  We interpret this as the fact that after a trade, the new bid-ask is generally formed around the impacted price.  Thus, during a few seconds, limit orders and cancellations tend to impact the midprice in the same direction as the trade.
This motivates us to distinguish the propagator $G(t)$ and the functional form of its long-range decay that we call the resilience, noted $R(t)$. This way, we can allow $R(0)\geq1$ and impose that $R$ is decreasing. One can then link $G$ and $R$ with a simple linear interpolation between $t=0$ and $t=\adjlag$, with $\adjlag>0$ the \enquote{adjustment lag} of the market
\begin{equation}
G(t) = \left[1+(R(\adjlag)-1) \ \frac{t}{\adjlag}\right] \indi{t\leq \adjlag} + R(t) \indi{t>\adjlag}.
\nonumber
\end{equation}
This choice has the merit that once $\adjlag$ is fixed, only the resilience curve needs to be estimated since $G$ is characterized by $R$.
Therefore, to estimate $R$ with an imposed decreasing functional form, we place ourselves in the following version of the price model
\begin{equation}
P_t = P_{t-\regwin}
+ \underset{t-\regwin \leq \tau < t-\adjlag}{\sum} \Delta \midprice_\tau \ R(t-\tau)
+ \underset{t-\adjlag \leq \tau \leq t}{\sum} \Delta \midprice_\tau \ \left[1+(R(\adjlag)-1) \ \frac{t-\tau}{\adjlag}\right]
+ \sigma (W_t-W_{t-\regwin}).
\nonumber
\end{equation}
We consider two types of parameterization for the resilience $R(t)$:
\begin{itemize}
\item The mono-exponential curve
\begin{equation}
R(t) = \resfact \ [1-\lambda(1-\exp(-\rho t))],
\label{eqn:monoexpo_resil}
\end{equation}
with three parameters $\resfact,\rho>0, \lambda \in [0,1]$. $\resfact$ is an amplification factor, $\rho$ is the resilience speed of the market, $\lambda$ is the transient part of the price impact of trades, and $\nu = 1-\lambda$ is the permanent part.
The mono-exponential curve is the type of resilience considered in the theoretical model of~\cite{AB_DynHawkes}.
\item The multi-exponential curve
\begin{equation}
R(t) = \resfact \ \left[1-\overset{p}{\underset{i=1}{\sum}}\lambda_i(1-\exp(-\rho_i t)) \right],
\label{eqn:multiexpo_resil}
\end{equation}
is a generalization of the previous one, with $2p+1$ parameters $\resfact,\rho_1,\cdots,\rho_p>0, \lambda_1,\cdots,\lambda_p \in [0,1]$, 
$\sum_i \lambda_i \leq 1$. For $1\leq i \leq p$, $\lambda_i$ is the proportion of transient impact that decays at speed $\rho_i$, and $\nu = 1- \sum_i\lambda_i$ is the proportion of permanent impact.
\end{itemize} 
For both parameterizations, we estimate a posteriori the volatility $\sigma$ of the Brownian noise with
\begin{equation}
\hat\sigma =
\sqrt{
\frac{\overset{n}{\underset{i=1}{\sum}} \left[P^i_T -P^i_0 - \underset{0 < \tau < T}{\sum} \Delta \midprice^i_\tau \ G(T-\tau)\right]^2}
{n\times T}
},
\label{eqn:def_sigma_noise}
\end{equation}
%
where $n$ is the number of days of the sample, and for $i \in \{1,\cdots,n\}$, $P^i$ and $\midprice^i$ are respectively
the real price and the midprice jump process for day $i$, and the $\tau$'s are the jump times of $M^i$.
Also, since the prediction model defined by~\eqref{eq:estimated_price_model} and~\eqref{eqn:def_regr_model} can be seen as a continuous-time linear regression, where the explained variables are the price increments $P_t - P_{t-\regwin}$ and the regressors are the past price jumps $\Delta \midprice_\tau$ triggered by trades, we can evaluate its quality of fit using a usual analysis of variance.
We define the $r^2$ value as
\begin{equation}
r^2 \ = \ 1 \ - \
\frac
{\overset{n}{\underset{i=1}{\sum}} \ 
\underset{\regwin<\theta<T}{\sum} [\hat P^i_\theta - P^i_\theta]^2}
{\overset{n}{\underset{i=1}{\sum}} \ 
 \ \underset{\regwin<\theta<T}{\sum}
\left[ P^i_\theta - P^i_{\theta-\regwin} -\overline{\Delta P} \right]^2},
\label{eqn:def_r2}
\end{equation}
where for $i \in \{1,\cdots,n\}$, $\hat P^i$ is
the predicted price process for day $i$, and
\begin{equation}
\overline{\Delta P}= \frac1{\sum_{i=1}^n \# \theta_i} \ \overset{n}{\underset{i=1}{\sum}} \ 
 \ \underset{\regwin<\theta<T}{\sum} \left( P^i_\theta - P^i_{\theta-\regwin} \right)
\nonumber
\end{equation}
 is the average price move between $\theta-\regwin$ and $\theta$, where the $\theta$'s are the times of price changes with no executed volumes and
and $\# \theta_i$ is the number of such price changes on day $i$.
Note that since there is no constant in the regression model~\eqref{eqn:def_regr_model}, the $r^2$ could theoretically be negative, but this is not the case in practice. The $r^2$ constitutes a useful comparison criterion between different estimated propagators, and we use it in Section~\ref{section:calib_results}.

Now that the global practical framework is set, the estimation protocol for $G$ needs to be detailed. This is the object of the following section.

\subsubsection{Estimation protocol}\label{section:estim_protocol_resil}

We use a multi-step estimation protocol, that mainly resorts to the minimization of the quadratic error $\mathcal{E}$ defined in~\eqref{eqn:quad_error}. When $G(t)$ is linear with respect to its parameters, $\mathcal{E}$ is quadratic and one step of Newton-Raphson's algorithm is enough to find the minimum (see Appendix~\ref{appendix:NR_propagator}). When the dependency in the parameters is non-linear, we first use grid minimizations to find a suitable starting point for the algorithm.

As a first step, we estimate the \enquote{unconstrained} propagator curve. Then, we estimate the resilience curve using the two parameterizations presented in Section~\ref{section:calib_res_framework}.

\textbf{Estimation of the unconstrained propagator curve}

We first estimate $G$ by the linear interpolation $\hat G$
\begin{equation}
\hat G(t) = g_l \mathbf{1}_{[t_l,T[}(t) + \overset{l-1}{\underset{i=0}{\sum}} \frac{(t_{i+1}-t) g_i + (t-t_i)g_{i+1}}{t_{i+1}-t_i} 
 \ \mathbf{1}_{[t_i,t_{i+1}[}(t).
\nonumber
\end{equation}
For $t_1,\cdots,t_l$ fixed a priori, $\hat G$ is linear with respect to $(g_1,\cdots,g_l)$. Thus, one step of Newton-Raphson's method (see Appendix~\ref{appendix:calib_UC}) determines the parameters that minimize the quadratic error $\mathcal{E}(\hat G)$.
To approximate the long-range propagator, we choose a uniform grid $t_i = i/l$ with $l=20$ on the interval $[0,0.2]$. On the other hand, for a zoom on the beginning of the curve, we concentrate the $t_i$'s near zero.

\textbf{Estimation of the multi-exponential resilience curve}

The simultaneous estimation of multiple $\rho_i$'s being too unstable, we choose to fix four components associated to four simple characteristic time scales (the $\rho_i$'s are expressed in inverse hours): $\rho_1 = 6$ (ten minutes), $\rho_2 = 60$ (one minute), $\rho_3 = 120$ (thirty seconds) and $\rho_4=360$ (ten seconds). We then assume that the vector $(\rho_1,\cdots,\rho_4)$ is rich enough to represent all the relevant time scales in our framework, and we focus on the weights $\lambda_1,\cdots,\lambda_4$ associated to each scale to characterize the decay of the curve.
The multi-exponential resilience given by equation~\eqref{eqn:multiexpo_resil} becomes
\begin{equation}
R(t) = \overline \nu + \overset{4}{\underset{i=1}{\sum}}\overline \lambda_i \exp(-\rho_i t),
\nonumber
\end{equation}
where we re-parameterize $\overline \nu = \gamma (1-\sum_{i=1}^4 \lambda_i)>0$ and $\overline \lambda_i = \gamma \lambda_i>0$. Reciprocally, one has $\resfact = \overline \nu + \sum_{i=1}^4 \overline \lambda_i$ and $\lambda_i = \overline \lambda_i / \resfact$. 
Since the $\rho_i$'s are fixed, the resilience curve $R(t)$ is linear w.r.t. the parameter $(\overline \nu, \overline \lambda_1, \cdots, \overline \lambda_4)$ that remains to be estimated, thus Newton-Raphson's algorithm (see Appendix~\ref{appendix:calib_ME}) converges with a single iteration. We then select the significant $\rho_i$'s as follows:

\begin{enumerate}

\item A first estimation yields a \enquote{full} parameter $(\overline \nu, \overline \lambda_1, \cdots, \overline \lambda_4)$. Some of the resulting $\overline \lambda_i$'s may be non-positive, which is incompatible with the model.

\item While there exists $i$ such that $\overline \lambda_i \leq 0$, we remove the $\rho_i$ corresponding to the minimal $\overline \lambda_i$, and we launch the algorithm again with one less parameter.

\item Eventually, we have selected one to four \enquote{significant} $\rho_i$'s, of associated weights $\overline \lambda_i$'s that are positive, and the estimation process is complete.

\end{enumerate}

Since each of these steps only take one iteration of Newton-Raphson's algorithm, the whole estimation protocol for the multi-exponential curve is quite fast. Therefore, in order to estimate the market adjustment lag $\adjlag$, we can conduct the estimation several times for $\adjlag$ on some discrete grid, and compare the regression $r^2$'s as defined by~\eqref{eqn:def_r2}. The result associated to the maximal $r^2$ gives the parameters $\resfact_\text{multi}$, $\lambda_\text{multi}$ and $\rho_\text{multi}$ for the multi-exponential resilience, along with the adjustment lag $\adjlag$.

\textbf{Estimation of the mono-exponential resilience curve}

The multi-exponential estimation presented above serves as a starting point for the following. The market adjustment lag $\adjlag$ is already estimated, along with the associated set of parameters $\resfact_\text{multi}, \lambda_\text{multi}, \rho_\text{multi}$ for the multi-exponential resilience curve. We set
\begin{equation}
\resfact = \resfact_\text{multi},
\quad
\lambda = \sum_i \lambda_\text{multi}^i,
\quad
\rho = \sum_i \frac{\lambda_\text{multi}^i}{\lambda} \rho_\text{multi}^i
\label{eqn:mono_param_start}
\end{equation}
as a starting parameter for the mono-exponential estimation.
As in the multi-exponential case, we re-parameterize~\eqref{eqn:monoexpo_resil} as
\begin{equation}
R(t) = \overline \nu  + \overline \lambda \exp(-\rho t),
\nonumber
\end{equation}
with $\overline \nu = \gamma (1-\lambda)>0$ and $\overline \lambda = \gamma \lambda>0$. We then proceed as follows
\begin{enumerate}

\item We use Newton-Raphson's algorithm to minimize the quadratic error on the whole parameter $(\overline \nu,\overline \lambda,\rho)$ (see Appendix~\ref{appendix:calib_ME_gene} for $p=1$ exponential component). If the starting point is convex and the algorithm converges to a satisfying level, we proceed directly to Step 6. Else, we go to Step 2.

\item Keeping $\rho$ fixed to its starting value~\eqref{eqn:mono_param_start}, the dependency of $R(t)$ on $\overline \nu$ and $\overline \lambda$ is linear. Thus, with one step of Newton-Raphson's algorithm, we get the optimal values of $\overline \nu$ and $\overline \lambda$ for the current value of $\rho$.

\item For $\resfact = \overline{\nu}+\overline{\lambda}$ fixed by Step 2, $\lambda$ initialized to $\overline \lambda / \resfact$ and $\rho$ as in~\eqref{eqn:mono_param_start},  we minimize the quadratic error $(\lambda,\rho) \mapsto \mathcal{E}(\lambda,\rho)$ on a local two-dimensional grid in the vicinity of the starting point. 

\item The pair that reaches the minimum of the error grid at Step 3 is again used as a starting point to Newton-Raphson's algorithm, to determine the optimal $(\lambda,\rho)$ for the current fixed value of $\resfact$, using the \enquote{unit} mono-exponential parameterization of Appendix~\ref{appendix:calib_ME_unit}.  We actualize $(\lambda,\rho)$ to this optimum, along with $\overline \nu = \gamma (1-\lambda)$ and  $\overline \lambda = \gamma \lambda$.

\item The parameter $(\overline \nu,\overline \lambda,\rho)$ is now in a region where the quadratic error is more likely to be convex. Therefore, we use this new starting point for an error minimization using Newton-Raphson's algorithm on the whole parameter.

\item We obtain the parameter $\resfact_\text{mono}, \lambda_\text{mono}, \rho_\text{mono}$ for the mono-exponential resilience curve.

\end{enumerate}

The above estimation protocol for the mono-exponential resilience curve may seem complicated: in particular, it is more subtle than the multi-exponential estimation. The reason for this is that we want here to determine the most significant characteristic time scale of the resilience through the parameter $\rho$. The dependency of the quadratic error $\mathcal{E}$ on this parameter being non-linear, nothing guarantees a priori that Newton-Raphson's algorithm (or more simply a gradient algorithm) has a convex starting point, which is a necessary condition to ensure its convergence. Hence we have to proceed more carefully and introduce several intermediary steps.

\subsection{Estimation of the Hawkes parameters}\label{section:calib_method_int}

\subsubsection{Framework}

Independently of the propagator, we also estimate the parameters of the Hawkes-based model presented in Section~\ref{sec_model} for the price jumps due to transactions. 
We choose the self-excitation functions $\phis$ and $\phic$ to be affine, i.e.
\begin{equation}
\phis(x) = \phicos^0 + \phicos^1 x
\quad , \quad
\phic(x) = \phicoc^0 + \phicoc^1 x.
\label{eqn:affine_exc}
\end{equation}
In the standard Hawkes framework, self-excitation in the order flow is not marked, i.e. only the constant terms $\phicos^0, \phicoc^0$ appear in $\phis$ and $\phic$. In spite of its simplicity, the affine structure allows us to underline the deviation from the standard Hawkes benchmark, and to detect an increasing part in the self-excitation function.

As pointed out in Section~\ref{section:calib_res_framework}, there are two possible interpretations for the marks associated to the jumps of $N$. Since each of these jumps corresponds to a price jump due to a transaction, they are all associated to two positive variables: the price impact on the one hand, and the traded volume on the other hand. Therefore, we estimate three sets of parameters for different versions of the Hawkes model (unit marks, volume marks, and price marks), each with a different practical interpretation of the intensity jump terms. Precisely, we replace $\varphi_{\textup{s}/\textup{c}}(\dd N_t^\pm)$ in~\eqref{def_kpi} and~\eqref{def_kmi} at the jump times~$t$ by either of the three possibilities
\begin{align}
\phi^0_\text{s/c,unit}, \  \phi^0_\text{s/c,vol.} + \phi^1_\text{s/c,vol.} |\Delta V_t| / m_1, \ \phi^0_\text{s/c,price} + \phi^1_\text{s/c,price} |\Delta M_t| / \overline m,
\label{three_phis}
\end{align}
where $m_1$ is the average executed volume and $\overline m$ is the average price impact.

\subsubsection{Estimation protocol}

Our estimation protocol for the Hawkes part of the model is then as follows: we first estimate the mono-exponential Hawkes model $K(u) = \exp(-\beta u)$, which allows us to estimate the Hawkes norm and its repartition in terms of self and cross-excitation, and to select the optimal mark type for the jumps. Then we estimate the multi-exponential Hawkes model $K(u) = \sum_{i=1}^p w_i \exp(-\beta_i u)$ with the $\beta_i$'s fixed a priori.

\textbf{Mono-exponential kernel}

Let us consider the mono-exponential Hawkes model of equation~\eqref{kappa_expo}, for which the Hawkes decay kernel is $K(u) = \exp(-\beta u), \ \beta>0$.
 We first focus on the parameters of the total intensity $\Sigma_t = \kappa^+_t + \kappa^-_t$ by aggregating all the price jumps due to trades, regardless of their signs. In the mono-exponential case, one has
\begin{equation}
\textup{d}\Sigma_t  = \ - \beta \ (\Sigma_t - 2\kappa_{\infty}) \ \textup{d}t
\ + \ \iota \ \dd J_t,
\nonumber
\end{equation}
where $\iota$ is the average excitation, so that the jumps of $J$ have an average of one.
We use a Generalized Method of Moments (GMM) to estimate $\beta, \kappa_\infty$ and $\iota$.
We divide the time window $[0,T]$ of length $T=2$ hours in $720$ bins of length $h=1/360$ (i.e. ten seconds). Then, we compute the number $\Delta \tilde{J}^i_l$ of price jumps due to trades in the time bin $[(l-1)h,lh], \ l \in\{1,\cdots,\lfloor T/h \rfloor\}$ on day $i \in \{ 1,\cdots, n\}$, for each time bin and each day. If $l$ is the row index and $i$ is the column index, we obtain a $\lfloor T/h \rfloor\times n$ matrix of which the entries are the positive numbers $\Delta \tilde{J}^i_l$. We normalize this dataset by dividing each column by its mean value and multiplying the whole matrix by the original global mean value, so that the global mean is unchanged and each column has the same mean.
We first compute the empirical mean $\overline{\Delta \tilde{J}}$ and variance $\mathcal{V}$ of the discrete process $\Delta \tilde{J}$
\begin{equation}
\overline{\Delta \tilde{J}} = \frac1{n \times \lfloor T/h \rfloor} 
\overset{n}{\underset{i=1}{\sum}} \overset{\lfloor T/h \rfloor}{\underset{l=1}{\sum}} \Delta \tilde{J}^i_l,
\quad
\mathcal{V} = \frac1{n \times \lfloor T/h\rfloor -1} \
\overset{n}{\underset{i=1}{\sum}} \overset{\lfloor T/h \rfloor}{\underset{l=1}{\sum}} 
\left[ \Delta \tilde{J}^i_l - \overline{\Delta \tilde{J}} \right]^2.
\nonumber
\end{equation}
The average jump intensity $2\overline{\kappa}$ of the total jump process is obtained with the formula
$2\overline{\kappa} = \overline{\Delta \tilde{J}}/h$.
Besides, the empirical auto-correlation function of $\Delta \tilde{J}$ is given by
\begin{equation}
\forall k \in \{ 1,\cdots, k_\text{max}\},
\quad
\widehat{\cov}(k) = 
\frac1{\mathcal{V}} \times
 \left\{ 
\frac1{n \times (\lfloor T/h \rfloor-k)} \left[ 
\overset{n}{\underset{i=1}{\sum}} \ \overset{\lfloor T/h \rfloor}{\underset{l=k+1}{\sum}}
\Delta \tilde{J}^i_l \ \Delta \tilde{J}^i_{l-k} \right]
- \overline{\Delta \tilde{J}}^2
 \right\},
\label{eqn:emp_autocorrel}
\end{equation}
where $k_\text{max} = 36$ is the maximum lag (so that the maximum range $k_\text{max} h = 0.1$ equals six minutes). 
Using the results of Da Fonseca and Zaatour~\cite{DFZ} for mono-exponential Hawkes processes, we have that $\widehat{\cov}(k)$ decays as $\exp(-(\beta-\iota) k)$. Therefore, the exponential fit of the empirical curve $\widehat{\cov}(k)$ yields an estimate of $d := \beta-\iota$. Then, we also get from~\cite{DFZ}
\begin{equation}
\mathcal{V} = 2\overline{\kappa} \ \left\{ \frac {\beta^2 h} {d^2} + \left(1-\frac{\beta^2}{d^2}\right) \frac{1-\exp(-d h)}{d}\right\}.
\nonumber
\end{equation}
This relation can be inverted to obtain an estimate for $\beta$: if we note $z_h = (1-\exp(-d h))/d$, we get
\begin{equation}
\beta \ = \ d \ \sqrt{\frac{\mathcal{V}/(2\overline{\kappa}) - z_h}{h - z_h}}.
\nonumber
\end{equation}
Then, $\iota = \beta - d$ and $\kappa_\infty = (1-\iota/\beta) \ \overline{\kappa}$ can be deduced from the above equation. We also obtain the mono-exponential branching ratio
\begin{equation}
\text{BR}_\text{mono} = \iota/\beta.
\nonumber
\end{equation}
 Keeping $\beta, \iota$ and $\kappa_\infty$ fixed to these GMM estimates, we now turn to the bi-dimensional intensity model~\eqref{kappa_expo}. We use Maximum Likelihood Estimation (see Appendix~\ref{section:calib_intensity}) on one-dimensional grids to determine the self and cross-excitation parameters:

\begin{enumerate}

\item We determine the proportion $u \in [0,1]$ such that $\ios = u \ \iota, \ \ioc = (1-u) \ \iota$ maximize the likelihood of the two-dimensional intensity $(\kappa^+,\kappa^-)$, where $\ios$ and $\ioc$ are respectively the average self-excitation and cross-excitation parameters.

\item For volume marks and price marks separately, we determine the proportion $u_\text{s} \in [0,1]$ such that 
$\phicos^0 = u_\text{s} \ \ios, \ \phicos^1 = (1-u_\text{s}) \ \ios$ maximize the likelihood of $(\kappa^+,\kappa^-)$, where $\phicos^0, \phicos^1$ are defined in equation~\eqref{eqn:affine_exc}. Similarly, we determine the optimal proportion $u_\text{c}$ for 
$\phicoc^0 = u_\text{c} \ \ioc, \ \phicoc^1 = (1-u_\text{c}) \ \ioc$. For $\ios$ and $\ioc$ fixed, we obtain the optimal constant and linear parts for self and cross-excitation, for the two possible types of marks.

\item The likelihoods obtained for the three models are then compared to determine which of the unit / volumes / price marks yield the best model.

\end{enumerate}

Eventually, we obtain estimates for all the parameters $\beta_\text{mono}, {\kappa_\infty}_\text{mono}, {\phicos}^0_\text{mono}, {\phicos}^1_\text{mono}, {\phicoc^0}_\text{mono}, {\phicoc}^1_\text{mono}$ of the mono-exponential Hawkes model, along with the optimal type of marks.

\textbf{Multi-exponential kernel}

We turn to the multi-exponential Hawkes model $K(u) = \sum_{i=1}^p w_i \exp(-\beta_i u)$.
As in the case of the estimation of the multi-exponential resilience in Section~\ref{section:estim_protocol_resil}, we fix four $\beta_i$'s associated to four simple characteristic time scales. In fact, we choose the same time scales as for the resilience:
$\beta_1=6,\beta_2=60,\beta_3=120$ and $\beta_4=360$. We then calibrate the $w_i$'s associated to each $\beta_i$, and these weights tune the shape of the Hawkes kernel.

The results of the mono-exponential estimation are used to select the type of marks (unit, volume or price) and to get a starting point for $\kappa_\infty, \phicos^0, \phicos^1, \phicoc^0,\phicoc^1$ and the branching ratio $\text{BR}$. The starting point for the $w_i$'s is chosen to be uniformly distributed
\begin{equation}
w_i = \frac{\text{BR}_\text{mono}}{\sum_{i=1}^4 \frac{1/4}{\beta_i}} \times \frac14,
\nonumber
\end{equation}
with a scaling that matches the initial branching ratio. Then, we maximize the likelihood of the model on the parameter $(\kappa_\infty,w_1,w_2,w_3,w_4)$ using Newton-Raphson's algorithm, as explained in Appendix~\ref{section:calib_intensity}. We use the same selection method as for the multi-exponential resilience estimation of Section~\ref{section:estim_protocol_resil}: if at least one of the $w_i$'s is non-positive, we delete the $\beta_i$ associated to the minimal $w_i$ and launch the algorithm again, with one less parameter. Finally, we multiply $(\phicos^0, \phicos^1, \phicoc^0,\phicoc^1)$ by the sum of the remaining $w_i$'s, and scale the latter to one. Without changing the overall model, this imposes $K(0)=1$ for the Hawkes decay kernel $K$. We obtain the parameters $\beta_\text{multi}, w_\text{multi}, {\kappa_\infty}_\text{multi}, {\phicos}^0_\text{multi}, {\phicos}^1_\text{multi}, {\phicoc^0}_\text{multi}, {\phicoc}^1_\text{multi}$ for the multi-exponential Hawkes model.

\section{Calibration results}\label{section:calib_results}

\subsection{Description of the results}\label{section:description_results}

This section is dedicated to the presentation of our calibration results. The calibration method of Section~\ref{sec_cal} is first applied to simulated data to test its validity, and then to actual financial data from French stocks. We also provide some qualitative comments.
For each simulated dataset and each stock, the results are summarized in tables, plus a few graphs for BNP Paribas. The content of the tables is explained below.

\textbf{Adjustment lag table:}
This table gives the regressions $r^2$'s of the multi-exponential resilience curve, for several values of the market adjustment lag $\adjlag$. It is used to select the optimal value of $\adjlag$ on a discrete grid.

\textbf{Resilience table:}
The resilience table gives the estimation results for the propagator. We give the selected adjustment lag $\adjlag$ and the estimated parameters for the two types of resilience curve
\begin{align}
R_\text{mono}(t) &= \resfact_\text{mono} \ [1-\lambda_\text{mono}(1-\exp(-\rho_\text{mono} t))],
\nonumber \\
R_\text{multi}(t) &= \resfact_\text{multi} \ \left[1-\sum \lambda_\text{multi}^j(1-\exp(-\rho_\text{multi}^j t))\right],
\nonumber
\end{align}
along with the estimated volatility $\sigma$ of the noise and the regression $r^2$, defined respectively by equations~\eqref{eqn:def_sigma_noise} and~\eqref{eqn:def_r2}.

\textbf{Marks table:}
In this table, we give the maximized log-likelihoods per point $\mathcal{L}_\text{unit}$, $\mathcal{L}_\text{vol.}$ and $\mathcal{L}_\text{price}$ for each type of mark (unit, volumes and price jumps), in the mono-exponential Hawkes model. It serves as a selection criterion for the optimal type of mark.

\textbf{Intensity table:}
This table gives the estimated parameters for the Hawkes model described in Section~\ref{section:markov_spec}, for both the mono-exponentiel decay kernel $K(u) = \exp(-\beta u)$ and the multi-exponential one $K(u) = \sum_{i=1}^p w_i \exp(-\beta_i u)$. We also give the maximized log-likelihoods per point $\mathcal{L}_\text{mono}$ and $\mathcal{L}_\text{multi}$, which can be compared to one another or between datasets to quantify the quality of fit of the Hawkes model. Eventually, we give the branching ratio $\text{BR}$ and the directional branching ratio $\text{DBR}$ defined by equation~\eqref{eqn:DBR_BR}, that are obtained with the multi-exponential parameterization.

\subsection{Simulated data}\label{section:simu_calib}

We first give in Tables~\ref{table:SIMU_A1_resil_int} and~\ref{table:SIMU_A2_resil_int} the results of our calibration protocol on two datasets simulated with the price model~\eqref{Propag_model}. In each table, the first column gives the \enquote{real} simulation parameters and the second gives the estimated ones. Both datasets are composed of $150$ independent realizations of the price process on two-hour windows, and we choose simulation parameters close to what is found further for stock data in order to obtain relevant benchmarks. Note that Simulation 1 features a non-zero Brownian volatility, whereas Simulation 2 is generated by the \enquote{pure} propagator model without noise.

\begin{table}[h]
\center
\parbox{.40\linewidth}{
\begin{tabular}{|c||c|c|}\hline
	\text{Year}					& 	\text{Simu.}		& 	\text{Calib.}		\\ \hline
	$\adjlag \text{ (sec)}$			&  	$4$			& 	$4$	 		\\ \hline
	$\resfact_\text{multi}$			&  	$2.70$		& 	$2.35$		\\
	$\rho_\text{multi}$				& 	$60/360$ 		& 	$6/60/360$	 	\\
	$\lambda_\text{multi}$			& 	$0.50/0.10$ 		& 	$0.13/0.35/0.11$ 	\\
	$\nu_\text{multi}$				& 	$0.40$ 		& 	$0.41$		 \\
	$\sigma_\text{multi}$			& 	$0.1000$ 		& 	$0.1917$ 		\\
	$r^2_\text{multi}$				& 	$-$			& 	$9.554\%$	 	\\ \hline
	$\resfact_\text{mono}$			&  	$-$			& 	$2.38$ 		\\
	$\rho_\text{mono}$			& 	$-$ 			& 	$68.2$ 		\\
	$\lambda_\text{mono}$			& 	$-$ 			& 	$0.55$		\\
	$\sigma_\text{mono}$			& 	$-$ 			& 	$0.1923$ 		\\
	$r^2_\text{mono}$			& 	$-$	 		& 	$9.519\%$ 		\\ \hline
\end{tabular}
}
\quad
\parbox{.45\linewidth}{
\begin{tabular}{|c||c|c|}\hline
	\text{Year}				& 	\text{Simu.}			& 	\text{Calib.}				\\ \hline	
	\text{Marks type}			& 	\text{Volume}		& 	\text{Volume}			\\ \hline	
	$\beta_\text{multi}$		&  	$60/360$			& 	$60/360$ 			\\	
	$w_\text{multi}$			&  	$0.100/0.900$		& 	$0.102/0.898$ 		\\	
	${\kappa_\infty}_\text{multi}$	&  	$15.0$			& 	$15.2$ 				\\ 
	${\phicos}_\text{multi} $		& 	$110.5/19.5$		& 	$109.8/20.9$ 			\\ 
	${\phicoc}_\text{multi} $		& 	$66.5/3.5$	 		& 	$59.7/9.7$ 				\\
	$\mathcal{L}_\text{multi}$	& 	$-$	 			& 	$3.1659$ 				\\ \hline
	$\beta_\text{mono}$		&  	$-$				& 	$153.0$ 				\\		
	${\kappa_\infty}_\text{mono}$	&  	$-$				& 	$16.6$ 				\\
	${\phicos}_\text{mono}$		& 	$-$	 			& 	$68.7/13.1$ 				\\ 
	${\phicoc}_\text{mono}$		& 	$-$		 		& 	$37.4/6.1$ 				\\ 
	$\mathcal{L}_\text{mono}$	& 	$-$ 				& 	$3.1560$ 				\\ \hline
	\text{BR}				& 	$0.833$ 			& 	$0.839$ 				\\
	\text{DBR}				& 	$0.250$ 			& 	$0.257$ 				\\ \hline
\end{tabular}
}
\caption{Calibration of the resilience (left) and intensity (right) for Simulation 1. For the $\phi$'s, the first entry is the constant term and the second one is the linear term.}
\label{table:SIMU_A1_resil_int}
\end{table}

\begin{table}[h]
\center
\parbox{.40\linewidth}{
\begin{tabular}{|c||c|c|}\hline
	\text{Year}					& 	\text{Simu.}		& 	\text{Calib.}		\\ \hline
	$\adjlag \text{ (sec)}$			&  	$2$			& 	$2$	 		\\ \hline
	$\resfact_\text{multi}$			&  	$-$			& 	$3.05$		\\
	$\rho_\text{multi}$				& 	$-$ 			& 	$6/120$	 	\\
	$\lambda_\text{multi}$			& 	$-$ 			& 	$0.0005/0.6850$ 	\\
	$\nu_\text{multi}$				& 	$-$	 		& 	$0.31$		 \\
	$\sigma_\text{multi}$			& 	$-$ 			& 	$0.0055$ 		\\
	$r^2_\text{multi}$				& 	$-$			& 	$96.92\%$	 	\\ \hline
	$\resfact_\text{mono}$			&  	$3.20$		& 	$3.06$ 		\\
	$\rho_\text{mono}$			& 	$130$ 		& 	$121.3$ 		\\
	$\lambda_\text{mono}$			& 	$0.70$ 		& 	$0.69$		\\
	$\sigma_\text{mono}$			& 	$0.0000$ 		& 	$0.0055$ 		\\
	$r^2_\text{mono}$			& 	$-$	 		& 	$96.92\%$ 		\\ \hline
\end{tabular}
}
\quad
\parbox{.45\linewidth}{
\begin{tabular}{|c||c|c|}\hline
	\text{Year}				& 	\text{Simu.}			& 	\text{Calib.}				\\ \hline	
	\text{Marks type}			& 	\text{Volume}		& 	\text{Volume}			\\ \hline	
	$\beta_\text{multi}$		&  	$120/360$			& 	$6/120/360$ 			\\	
	$w_\text{multi}$			&  	$0.050/0.950$		& 	$0.0007/0.0505/0.9488$ 		\\	
	${\kappa_\infty}_\text{multi}$	&  	$40.0$			& 	$39.1$ 				\\ 
	${\phicos}_\text{multi} $		& 	$84.0/36.0$			& 	$72.8/40.9$ 			\\ 
	${\phicoc}_\text{multi} $		& 	$45.0/5.0$	 		& 	$47.4/7.7$ 				\\
	$\mathcal{L}_\text{multi}$	& 	$-$	 			& 	$2.7218$ 				\\ \hline
	$\beta_\text{mono}$		&  	$-$				& 	$82.2$ 				\\		
	${\kappa_\infty}_\text{mono}$	&  	$-$				& 	$19.3$ 				\\
	${\phicos}_\text{mono}$		& 	$-$	 			& 	$27.3/15.4$ 				\\ 
	${\phicoc}_\text{mono}$		& 	$-$		 		& 	$17.8/2.9$ 				\\ 
	$\mathcal{L}_\text{mono}$	& 	$-$ 				& 	$2.6740$ 				\\ \hline
	\text{BR}				& 	$0.519$ 			& 	$0.535$ 				\\
	\text{DBR}				& 	$0.214$ 			& 	$0.186$ 				\\ \hline
\end{tabular}
}
\caption{Calibration of the resilience (left) and intensity (right) for Simulation 2. For the $\phi$'s, the first entry is the constant term and the second one is the linear term.}
\label{table:SIMU_A2_resil_int}
\end{table}

Overall, the accuracy of the estimation is satisfying. The estimated Hawkes parameters are very close to the real ones, although the dimensionality is high. Importantly, the branching ratios and directional branching ratios are all determined accurately, within a precision of $\pm 0.03$ on our experiments. Concerning the propagator, the results are more noisy for Simulation 1, which is not surprising since it includes some Brownian noise. Still, the proportion of transient impact is nearly exact and the dominant time scale is well determined. Simulation 2 is generated with a mono-exponential propagator, and the resilience speed $\rho_\text{mono}$ is slightly underestimated; however this parameter is less stable than the $\lambda$'s and the accuracy that we obtain seems reasonable. In this second case, the values that we find for the volatility and the regression $r^2$ are satisfyingly close to $0$ and $100\%$ respectively.

\subsection{BNP Paribas}\label{section:calib_results_BNPP}

Tables~\ref{table:BNPP_reactime}, \ref{table:BNPP_marks_type} and \ref{table:BNPP_resil_int_2012-2013} and Figures~\ref{fig:LR_propag_BNPP}, \ref{fig:zoom_propag_BNPP} and \ref{fig:Int_BNPP} present our estimation results for the French stock BNP Paribas on the periods February-September 2012 and January-September 2013.

\begin{table}[h]
\center
\begin{tabular}{|c||c|c|c|c|}\hline
	$\adjlag \text{ (sec)}$	 	& 	0			&	2		&	4		&	6		\\
	$r^2_\text{multi} (2012)$	  	&	$24.572\%$		& 	$24.675\%$	&	$24.677\%$	&	$24.672\%$ \\ 
	$r^2_\text{multi} (2013)$	  	&	$10.607\%$		& 	$10.674\%$	&	$10.668\%$	&	$10.649\%$ \\ \hline
\end{tabular}
\caption{Regression $r^2$ for the multi-exponential resilience curve, evaluated for several market adjustment lags $\adjlag = 0,2,4,6$ seconds, for the stock BNP Paribas.}
\label{table:BNPP_reactime}
\end{table}

\begin{table}[h]
\center
\begin{tabular}{|c||c|c|c|}\hline
	\text{Marks type}				& 	\text{Unit}		&	\text{Volume}	&	\text{Price jump}		\\
	$\mathcal{L}_\text{mono} (2012)$	&	$2.6804$		& 	$2.6826$		&	$2.6791$			\\ 
	$\mathcal{L}_\text{mono} (2013)$ 	&	$2.5772$		& 	$2.5794$		&	$2.5750$			\\ \hline
\end{tabular}
\caption{Log-likelihood per point for the mono-exponential Hawkes model, evaluated for several types of marks: unit, volumes and price jumps (see eq.~\eqref{three_phis}), for the stock BNP Paribas.}
\label{table:BNPP_marks_type}
\end{table}

\begin{table}[h]
\center
\parbox{.45\linewidth}{
\begin{tabular}{|c||c|c|}\hline
	\text{Year}					& 	2012			& 	2013			\\ \hline
	$\adjlag \text{ (sec)}$			&  	$4$			& 	$2$	 		\\ \hline
	$\resfact_\text{multi}$			&  	$2.69$		& 	$2.99$		\\
	$\rho_\text{multi}$				& 	$60$ 			& 	$60/360$	 	\\
	$\lambda_\text{multi}$			& 	$0.61$ 		& 	$0.30/0.53$ 		\\
	$\nu_\text{multi}$				& 	$0.39$ 		& 	$0.17$		 \\
	$\sigma_\text{multi}$			& 	$0.2253$ 		& 	$0.2153$ 		\\
	$r^2_\text{multi}$				& 	$24.677\%$		& 	$10.674\%$ 	\\ \hline
	$\resfact_\text{mono}$			&  	$2.70$		& 	$2.56$ 		\\
	$\rho_\text{mono}$			& 	$60.8$ 		& 	$116.5$ 		\\
	$\lambda_\text{mono}$			& 	$0.62$ 		& 	$0.80$ 		\\
	$\sigma_\text{mono}$			& 	$0.2253$ 		& 	$0.2153$ 		\\
	$r^2_\text{mono}$			& 	$24.678\%$ 	& 	$10.688\%$ 	\\ \hline
\end{tabular}
}
\quad
\parbox{.45\linewidth}{
\begin{tabular}{|c||c|c|}\hline
	\text{Year}				& 	2012				& 	2013					\\ \hline	
	\text{Marks type}			& 	\text{Volume}		& 	\text{Volume}			\\ \hline	
	$\beta_\text{multi}$		&  	$6/360$			& 	$6/360$	 			\\	
	$w_\text{multi}$			&  	$0.010/0.990$		& 	$0.011/0.989$	 		\\	
	${\kappa_\infty}_\text{multi}$	&  	$15.1$			& 	$12.1$				\\ 
	${\phicos}_\text{multi} $		& 	$112.8/18.4$		& 	$115.4/15.7$ 			\\ 
	${\phicoc}_\text{multi} $		& 	$50.4/2.1$	 		& 	$46.4/0.9$ 				\\
	$\mathcal{L}_\text{multi}$	& 	$2.7720$ 			& 	$2.6708$ 				\\ \hline
	$\beta_\text{mono}$		&  	$73.0$			& 	$114.1$ 				\\		
	${\kappa_\infty}_\text{mono}$	&  	$13.9$			& 	$14.0$ 				\\
	${\phicos}_\text{mono}$		& 	$38.3/6.2$ 			& 	$58.5/8.0$ 				\\ 
	${\phicoc}_\text{mono}$		& 	$17.1/0.7$		 	& 	$23.5/0.5$ 				\\ 
	$\mathcal{L}_\text{mono}$	& 	$2.6826$ 			& 	$2.5794$ 				\\ \hline
	\text{BR}				& 	$0.820$ 			& 	$0.810$ 				\\
	\text{DBR}				& 	$0.351$ 			& 	$0.380$ 				\\ \hline
\end{tabular}
}
\caption{Calibration of the resilience (left) and intensity (right) for the stock BNP Paribas for the periods February-September 2012 and January-September 2013, between 11 a.m. and 1 p.m. For the $\phi$'s, the first entry is the constant term and the second one is the linear term.}
\label{table:BNPP_resil_int_2012-2013}
\end{table}

\begin{figure}[h]
	\centering
\subfigure[2012]{
      \includegraphics[height=6cm,width=6cm]{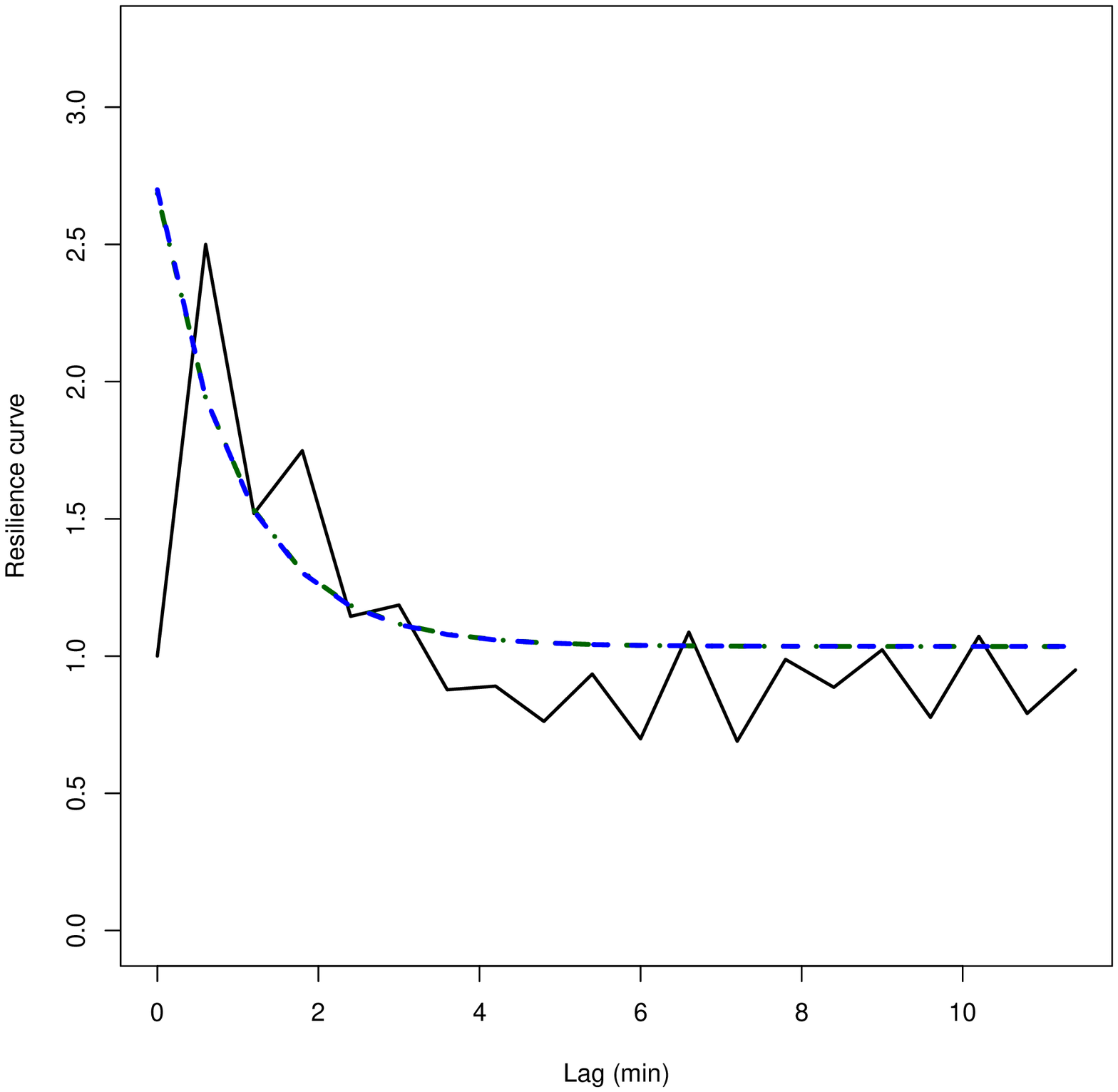}
	\label{fig:LR_propag_BNPP_2012}
    }
\subfigure[2013]{
    \includegraphics[height=6cm,width=6cm]{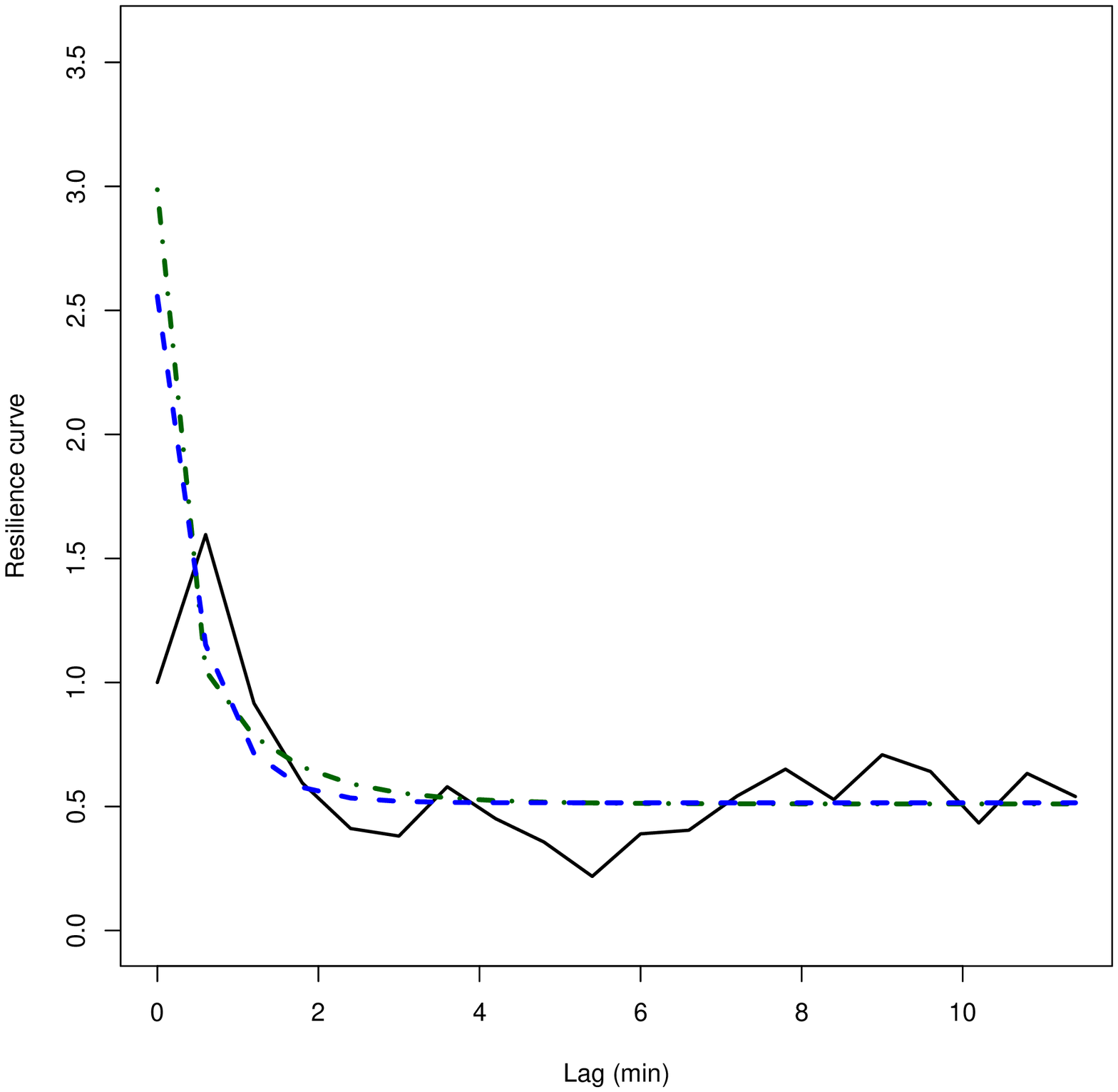}
	\label{fig:LR_propag_BNPP_2013}
    }
\caption{The estimated propagator for BNP Paribas. The plain line is the unconstrained propagator, the (blue) dashed line is the mono-exponential resilience curve, and the (green) dot-dashed line is the multi-exponential resilience curve.}
    \label{fig:LR_propag_BNPP}
\end{figure}

\begin{figure}[h]
	\centering
\subfigure[2012]{
      \includegraphics[height=6cm,width=6cm]{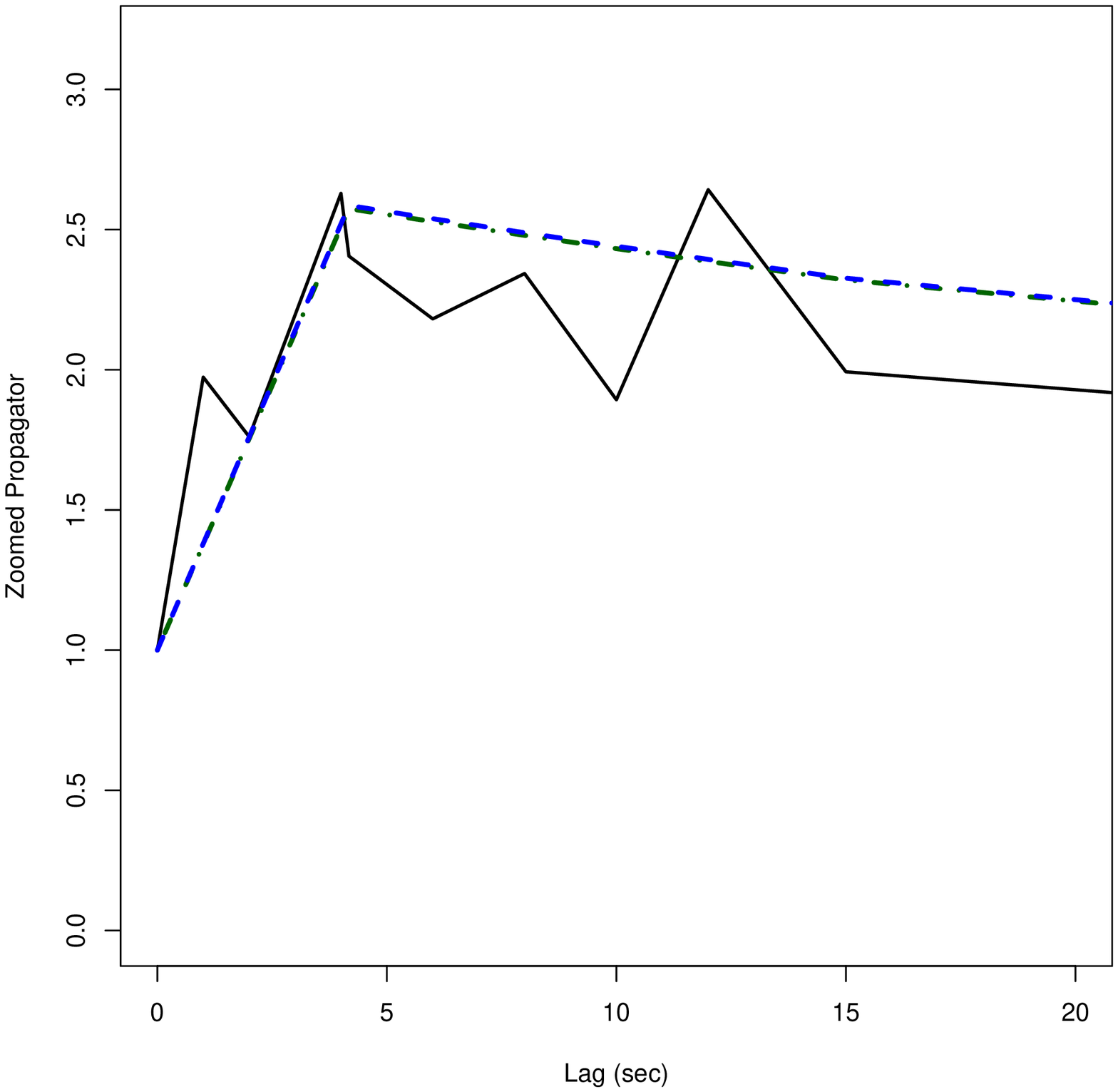}
	\label{fig:zoom_propag_BNPP_2012}
    }
\subfigure[2013]{
    \includegraphics[height=6cm,width=6cm]{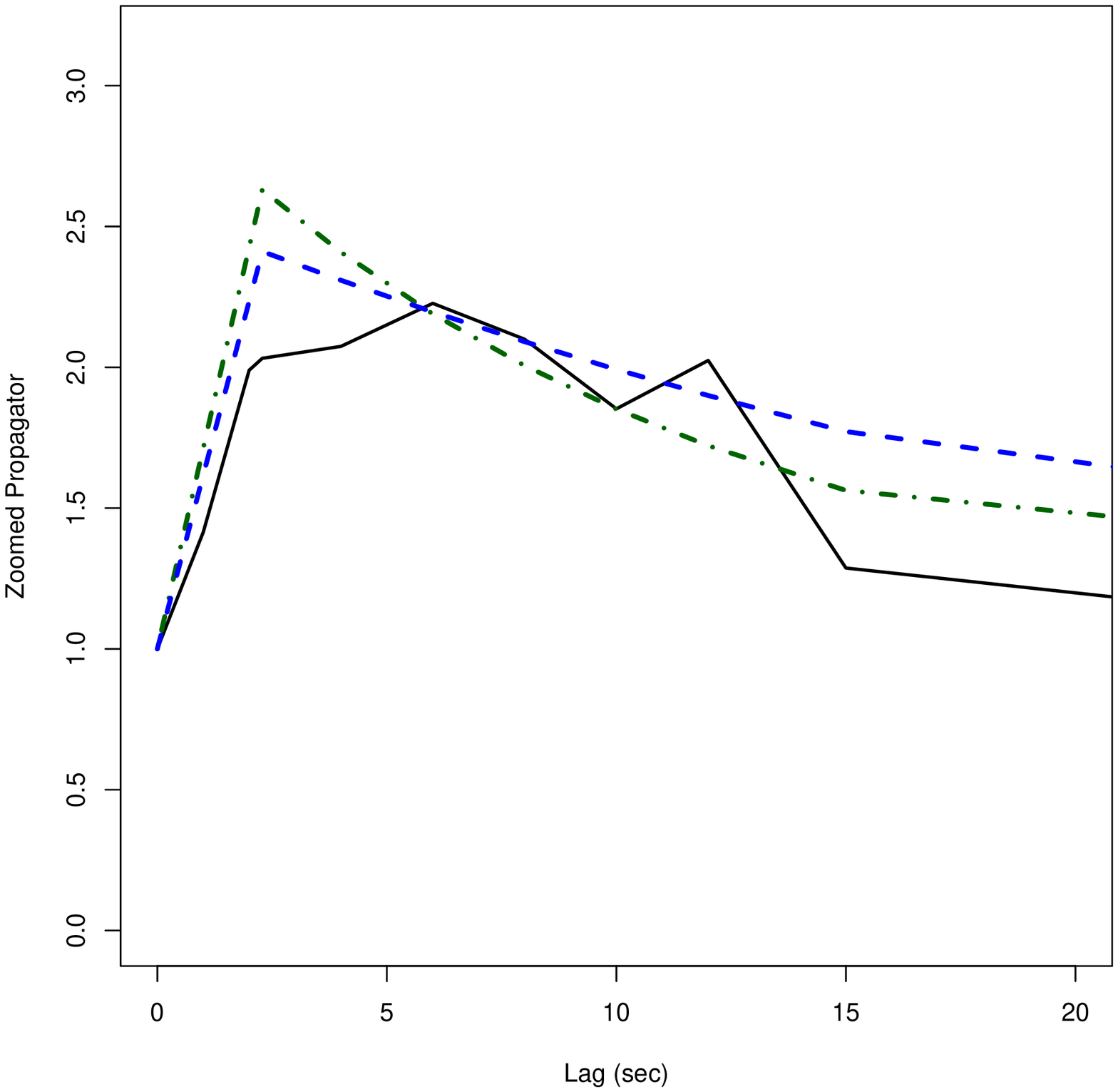}
	\label{fig:zoom_propag_BNPP_2013}
    }
\caption{Zoom on the first twenty seconds of the propagator curve for BNP Paribas. The plain line is the unconstrained curve, the (blue) dashed line is the mono-exponential curve, and the (green) dot-dashed line is the multi-exponential curve. The propagator is increasing during a few seconds, before the resilience effect kicks in.}
    \label{fig:zoom_propag_BNPP}
\end{figure}

\begin{figure}[h]
	\centering
\subfigure[2012]{
      \includegraphics[height=6cm,width=6cm]{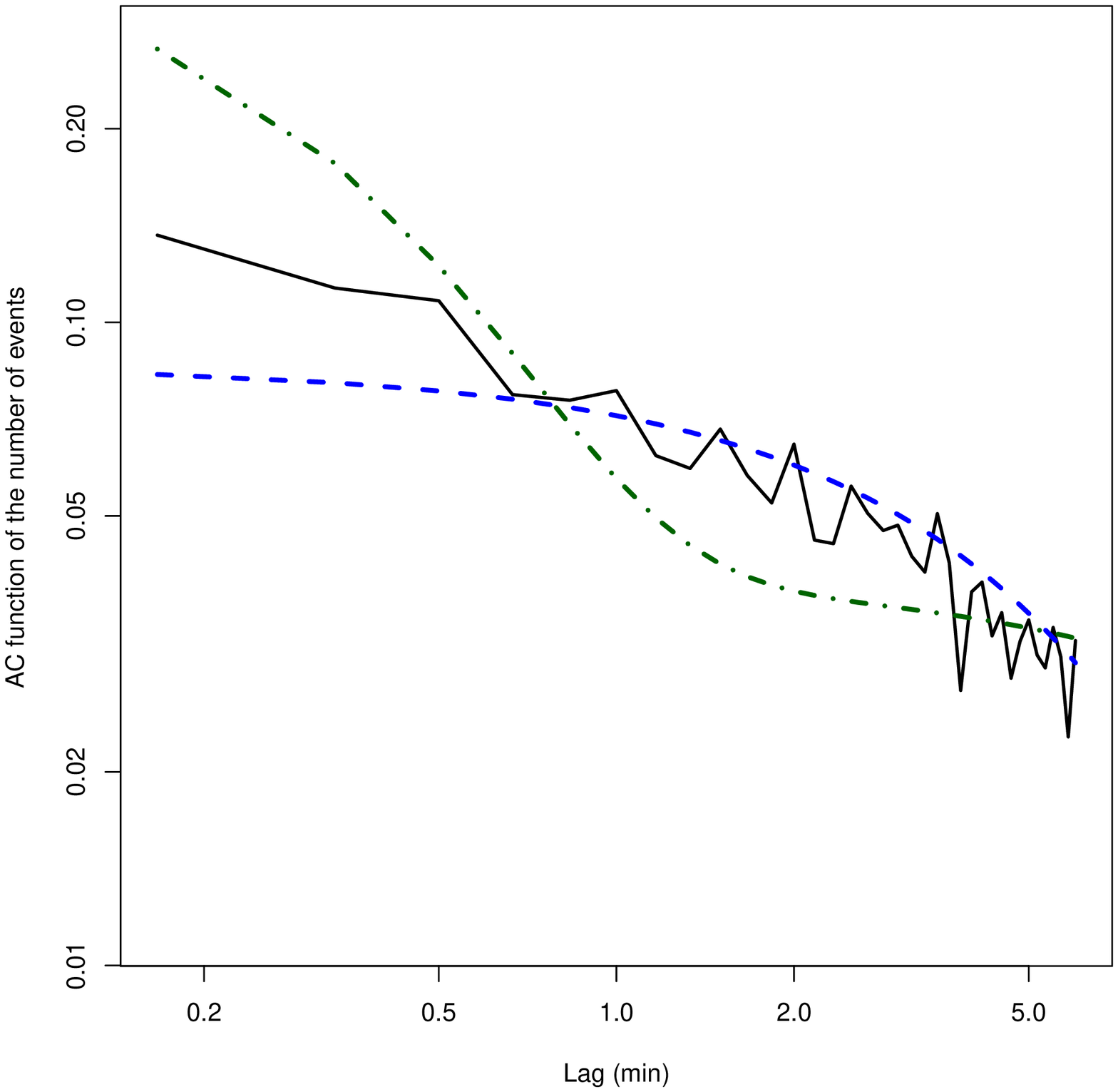}
	\label{fig:Int_BNPP_2012}
    }
\subfigure[2013]{
    \includegraphics[height=6cm,width=6cm]{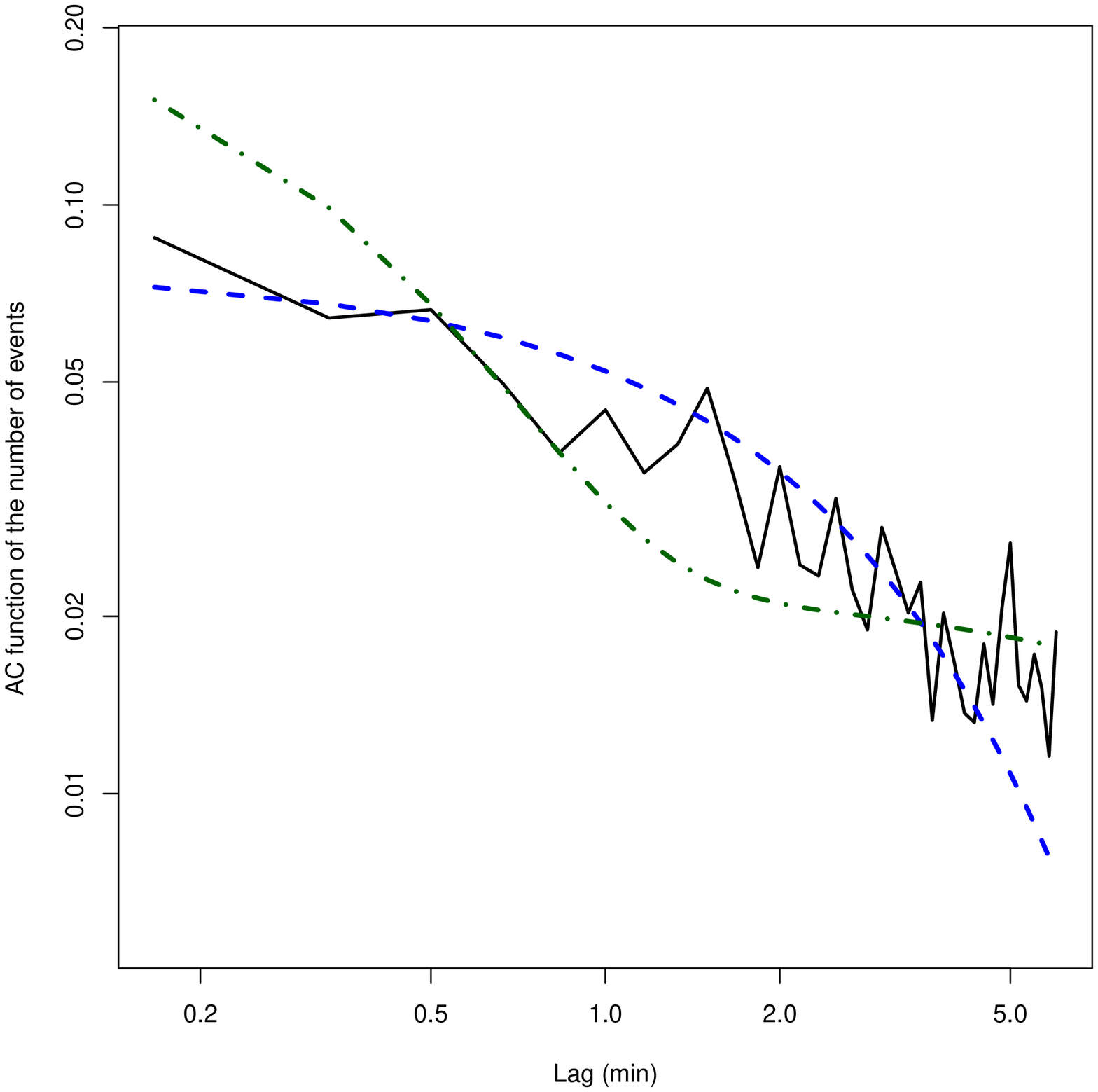}
	\label{fig:Int_BNPP_2013}
    }
	\caption{Auto-correlation function of the number of midpoint moves triggered by trades (plain line), in log-log scale, for BNP Paribas. The (blue) dashed line is the auto-correlation generated by the mono-exponential Hawkes model, the (green) dot-dashed line is generated by the multi-exponential Hawkes model.}
    \label{fig:Int_BNPP}
\end{figure}

Let us first look at the estimation results for the propagator. Table~\ref{table:BNPP_reactime} and Figure~\ref{fig:zoom_propag_BNPP} show that the adjustment lag $\adjlag$ defined in Section~\ref{section:calib_res_framework} is positive and thus that the propagator is increasing near zero. The estimation yields $\adjlag = 4$ sec. for 2012 and $\adjlag=2$ sec. for 2013, and the increasing part lasts indeed longer on Figure~\ref{fig:zoom_propag_BNPP_2012} than on Figure~\ref{fig:zoom_propag_BNPP_2013}. The parameter $\gamma$ given in Table~\ref{table:BNPP_resil_int_2012-2013} tunes the maximum value reached by the propagator at the end of the increasing phase. We find a result between two and three. This means that on average, after a large trade, not only does the bid-ask close around the impacted price (which would yield\footnote{To be more precise, let us consider for example a buy order that increases the ask of one tick. Then, the midprice jumps of one half tick. If the bid price follows shortly the ask and increases of one tick, this moves again the mid of one half tick upward, which gives $\gamma=2$.} $\resfact\approx2$), but cancellations at the new best queue also push the price in the same direction as the trade.

After its short increasing part, the propagator switches to its resilience mode described by Table~\ref{table:BNPP_resil_int_2012-2013} and Figure~\ref{fig:LR_propag_BNPP}. The unconstrained resilience curve is quite smooth, and one can observe on Figure~\ref{fig:LR_propag_BNPP} that it decays to a non-zero proportion of permanent impact ($\approx 40\%$ for 2012 and $\approx 20\%$ for 2013). Also, the results given in Table~\ref{table:BNPP_resil_int_2012-2013} indicate that the mono-exponential fit for the resilience is good on this dataset. For 2012, only the speed $\rho=60$ (i.e. a characteristic time scale of one minute) is selected in the multi-exponential estimation. On the other hand, for 2013, there are two selected speeds (corresponding to one minute and ten seconds), but the mono-exponential fit with $\rho_\text{mono} =116.5$ (approximately thirty seconds) yields a higher regression $r^2$. These dominant characteristic time scales motivate the use of the particular case considered for the optimal strategy in Section~\ref{section:opt_strat}.

We now focus on the estimation results for the Hawkes parameters. Table~\ref{table:BNPP_marks_type} justifies the selection of volume marks: indeed, they yield a higher likelihood per point than unit marks and price marks. Unit marks are the benchmark model for Hawkes processes, but they fail to reproduce the fact that large orders trigger more activity on the market. Indeed, we see on Table~\ref{table:BNPP_resil_int_2012-2013} that the self-excitation parameter $\phicos$ and the cross-excitation parameter $\phicoc$ have non-negligible linear parts ($10-15\%$ for self-excitation and $2-5\%$ for cross-excitation). As for price marks, we think that they give less information than volume marks since price jumps cluster on a few values (one or two ticks in most cases), while the distribution of volumes is much wider.

Hawkes parameters seem to be quite stable, especially in the multi-exponential case where the estimation results are very similar for 2012 and 2013. Two decay speeds are selected for the intensity, and these are the two extreme ones: the long range $\beta = 6$ (10 minutes) and the short range $\beta=360$ (10 seconds). The importance of each time scale $\beta_i$ can be measured by the proportion of the norm that it accounts for, given by
\begin{equation}
\frac{w_i/\beta_i}{\sum_j w_j/\beta_j}.
\nonumber
\end{equation}
Here, the long-range component $\beta=6$ accounts for $\approx 40\%$ of the norm, and the short-range one $\beta=360$ for the remaining $60\%$. Therefore, both decay speeds are important, which is also reflected by the significant increase from the log-likelihood per point $\mathcal{L}_\text{mono}$ of the mono-exponential model to $\mathcal{L}_\text{multi}$ for the multi-exponential one. One can deduce that contrary to the propagator, the Hawkes kernel includes at least two exponential components.

Figure~\ref{fig:Int_BNPP} gives a visual comparison between the data, the mono-exponential Hawkes model and the multi-exponential one through the auto-correlation of the number of events. The formula for the empirical auto-correlation function $\widehat{\cov}(k)$  is given by equation~\eqref{eqn:emp_autocorrel}.
Using equations~\eqref{eqn:autocov_Hawkes_def} and~\eqref{eqn:multi_autocorrel}, we have that if $h>0$ is small and $\tau>0$, $\widehat{\cov}(\tau/h)$ approximates the auto-correlation function $\cov(\tau)/\cov(0)$ of the total intensity process $\Sigma_t$. For a multi-exponential Hawkes kernel, one has
\begin{equation}
\widehat{\cov}(\tau/h) \approx \frac{\cov(\tau)}{\cov(0)} = \sum_{j=1}^p \frac{a_j}{\sum_k a_k} \exp(-b_j |\tau|),
\nonumber
\end{equation}
where the coefficients $a_1, \cdots, a_p, b_1, \cdots, b_p > 0$ are determined as in Proposition~\ref{prop_autocov}.
One can see on Figure~\ref{fig:Int_BNPP} that the mono-exponential model fits the end of the curve rather well but that its initial decay is too slow. On the other hand, the multi-exponential model does show a transition between two decay speeds, and captures the short-range behavior of the curve better. Still, the accuracy of the fit is not very satisfactory and it seems that the functional form of the auto-correlation is more subtle than a multi-exponential one.

Finally, we confront our calibration results to the conditions derived in Section~\ref{subsec_noarb} for the absence of Price Manipulation Strategies in the model. It is complicated in practice to quantify the deviation of our set of parameters to the equilibrium using equation~\eqref{cond_MIHM_ext}. On the other hand, equation~\eqref{eqn:DBR_transient} gives a simpler criterion: the directional branching ratio DBR and the proportion $1-\nu$ of transient impact should be equal for PMS to be ruled out. Here, the standard branching ratio BR $\approx 80\%$ is high, but the directional branching ratio DBR $\approx 40\%$ is quite low, which is due to a non-negligible part of cross-excitation in the order flow. It implies that the equilibrium condition is violated since $1-\nu\approx60\%$ for 2012 and $1-\nu\approx80\%$ for 2013. Since $1-\nu>\text{DBR}$ holds in both cases, we find that the price process is mean-reverting on average, rather than diffusive. This should lead to the existence of PMS in practice, which is the object of Section~\ref{section:eval_model}.

\subsection{Total}

Tables~\ref{table:TOTF_reactime}, \ref{table:TOTF_marks_type} and \ref{table:TOTF_resil_int_2012-2013} present our estimation results for the French stock Total on the periods January-September 2012 and January-September 2013.

\begin{table}[h]
\center
\begin{tabular}{|c||c|c|c|c|}\hline
	$\adjlag \text{ (sec)}$	 	& 	0			&	2		&	4		&	6		\\
	$r^2_\text{multi} (2012)$	  	&	$23.093\%$		& 	$23.166\%$	&	$23.137\%$	&	$23.108\%$ \\ 
	$r^2_\text{multi} (2013)$	  	&	$11.604\%$		& 	$11.613\%$	&	$11.608\%$	&	$11.606\%$ \\ \hline
\end{tabular}
\caption{Regression $r^2$ for the multi-exponential resilience curve, evaluated for several market adjustment lags $\adjlag = 0,2,4,6$ seconds, for the stock Total.}
\label{table:TOTF_reactime}
\end{table}

\begin{table}[h]
\center
\begin{tabular}{|c||c|c|c|}\hline
	\text{Marks type}				& 	\text{Unit}		&	\text{Volume}	&	\text{Price jump}		\\
	$\mathcal{L}_\text{mono} (2012)$	&	$2.2981$		& 	$2.3034$		&	$2.2965$			\\ 
	$\mathcal{L}_\text{mono} (2013)$ 	&	$2.2065$		& 	$2.2127$		&	$2.2063$			\\ \hline
\end{tabular}
\caption{Log-likelihood per point for the mono-exponential Hawkes model, evaluated for several types of marks: unit, volumes and price jumps (see eq.~\eqref{three_phis}), for the stock Total.}
\label{table:TOTF_marks_type}
\end{table}

\begin{table}[h]
\center
\parbox{.45\linewidth}{
\begin{tabular}{|c||c|c|}\hline
	\text{Year}					& 	2012			& 	2013			\\ \hline
	$\adjlag \text{ (sec)}$			&  	$2$			& 	$2$	 		\\ \hline
	$\resfact_\text{multi}$			&  	$3.72$		& 	$2.21$		\\
	$\rho_\text{multi}$				& 	$60/360$ 		& 	$6/120/360$ 	\\
	$\lambda_\text{multi}$			& 	$0.29/0.55$ 		& 	$0.004/0.651/0.268$ \\
	$\nu_\text{multi}$				& 	$0.16$ 		& 	$0.08$		 \\
	$\sigma_\text{multi}$			& 	$0.1400$ 		& 	$0.1124$ 		\\
	$r^2_\text{multi}$				& 	$23.166\%$		& 	$11.613\%$ 	\\ \hline
	$\resfact_\text{mono}$			&  	$3.84$		& 	$2.65$ 		\\
	$\rho_\text{mono}$			& 	$187.2$ 		& 	$191.3$ 		\\
	$\lambda_\text{mono}$			& 	$0.84$ 		& 	$0.93$ 		\\
	$\sigma_\text{mono}$			& 	$0.1399$ 		& 	$0.1123$ 		\\
	$r^2_\text{mono}$			& 	$23.132\%$ 	& 	$11.586\%$ 	\\ \hline
\end{tabular}
}
\quad
\parbox{.45\linewidth}{
\begin{tabular}{|c||c|c|}\hline
	\text{Year}				& 	2012				& 	2013					\\ \hline	
	\text{Marks type}			& 	\text{Volume}		& 	\text{Volume}			\\ \hline	
	$\beta_\text{multi}$		&  	$120/360$			& 	$6/60/360$	 			\\	
	$w_\text{multi}$			&  	$0.052/0.948$		& 	$0.010/0.035/0.955$ 		\\	
	${\kappa_\infty}_\text{multi}$	&  	$21.0$			& 	$9.7$					\\ 
	${\phicos}_\text{multi} $		& 	$98.7/21.7$			& 	$84.5/18.5$	 			\\ 
	${\phicoc}_\text{multi} $		& 	$44.3/3.9$	 		& 	$36.5/0.7$ 				\\
	$\mathcal{L}_\text{multi}$	& 	$2.3801$ 			& 	$2.2842$ 				\\ \hline
	$\beta_\text{mono}$		&  	$93.0$			& 	$109.1$ 				\\		
	${\kappa_\infty}_\text{mono}$	&  	$9.2$				& 	$9.0$	 				\\
	${\phicos}_\text{mono}$		& 	$43.5/9.6$ 			& 	$47.4/10.4$ 				\\ 
	${\phicoc}_\text{mono}$		& 	$19.5/1.7$		 	& 	$20.4/0.4$ 				\\ 
	$\mathcal{L}_\text{mono}$	& 	$2.3034$ 			& 	$2.2127$ 				\\ \hline
	\text{BR}				& 	$0.517$ 			& 	$0.688$ 				\\
	\text{DBR}				& 	$0.222$ 			& 	$0.323$ 				\\ \hline
\end{tabular}
}
\caption{Calibration of the resilience (left) and intensity (right) for the stock Total for the periods January-September 2012 and January-September 2013, between 11 a.m. and 1 p.m. For the $\phi$'s, the first entry is the constant term and the second one is the linear term.}
\label{table:TOTF_resil_int_2012-2013}
\end{table}

The qualitative interpretation of the results is similar to that of Section~\ref{section:calib_results_BNPP}. Yet, one should note the following points that are observable on Table~\ref{table:TOTF_resil_int_2012-2013}. First, we notice that there is no significant difference between the mono and multi-exponential propagator. Here, contrary to the BNP Paribas case, the fit is slightly better with two time scales. Second, the branching ratio BR $\approx60\%$ and the directional branching ratio DBR$\approx30\%$ are smaller for Total, whereas the proportion $1-\nu$ of transient impact ($84\%$ for 2012 and $92\%$ for 2013) is higher, which means that the price has an even stronger mean-reversion tendency.

\clearpage

\section{Test of some Price Manipulation Strategies}\label{section:eval_model}

In this section, we apply the optimal strategy derived in~\cite{AB_DynHawkes} and Theorem~\ref{thm_opt_strat} to our dataset, with the parameters obtained by our calibration protocol. Essentially, we run the strategy each day with a zero initial and final position. If the model is relevant, this should give some profit on average. This backtest serves as a practical evaluation of our calibration results, and of the model itself.

\subsection{Scaling and discretization of the optimal strategy}

The simplest and most natural way is to use the optimal strategy~\eqref{opt_strat} is to consider a discrete subset $\Theta$ of $[0,T]$ (possibly made of stopping times) and to trade for each time~$t\in\Theta$ the quantity
\begin{align}
\xi_{t,T}^s &=
- \ \frac{[1+\rho (T-t)] q s D_t
+ X_t}
{2+\rho (T-t)}
\nonumber \\
& \ + \ \frac{m_1}{2\rho} \times \left[(1,\cdots,1) \ . \
 \left\{ I_p + \frac{\rho (T-t)}{2+\rho (T-t)} \times [\zeta((T-t) H)+ \nu \rho (T-t) \ \omega((T-t) H)] \right\} \ . \ s \delta_t \right],
\label{eqn:optimal_trade}
\end{align}
so that~\eqref{opt_strat} holds in~$t+$ if $s=1$. Here, $\delta_t =   \left({\kappa^+_t}^{(i)}-{\kappa^-_t}^{(i)}\right)_i$ is the vector of intensity imbalances and we calculate $D$ by using the following formula
\begin{equation}
D_t = \sum_{\tau \leq t} \Delta \midprice_\tau \ [G(t-\tau)-G(\infty)].
\nonumber
\end{equation}
%


In order to tune the leveraging of the strategy and its discreteness on the market, we introduce a scaling factor $s \in [0,1]$ that multiplies $\delta_t$ and $D_t$. By doing so, we multiply by $s$ the deviation of the whole strategy from the standard Obizhaeva and Wang~\cite{OW} liquidation scheme. The latter is static since it assumes that the observed price process is always a martingale. The limit $s=0$ thus corresponds to the static strategy, whereas $s=1$ is the optimal strategy given by Theorem~\ref{thm_opt_strat}, which may be very aggressive in standard market conditions. In fact, using the optimal strategy with $s=1$ may lead to buy and sell repeatedly quantities that exceed the size of the first queues, which is not realistic.

\subsection{Methodology}

To backtest the strategy in practice, we choose to update our position when we observe midprice moves. Let us define
\begin{equation}
\Theta = \{ \theta \in (\regwin,T),  \ \theta-\tau(\theta)>\adjlag \},
\nonumber
\end{equation}
where the $\theta$'s correspond to the times of price jumps due to cancellations and passive limit orders, $\tau(\theta)$ is the time of the last price jump due to a trade before $\theta$, $\regwin$ is the regression window  defined in Section~\ref{section:estim_propag} and $\adjlag$ is the market adjustment lag.
The position of the strategy at time $t \in [0,T]$ is given by
\begin{equation}
X_t^s = \underset{\theta \in \Theta}{\sum} \xi_{\theta,T}^s.
\nonumber
\end{equation}
At time $T$, we close the position with the transaction
\begin{equation}
\Delta X^s_T = -X^s_T.
\nonumber
\end{equation}
The time horizon is still $T=2$ hours, where $t=0$ corresponds to 11 a.m. and $t=T$ to 1 p.m.
We choose to apply the strategy on $[\regwin,T]$ instead of $[0,T]$, so that the values of $\delta_t$ and $D_t$ for $t\geq \regwin$ can be accurately computed.
Moreover, for each time $\tau \in (\regwin,T)$ where the price jumps because of a transaction, we do not trade on the time interval $[\tau,\tau+\adjlag]$. As a matter of fact, the market adjustment lag $\adjlag$ corresponds approximately to the time needed for the bid-ask to close after a trade that empties the best bid or the best ask. It would be meaningless to trade at the midprice (or even at the midprice $\pm1$ half-tick) before the bid-ask is closed, and we would artificially boost the performance of the strategy if we allowed it. However, this constraint is not needed for simulated data, for which we set $\Theta = \{ \theta \in (\regwin,T) \}$.

We assume that the scaling $s$ is small enough for the effective impact of the strategy on the market price to be negligible. Although approximative, this assumption allows us to backtest the strategy assuming that we can trade at the observed price.

In the sequel of this section, we apply the optimal strategy for the mono and multi-exponential Hawkes decay kernels and for several stocks. We summarize the results in one table and a few graphs for each stock. We note $Y_i$ is the profit made by the strategy on day $i \in \{1,\cdots, n\}$, $\overline{Y_n} = \frac1n \sum_{i=1}^n Y_i$ the empirical mean and $S^2_n = \frac1{n-1}\sum_{i=1}^n [Y_i-\overline{Y_n}]^2$ the empirical variance of daily profits. The values given in the table are
\begin{itemize}
\item The annualized Sharpe ratio of the strategy
\begin{equation}
\text{Sharpe} = \sqrt{n} \times \frac{\overline{Y_n}}{\sqrt{S^2_n}}.
\nonumber
\end{equation}
\item The empirical positivity probability, skew and kurtosis of daily gains
\begin{equation}
\text{Proba.} = \frac1n \sum_{i=1}^n \indi{Y_i>0},
\quad
\text{Skew} = \frac{\frac1n\sum_{i=1}^n [Y_i-\overline{Y_n}]^3}{S^3_n},
\quad
\text{Kurto.} = \frac{\frac1n\sum_{i=1}^n [Y_i-\overline{Y_n}]^4}{S^4_n}.
\nonumber
\end{equation}
\end{itemize}
The choice of the scaling $s$ has no impact on these results, since all the values above are invariant to the multiplication of the strategy by a positive constant.
Thus, only the units of the graphs are changed by the scaling, and we fix $s=0.001$. With this choice, the volumes of individual transactions never exceed $5\%$ of the average volume of the best bid/ask queue, which makes our toy backtest with no impact reasonable. 

For each stock and each period, we also evaluate of the \enquote{Poisson strategy} that one obtains if trades are modeled with two independent compound Poisson processes, which is equivalent to imposing $\kappa^+_t \equiv \overline{\kappa}, \ \kappa^-_t \equiv \overline{\kappa}$ and thus $\kappa^+_t - \kappa^-_t \equiv 0$. More precisely, we trade for $t\in\Theta$ (the same time grid as for the Hawkes model) the quantity
$$\xi^s_{t,T}=
- \ \frac{[1+\rho (T-t)] q s D_t
+ X_t}
{2+\rho (T-t)}.   $$
This strategy is entirely based on mean-reversion, and the trend-following part disappears. For the Hawkes and the Poisson strategies, we give in the tables the impact of a bid-ask cost of one half-tick on the results. This corresponds to a more realistic implementation of the strategy (which should trade at the best and not at the midpoint) and we see that this is sufficient to prevent Price Manipulation Strategies in most cases (the Sharpe ratio becomes close to zero or even negative). As a benchmark, we also present in Table~\ref{table:SIMU_A1_strat_opt} and~\ref{table:SIMU_A2_strat_opt} the results of these strategies on simulated data. These give an idea of the profits that the strategies could reach in theory.

Our findings are the following. On simulated data, the profits made by the strategies are evident and still significant with a half-tick penalty. On real data, the Sharpe ratios remain positive for all the tests, which indicates that the model is not out of scope and captures some characteristics of the real market flow. However, these ratios are lower than for simulated data and may become negative when we take the bid-ask spread into account.
Said differently, market participants who use mean-reverting and trend-following strategies already exploit most of the arbitrage opportunities described by our model, and the backtest of our optimal strategy in realistic market conditions does not yield significant gains.
 Somehow, this justifies the theoretical assumption to consider a market without PMS when dealing with both market impact and the bid-ask spread. 
Now, let us compare the different strategies used in Tables~\ref{table:BNPP_strat_opt_2012-2013} and~\ref{table:TOTF_strat_opt_2012-2013}. The results are rather similar for the three strategies and none of them seem to outperform the others. Intuitively, this means that the main component of the strategy is the mean-reverting one (which is common to the Poisson and Hawkes strategies), while the trend-following one has a minor contribution. This is confirmed by the statistical facts in Table~\ref{table:BNPP_resil_int_2012-2013} and~\ref{table:TOTF_resil_int_2012-2013} where the directional branching ratio DBR is much lower than the proportion of transient impact $\lambda_\text{mono} = 1-\nu$.


\subsection{Simulated data}

Tables~\ref{table:SIMU_A1_strat_opt} and~\ref{table:SIMU_A2_strat_opt} present the results of the optimal  strategy applied to simulated data. The simulation parameters are the same as in Section~\ref{section:simu_calib} (see Tables~\ref{table:SIMU_A1_resil_int} and~\ref{table:SIMU_A2_resil_int}), and both datasets are composed of $150$ independent two-hour windows. In Tables~\ref{table:SIMU_A1_strat_opt} and~\ref{table:SIMU_A2_strat_opt}, the two first columns contain the results of the strategy computed with the real simulation parameters for the Hawkes model, and the third and fourth columns contain the results for estimated Hawkes parameters. In both cases, the resilience is the estimated mono-exponential curve, since the optimal strategy is known explicitly only in that case.

\begin{table}[h]
\center
\begin{tabular}{|c||c|c||c|c|}\hline
	\text{Year}		& \text{Simu.}	& $+ \text{bid-ask}$& 	\text{Calib.}	& $+ \text{bid-ask}$	\\ \hline
	Sharpe (Multi)	&  	$6.759$	&$3.225$		& 	$6.764$ 	& $3.176$			\\ 
	Proba. (Multi)	& 	$74.0\%$ 	&$63.3\%$		& 	$74.0\%$	& $63.3\%$			\\ 
	Skew (Multi)		&  	$0.55$	&$0.23$		& 	$0.57$ 	& $0.24$			\\ 
	Kurtosis (Multi)	&  	$4.19$	&$4.03$		& 	$4.22$ 	& $4.05$			\\ \hline \hline
	Sharpe (Mono)	&  	$-$		&$-$			& 	$6.308$ 	& $3.371$			\\ 
	Proba. (Mono)	& 	$-$		&$-$			& 	$74.0\%$	& $62.7\%$			\\ 
	Skew (Mono)		&  	$-$		&$-$			& 	$0.47$ 	& $0.20$			\\ 
	Kurto. (Mono)	&  	$-$		&$-$			& 	$4.11$ 	& $3.97$			\\ \hline \hline
	Sharpe (Poisson)	&  	$-$		&$-$			& 	$6.630$ 	& $3.735$			\\ 
	Proba. (Poisson)	& 	$-$		&$-$			& 	$73.3\%$	& $64.0\%$ 			\\ 
	Skew (Poisson)	&  	$-$		&$-$			& 	$0.43$ 	& $0.18$			\\ 
	Kurto. (Poisson)	&  	$-$		&$-$			& 	$3.88$ 	& $3.80$			\\ \hline
\end{tabular}
\caption{Results statistics of the optimal  strategy applied on the data of Simulation 1 (simulation parameters of Table~\ref{table:SIMU_A1_resil_int}).}
\label{table:SIMU_A1_strat_opt}
\end{table}

\begin{table}[h]
\center
\begin{tabular}{|c||c|c||c|c|}\hline
	\text{Year}		& \text{Simu.}	& $+ \text{bid-ask}$& 	\text{Calib.}	& $+ \text{bid-ask}$	\\ \hline
	Sharpe (Multi)	&  	$33.268$	&$27.095$		& 	$32.302$ 	& $25.769$			\\ 
	Proba. (Multi)	& 	$100.0\%$ 	&$100.0\%$		& 	$100.0\%$	& $99.3\%$		\\ 
	Skew (Multi)		&  	$0.50$	&$0.51$		& 	$0.52$ 	& $0.54$			\\ 
	Kurtosis (Multi)	&  	$3.22$	&$3.35$		& 	$3.25$ 	& $3.40$			\\ \hline \hline
	Sharpe (Mono)	&  	$-$		&$-$			& 	$34.940$ 	& $28.605$			\\ 
	Proba. (Mono)	& 	$-$		&$-$			& 	$100.0\%$	& $100.0\%$		\\ 
	Skew (Mono)		&  	$-$		&$-$			& 	$0.45$ 	& $0.46$			\\ 
	Kurto. (Mono)	&  	$-$		&$-$			& 	$3.19$ 	& $3.31$			\\ \hline \hline
	Sharpe (Poisson)	&  	$-$		&$-$			& 	$34.986$ 	& $28.681$			\\ 
	Proba. (Poisson)	& 	$-$		&$-$			& 	$100.0\%$	& $100.0\%$ 		\\ 
	Skew (Poisson)	&  	$-$		&$-$			& 	$0.44$ 	& $0.45$			\\ 
	Kurto. (Poisson)	&  	$-$		&$-$			& 	$3.12$ 	& $3.25$			\\ \hline
\end{tabular}
\caption{Results statistics of the optimal  strategy applied on the data of Simulation 2 (simulation parameters of Table~\ref{table:SIMU_A2_resil_int}).}
\label{table:SIMU_A2_strat_opt}
\end{table}


\subsection{BNP Paribas}

\begin{table}[h]
\center
\begin{tabular}{|c||c|c||c|c||c|c|}\hline
	\text{Year}		& IS 2012& $+ \text{bid-ask}$& 	IS 2013& $+ \text{bid-ask}$&  	OS 2013& $+ \text{bid-ask}$	\\ \hline
	Sharpe (Multi)	&  	$1.382$	&$-0.675$	& 	$2.454$ 	& $0.725$	&	$2.248$	& $0.418$	\\ 
	Proba. (Multi)	& 	$65.9\%$ 	&$56.5\%$	& 	$61.3\%$	& $47.1\%$	&	$58.1\%$	& $48.2\%$	\\ 
	Skew (Multi)		&  	$-2.02$	&$-2.40$	& 	$3.65$ 	& $3.34$	&	$4.48$	& $4.14$	\\ 
	Kurtosis (Multi)	&  	$19.02$	&$19.94$	& 	$29.40$ 	& $27.71$	&	$36.96$	& $34.65$	\\ \hline \hline
	Sharpe (Mono)	&  	$1.263$	&$-0.713$	& 	$2.536$ 	& $0.771$	&	$2.430$	& $0.563$	\\ 
	Proba. (Mono)	& 	$62.9\%$	&$57.1\%$	& 	$62.3\%$	& $48.2\%$	&	$58.1\%$	& $49.7\%$	\\ 
	Skew (Mono)		&  	$-1.89$	&$-2.30$	& 	$2.94$ 	& $2.61$	&	$3.56$	&$3.21$	\\ 
	Kurto. (Mono)	&  	$16.64$	&$17.68$	& 	$23.27$ 	& $21.90$	&	$26.74$	& $24.87$	\\ \hline \hline
	Sharpe (Poisson)	&  	$1.056$	&$-0.849$	& 	$2.5888$ 	& $0.8077$	&	$2.513$	& $0.630$	\\ 
	Proba. (Poisson)	& 	$65.3\%$	&$55.9\%$	& 	$61.3\%$	& $49.7\%$	&	$60.2\%$	& $49.2\%$	\\ 
	Skew (Poisson)	&  	$-2.72$	&$-3.07$	& 	$3.09$ 	& $2.76$	&	$3.94$	& $3.58$	\\ 
	Kurto. (Poisson)	&  	$23.46$	&$24.68$	& 	$24.41$ 	& $22.82$	&	$31.13$	& $28.86$	\\ \hline
\end{tabular}
\caption{Results statistics of the optimal  strategy applied on BNP Paribas on the periods February-September 2012 and January-September 2013, every day between 11.30a.m. and 1p.m. The first two columns are In-Sample results, i.e. the data used to calibrate the model is the same as the evaluation data. The third column gives Out-of-Sample results, i.e. we calibrate the model on the 2012 data to apply the strategy on the 2013 data.}
\label{table:BNPP_strat_opt_2012-2013}
\end{table}


\begin{figure}[H]
	\centering
\subfigure[Trading at the midprice]{
    \includegraphics[height=7cm,width=7cm]{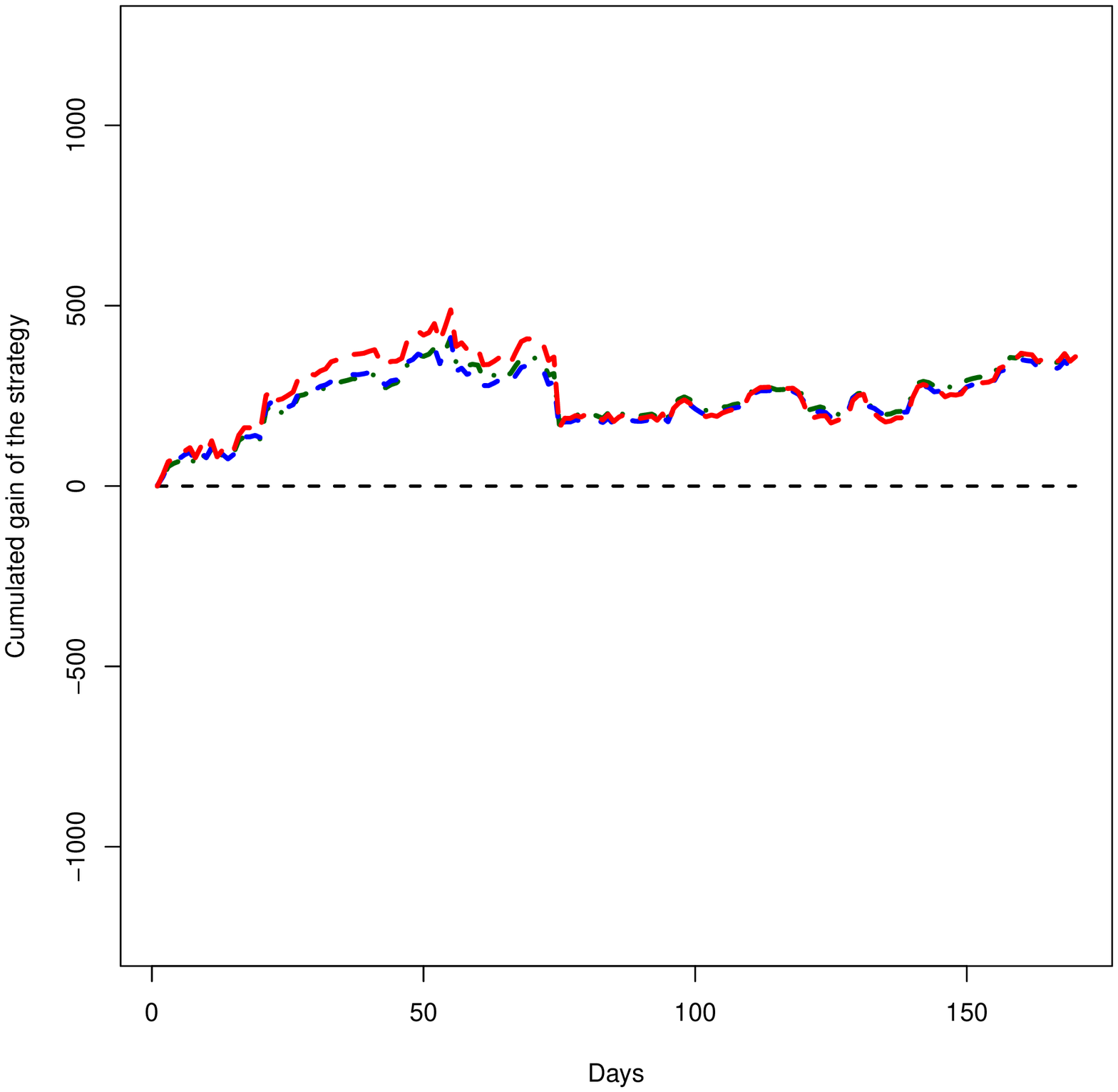}
	\label{fig:CumGain_tick0_BNPP_2012}
    }
\subfigure[One half-tick penalty]{
    \includegraphics[height=7cm,width=7cm]{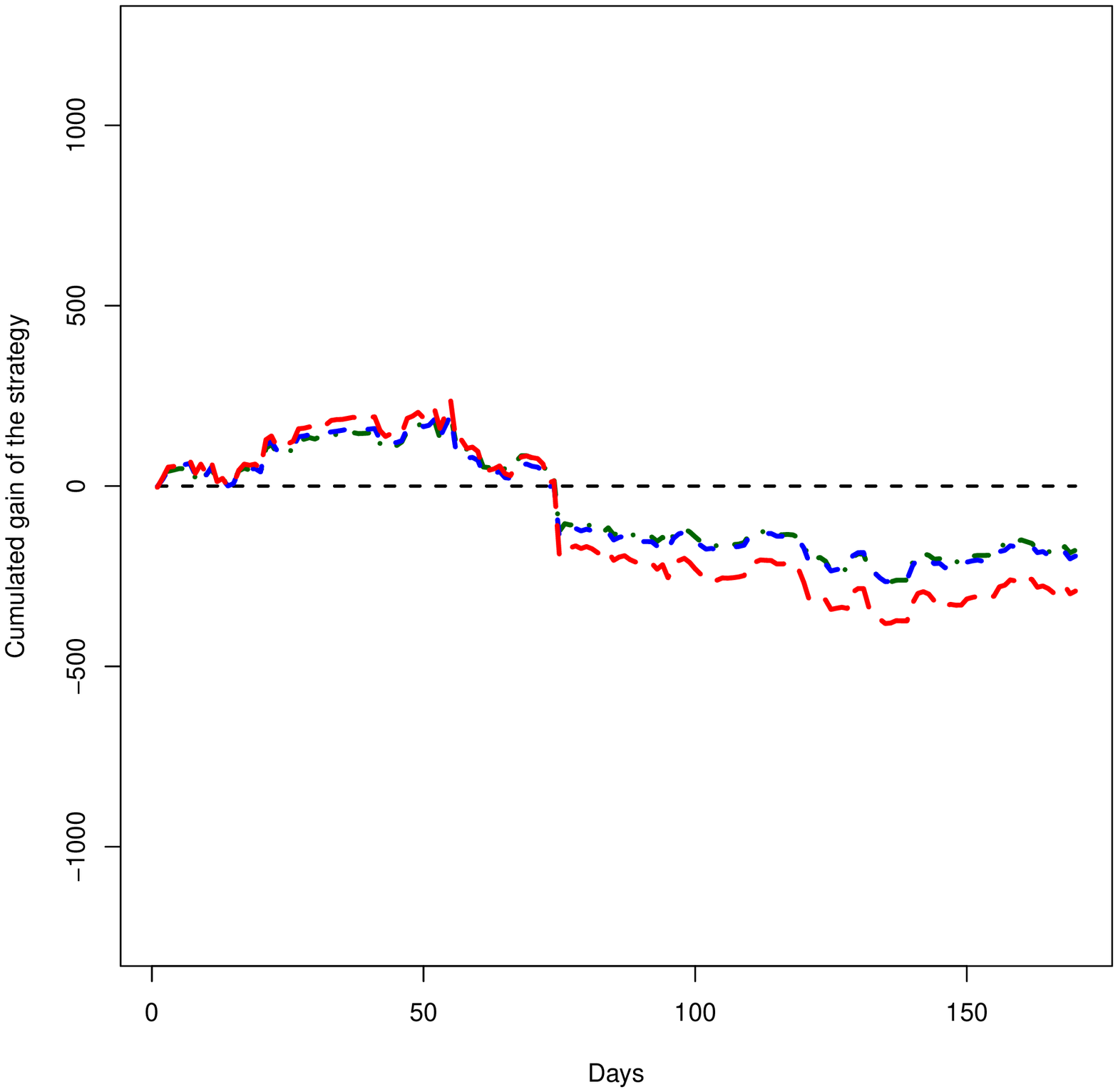}
	\label{fig:CumGain_tick1_BNPP_2012}
    }
	\caption{Cumulated gains of the strategy applied on BNP Paribas on the period February-September 2012, every day between 11.30a.m. and 1p.m. The (red) long-dashed line is the performance of the Poisson model, the (blue) dashed line is the mono-exponential Hawkes model, and the (green) dot-dashed line is the multi-exponential Hawkes model. Left: we allow the strategy to trade at the midprice. Right: we apply a posteriori a linear cost penalty of one half-tick to account for the bid-ask spread.}
    \label{fig:Eval_BNPP_2012}
\end{figure}

\begin{figure}[H]
	\centering
\subfigure[Mono-exponential Hawkes]{
    \includegraphics[height=7cm,width=7cm]{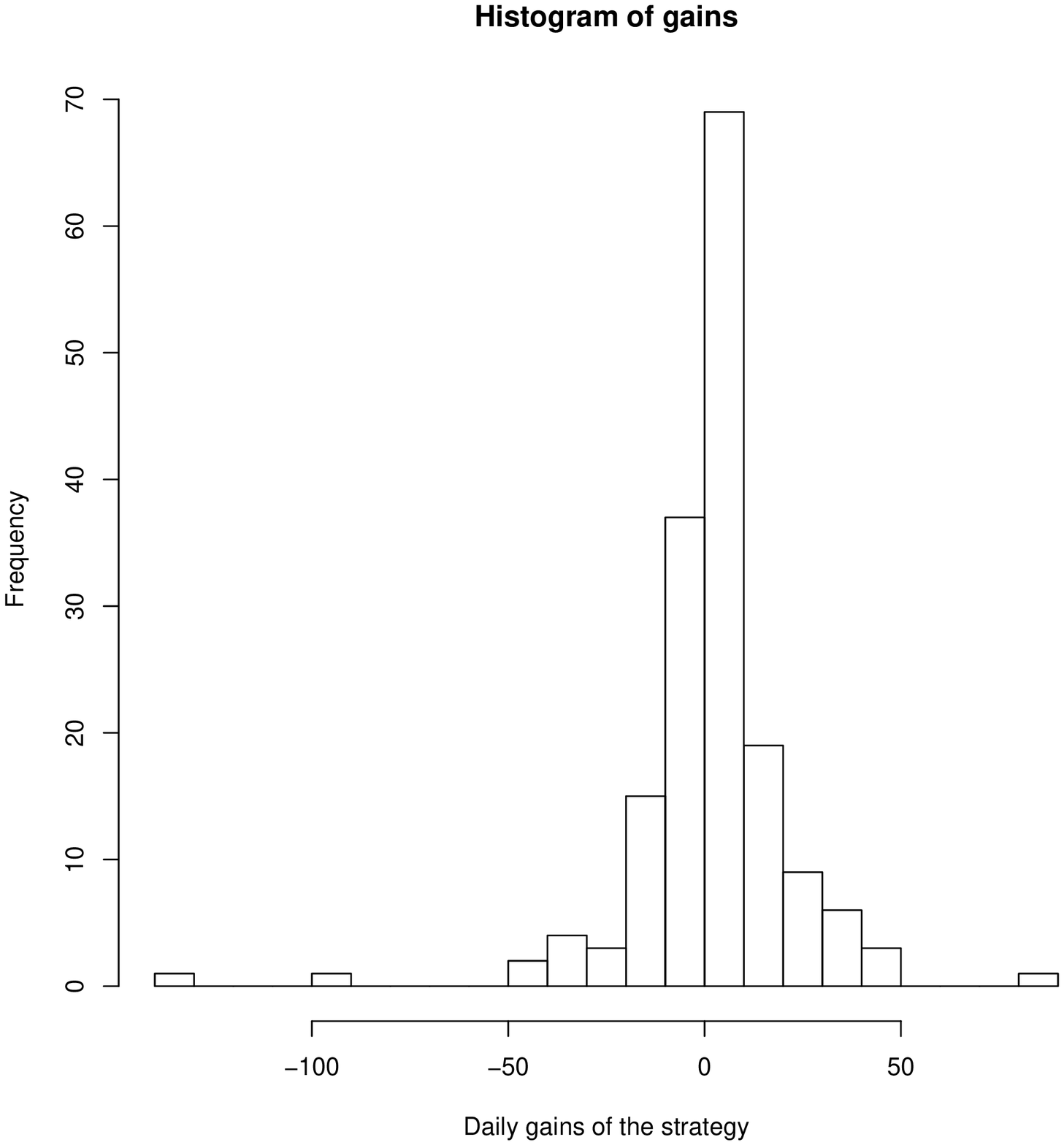}
    \label{fig:Histo_mono_BNPP_2012}
}
\subfigure[Multi-exponential Hawkes]{
    \includegraphics[height=7cm,width=7cm]{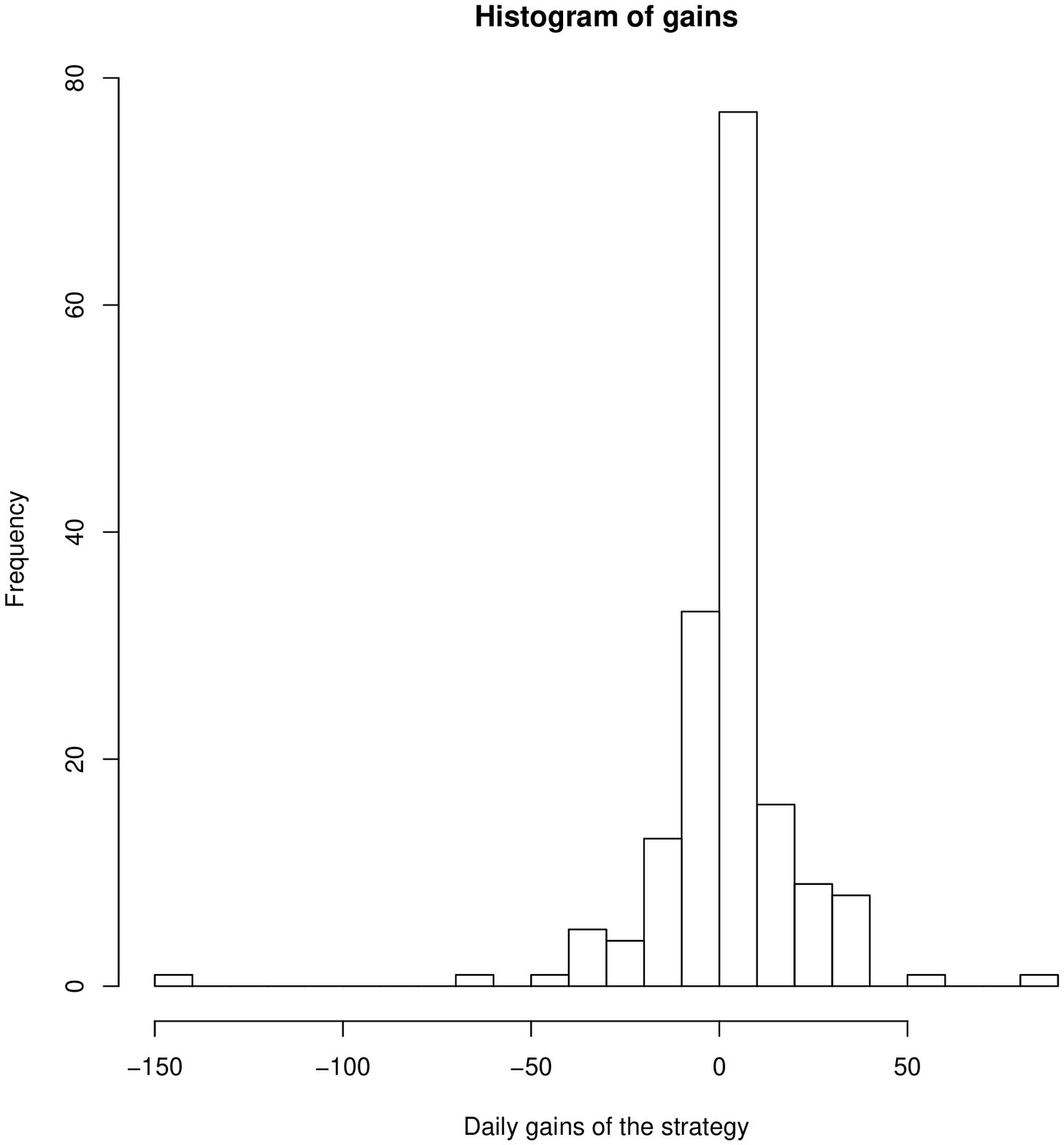}
    \label{fig:Histo_multi_BNPP_2012}
}
\caption{Histogram of the daily gains of the strategy applied on BNP Paribas on the period February-September 2012, between 11.30a.m. and 1p.m. Left: Mono-exponential Hawkes model. Right: Multi-exponential Hawkes model.}
    \label{fig:Histo_BNPP_2012}
\end{figure}



\begin{figure}[H]
	\centering
\subfigure[Trading at the midprice]{
    \includegraphics[height=7cm,width=7cm]{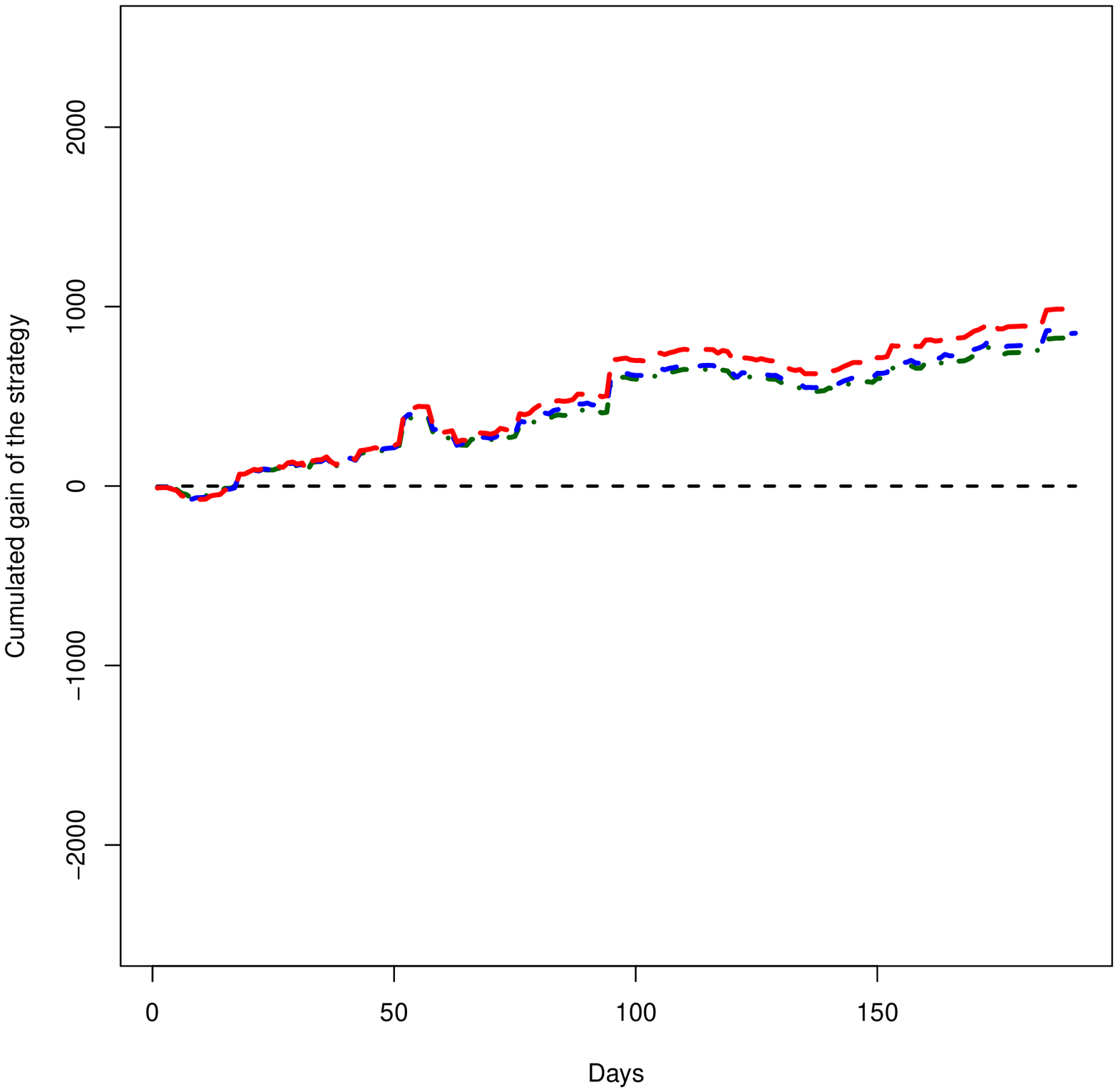}
	\label{fig:CumGain_tick0_BNPP_2013}
    }
\subfigure[One half-tick penalty]{
    \includegraphics[height=7cm,width=7cm]{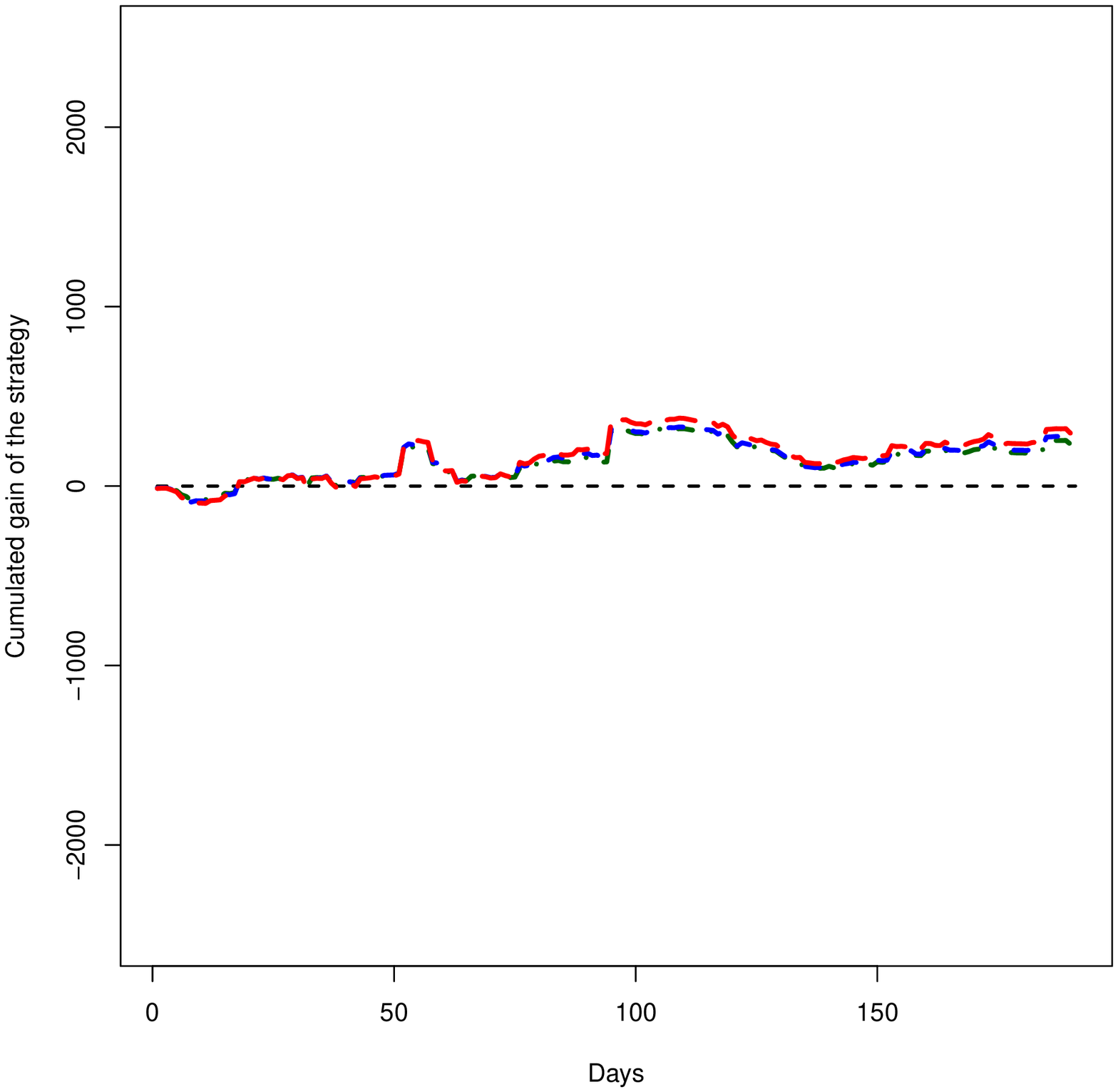}
	\label{fig:CumGain_tick1_BNPP_2013}
    }
	\caption{Cumulated gains of the strategy applied on BNP Paribas on the period January-September 2013, every day between 11.30a.m. and 1p.m. The (red) long-dashed line is the performance of the Poisson model, the (blue) dashed line is the mono-exponential Hawkes model, and the (green) dot-dashed line is the multi-exponential Hawkes model. Left: we allow the strategy to trade at the midprice. Right: we apply a posteriori a linear cost penalty of one half-tick to account for the bid-ask spread.}
    \label{fig:Eval_BNPP_2013}
\end{figure}

\begin{figure}[H]
	\centering
\subfigure[Mono-exponential Hawkes]{
    \includegraphics[height=7cm,width=7cm]{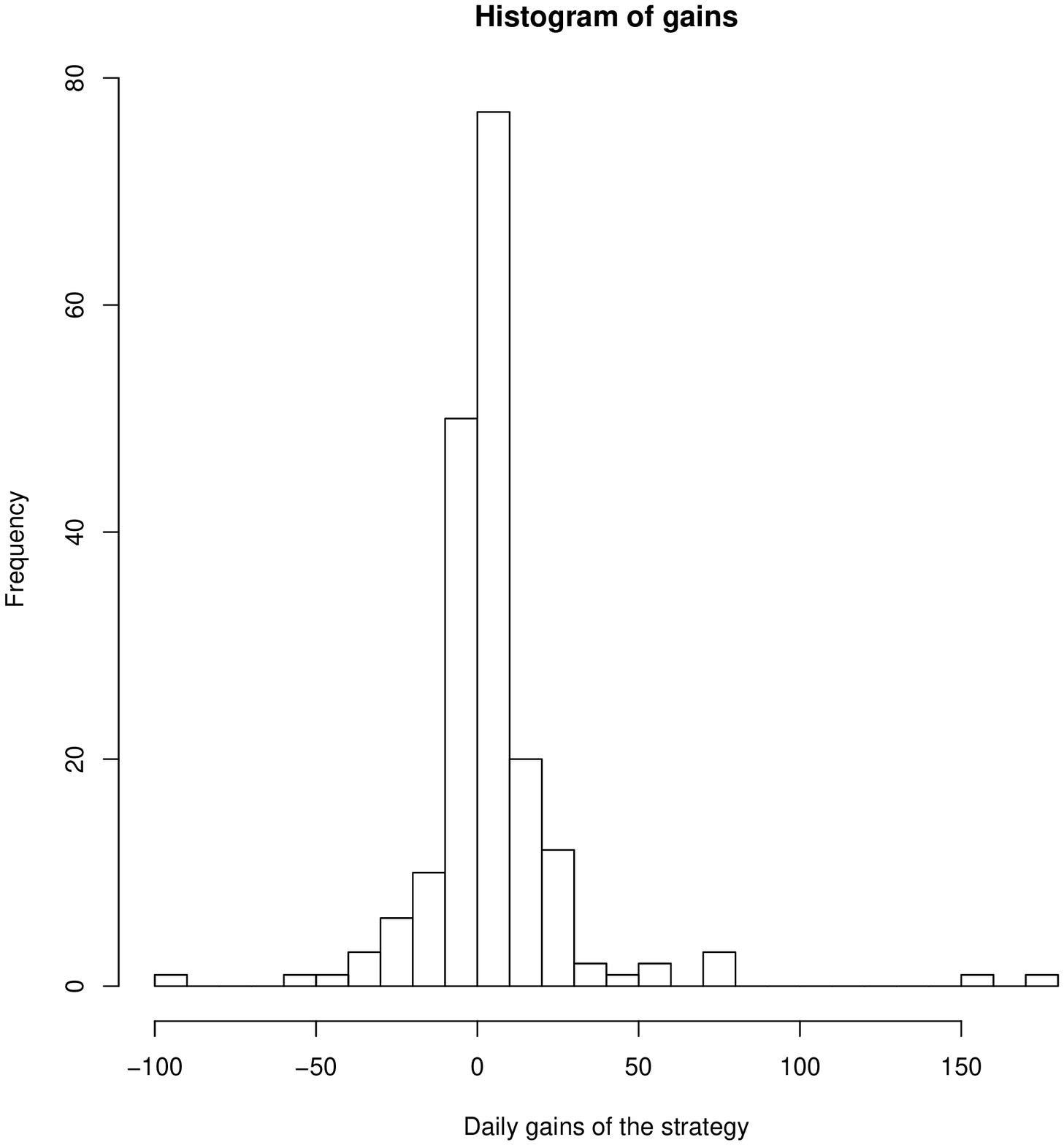}
    \label{fig:Histo_mono_BNPP_2013}
}
\subfigure[Multi-exponential Hawkes]{
    \includegraphics[height=7cm,width=7cm]{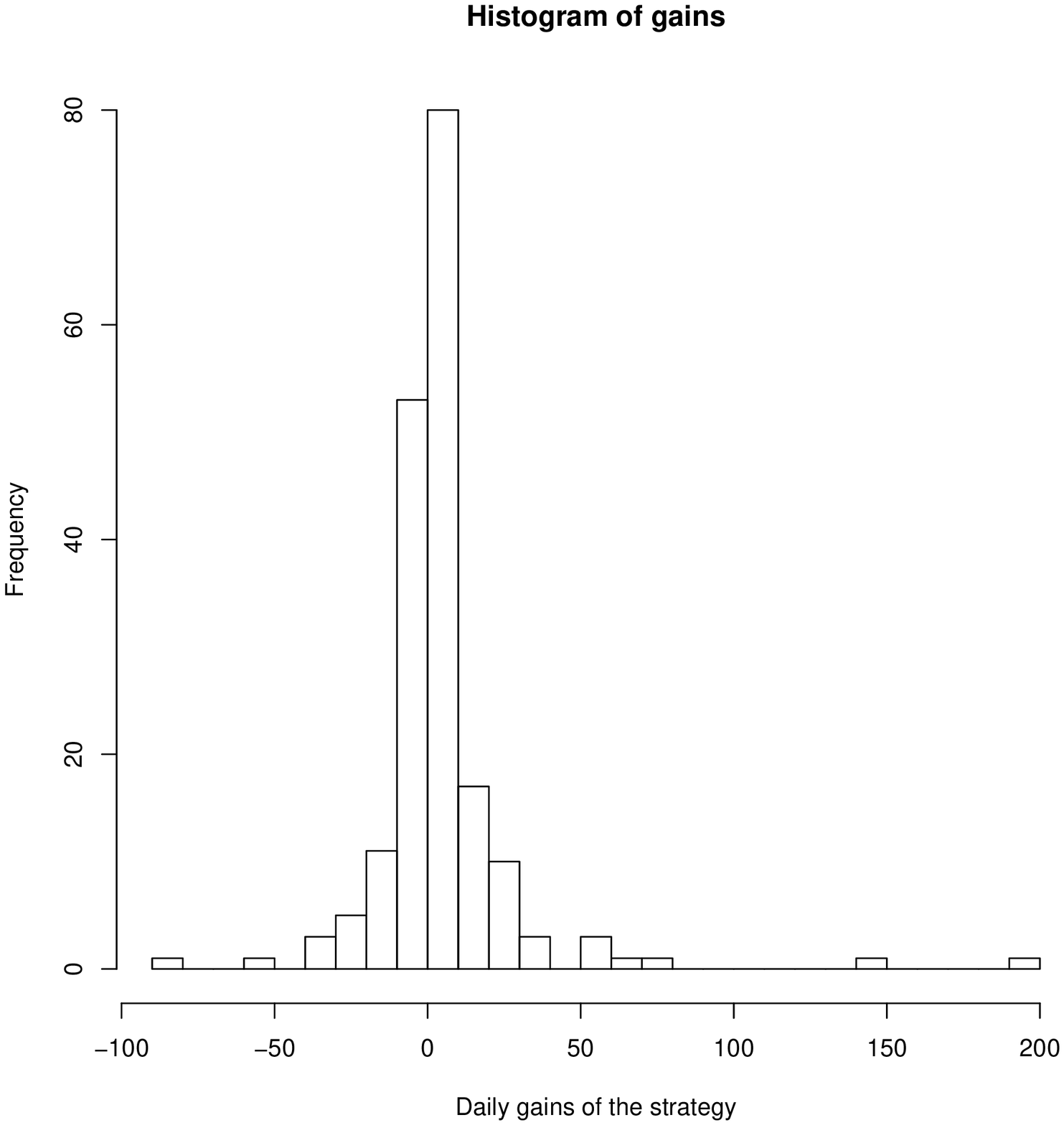}
    \label{fig:Histo_multi_BNPP_2013}
}
\caption{Histogram of the daily gains of the strategy applied on BNP Paribas on the period January-September 2013, between 11.30a.m. and 1p.m. Left: Mono-exponential Hawkes model. Right: Multi-exponential Hawkes model.}
    \label{fig:Histo_BNPP_2013}
\end{figure}


\clearpage

\subsection{Total}
\begin{table}[h]
\center
\begin{tabular}{|c||c|c||c|c||c|c|}\hline
	\text{Year}		& IS 2012& $+ \text{bid-ask}$& 	IS 2013& $+ \text{bid-ask}$&  	OS 2013& $+ \text{bid-ask}$	\\ \hline
	Sharpe (Multi)	&  	$0.067$	&$-0.763$	& 	$2.697$ 	& $1.016$	&	$2.794$	& $1.224$	\\ 
	Proba. (Multi)	& 	$57.8\%$ 	&$44.3\%$	& 	$66.0\%$	& $51.8\%$	&	$65.4\%$	& $51.8\%$	\\ 
	Skew (Multi)		&  	$-9.34$	&$-9.62$	& 	$6.38$ 	& $6.37$	&	$5.94$	& $5.97$	\\ 
	Kurtosis (Multi)	&  	$114.76$	&$117.75$	& 	$62.86$ 	& $65.85$	&	$53.84$	& $57.93$	\\ \hline \hline
	Sharpe (Mono)	&  	$0.126$	&$-0.770$	& 	$2.795$ 	& $1.191$	&	$2.760$	& $1.099$	\\ 
	Proba. (Mono)	& 	$59.4\%$	&$44.8\%$	& 	$66.0\%$	& $52.4\%$	&	$65.4\%$	& $52.4\%$	\\ 
	Skew (Mono)		&  	$-9.52$	&$-9.82$	& 	$6.01$ 	& $6.02$	&	$6.18$	&$6.18$	\\ 
	Kurto. (Mono)	&  	$118.29$	&$121.77$	& 	$55.54$ 	& $59.30$	&	$59.20$	& $62.65$	\\ \hline \hline
	Sharpe (Poisson)	&  	$0.001$	&$-0.810$	& 	$2.807$ 	& $1.259$	&	$2.790$	& $1.224$	\\ 
	Proba. (Poisson)	& 	$57.8\%$	&$43.8\%$	& 	$65.4\%$	& $50.8\%$	&	$65.4\%$	& $50.8\%$	\\ 
	Skew (Poisson)	&  	$-9.33$	&$-9.59$	& 	$5.96$ 	& $6.00$	&	$6.04$	& $6.08$	\\ 
	Kurto. (Poisson)	&  	$114.39$	&$116.97$	& 	$53.37$ 	& $57.35$	&	$54.90$	& $58.87$	\\ \hline
\end{tabular}
\caption{Results statistics of the optimal  strategy applied on Total on the period January-September 2012-2013, every day between 11.30a.m. and 1p.m. The first two columns are In-Sample results, i.e. the data used to calibrate the model is the same as the evaluation data. The third column gives Out-of-Sample results, i.e. we calibrate the model on the 2012 data to apply the strategy on the 2013 data.}
\label{table:TOTF_strat_opt_2012-2013}
\end{table}

\clearpage

\section{Conclusion}

In this paper we extend the theoretical model of~\cite{AB_DynHawkes} by allowing more general forms for the propagator and the Hawkes kernel.
Moreover, we derive the conditions that exclude Price Manipulation Strategies in the sense of Huberman and Stanzl~\cite{HS} in the case where both the propagator and the Hawkes part have a multi-exponential decay. This allows us to deduce some interesting links between the propagator and the Hawkes kernel for general completely monotone kernels. Besides, when the price propagator is mono-exponential and the Hawkes kernel is multi-exponential, we can still obtain the optimal strategy as a closed formula. This has some practical interest since the propagator seems to be well approximated by an exponential, while the Hawkes decay kernel clearly includes several characteristic time scales.

We also introduce a calibration protocol for the model, that we apply to tick-by-tick data from French stocks. The results show that the model explains a significant part of the variance of prices. The long-range propagator is a smoothly decaying curve, but the short-range part is increasing during a few seconds (which we think corresponds to the time that the bid-ask needs to close after a large trade). Concerning the estimation of the Hawkes process modeling the flow of trades, we obtain excitation parameters that significantly differ from zero, which shows in particular that the flow is not Poissonian. Also, we find that the main driver of the excitation between trades is volumes rather than price moves. The martingale conditions that prevent PMS are violated in practice, in particular the directional branching ratio is smaller than the proportion of transient price impact. Therefore, in our dataset, the price has a notable mean-reverting tendency.

A series of backtests shows that the optimal strategy used for round trips is profitable on average, therefore the model does offer a relevant prediction for midprice moves. However, a level of transaction costs compatible with the width of the bid-ask spread makes the profits close to zero. This confirms the natural idea that the absence of Price Manipulation Strategies at this frequency stems from both market impact and bid-ask costs.

We eventually draw some applications and perspectives on our study. A first straightforward application is to use the calibrated model for optimal execution, by using the block trades~\eqref{eqn:optimal_trade} on a given (possibly random) time grid~$\Theta$. Contrary to most existing models, this strategy takes the flow of trades into account. Another possible use of this model is to detect the instants when it is interesting to trade. In fact, equation~\eqref{eqn:dA2_Hawkes} gives the (theoretical) instantaneous cost of non-trading. One may decide to trade for example only if this cost is above some threshold, or optimize the trade-off between this cost and transaction costs. Such strategies could be interesting in practice, but need to be thoroughly investigated on market data. Let us now consider some possible extensions of our work. First, it would be interesting to handle a calibration of the model on an entire day instead of a two-hour window. This is certainly difficult due to intra-day variations of trading activity between the open and the close. Second, it would be nice to incorporate in our model transaction costs such as the bid-ask spread. A less ambitious goal would be at least to modify our optimal execution strategy to reduce transaction costs in a clever way, maybe by using equation~\eqref{eqn:dA2_Hawkes} as mentioned above.

\clearpage

\bibliography{ref_Hawkes_calib}
\bibliographystyle{plain}

\clearpage

\appendix

\section{Estimation of the propagator using Newton-Raphson's algorithm}\label{appendix:NR_propagator}

As explained in Section~\ref{section:estim_propag}, we resort to Newton-Raphson's algorithm to minimize the quadratic error
\begin{equation}
\mathcal{E}(\hat G) = \underset{\regwin<\theta<T}{\sum} [\hat P_\theta - P_\theta]^2
\nonumber
\end{equation}
which quantifies the distance between the observed midpoint price $P_t$ and the predicted price
\begin{equation}
\hat P_t \ = \ P_{t-\regwin}
\ + \ \underset{t-\regwin \leq \tau \leq t}{\sum} \Delta \midprice_\tau \ \hat G(t-\tau).
\nonumber
\end{equation}
Let us assume that $\pi \in \R^l$, $l \geq 1$, is a parameterization of $\hat G$, i.e. $\hat G = \hat G(\pi)$ is determined by $\pi$, and so is the error
$\mathcal{E}(\hat G) = \mathcal{E}(\pi)$.
For a starting point $\pi_0$, the principle of the algorithm is to approximate $G$ by the sequence $\hat G(\pi_n)$ such that
\begin{equation}
\forall n\in \N, \quad \pi_{n+1} \ = \ \pi_n 
\ - \ \big[\nabla^2 \mathcal{E}(\pi_n)\big]^{-1} . \nabla \mathcal{E}(\pi_n)
\nonumber
\end{equation}
where $\nabla \mathcal{E}(\pi) \in \R^l$ is the gradient of the error $\mathcal{E}$ and $\nabla^2 \mathcal{E}(\pi) \in \R^{l\times l}$ is its Hessian matrix, w.r.t. the parameter $\pi$. The convergence of the method is only guaranteed if the starting point $\pi_0$ is \enquote{good enough}, and if $\nabla^2 \mathcal{E}(\pi_n)$ is positive definite for all $n \in \N$.

To apply this method, one needs to compute the gradient $\nabla \mathcal{E}(\pi)$ and the Hessian matrix 
$\nabla^2 \mathcal{E}(\pi)$ of the error $\mathcal{E}$ for each parameterization $\pi$ of $\hat G$. One has
\begin{align}
\nabla \mathcal{E}(\pi) &= 2 \underset{\regwin<\theta<T}{\sum}
 [\hat P_\theta(\pi) - P_\theta] \times \nabla \hat P_\theta(\pi),
\nonumber \\
\nabla^2 \mathcal{E}(\pi) &= 2 \underset{\regwin<\theta<T}{\sum}
 \left\{ [\hat P_\theta(\pi) - P_\theta] \times \nabla^2 \hat P_\theta(\pi)
\ + \ \nabla \hat P_\theta(\pi) . \left(\nabla \hat P_\theta(\pi)\right)\tp
\right\}.
\nonumber
\end{align}
The problem boils down to computing $\nabla \hat P_\theta(\pi)$ and $\nabla^2 \hat P_\theta(\pi)$, which can themselves be expressed as
\begin{align}
\nabla \hat P_\theta(\pi) &=
\underset{t-\regwin \leq \tau \leq t}{\sum} \Delta \midprice_\tau \ \nabla \hat G(t-\tau),
\nonumber \\
\nabla^2 \hat P_\theta(\pi) &=
\underset{t-\regwin \leq \tau \leq t}{\sum} \Delta \midprice_\tau \ \nabla^2 \hat G(t-\tau),
\nonumber
\end{align}
where we drop the dependency of $\hat G$ in $\pi$ for clarity. Therefore, only the gradient $\nabla \hat G$ and the Hessian $\nabla^2 \hat G$ of the estimated propagator $\hat G$ need to be specifically derived for each parameterization, which is the object of the sequel.

An important particular case is when $\hat G$ is linear w.r.t. $\pi$. In that case, $\nabla^2 \hat G \equiv 0$, thus $\nabla^2 \hat P_\theta(\pi) \equiv 0$ and 
\begin{equation}
\nabla^2 \mathcal{E}(\pi) = 2 \underset{\regwin<\theta<T}{\sum}
\nabla \hat P_\theta(\pi) . \left(\nabla \hat P_\theta(\pi)\right)\tp
\nonumber
\end{equation}
is positive definite for any $\pi$. Also, in that case, $\nabla \hat G$ does not depend on the current values of the parameter $\pi$, and
\begin{equation}
\pi_1 \ = \ \pi_0
\ - \ \big[\nabla^2 \mathcal{E}(\pi_0)\big]^{-1} . \nabla \mathcal{E}(\pi_0)
\nonumber
\end{equation}
is the minimizer of the error $\mathcal{E}(\pi)$ for any $\pi_0$. Therefore, when the propagator is parameterized linearly, the starting point of the algorithm has no importance and one step is enough to find the optimum.

\subsection{Unconstrained propagator}\label{appendix:calib_UC}

We consider the unconstrained propagator
\begin{equation}
\hat{G}(t) = g_l \mathbf{1}_{[t_l,\regwin[}(t) + \overset{l-1}{\underset{i=0}{\sum}} \frac{(t_{i+1}-t) g_i + (t-t_i)g_{i+1}}{t_{i+1}-t_i} 
 \ \mathbf{1}_{[t_i,t_{i+1}[}(t),
\nonumber
\end{equation}
with $l \geq 2$, $0=t_0 < t_1 < \cdots < t_l$ fixed discretization times, $g_0 = 1$ and $\pi = (g_1,\cdots,g_l) \in [0,+\infty)^l$ the $l$-dimensional parameter to estimate. The dependence of $\hat G$ w.r.t. $\pi$ is linear, and we only need to compute the gradient:
\begin{align}
\frac{\partial \hat{G}(t)}{\partial g_i} & = \frac{t_{i+1}-t}{t_{i+1}-t_i} \mathbbm{1}_{[t_i,t_{i+1}[}(t)
+ \frac{t-t_{i-1}}{t_i-t_{i-1}} \mathbbm{1}_{[t_{i-1},t_i[}(t) 
\quad \text{ for } 1 \leq i \leq l-1,
\nonumber \\
\frac{\partial \hat{G}(t)}{\partial g_l} & = \mathbbm{1}_{[t_l,\regwin[}(t)
+ \frac{t-t_{l-1}}{t_l-t_{l-1}} \mathbbm{1}_{[t_{l-1},t_l[}(t).
\nonumber
\end{align}

\subsection{Multi-exponential curve}\label{appendix:calib_ME}

In this section we consider the multi-exponential resilience curve
\begin{equation}
\hat R(t) = \nu + \underset{i=1}{\overset{p}{\sum}} \lambda_i \exp(-\rho_i t),
\nonumber
\end{equation}
and the propagator
\begin{equation}
\hat G(t) = \left[1+(\hat R(\adjlag)-1) \ \frac{t}{\adjlag}\right] \indi{t\leq \adjlag} + \hat R(t) \indi{t>\adjlag},
\nonumber
\end{equation}
determined by $\hat R$ for $\adjlag \geq 0$ fixed \textit{a priori}. The dependence of $\hat G$ is linear w.r.t. the parameters if and only if the $\rho_i$'s are fixed.

\subsubsection{Unit Multi-exponential curve}\label{appendix:calib_ME_unit}

The \enquote{unit} multi-exponential resilience curve is the case where $\nu = 1-\sum_{i=1}^p \lambda_i$ is imposed.
This yields
\begin{equation}
\hat R(t) = 1-\overset{p}{\underset{i=1}{\sum}}\lambda_i(1-\exp(-\rho_i t)),
\nonumber
\end{equation}
and the parameter $\pi = (\lambda_1, \cdots, \lambda_p,\rho_1, \cdots, \rho_p)$ is $2p$-dimensional. One has for $i,j \in \{1,\cdots,p\}$,
\begin{align}
&\frac{\partial \hat R(t)}{\partial \lambda_i} = - \left\{ 1 - \exp(-\rho_i t) \right\},
&\frac{\partial \hat R(t)}{\partial \rho_i} = - t \ \lambda_i \exp(-\rho_i t),
\nonumber \\
&\frac{\partial^2 \hat R(t)}{\partial \rho_i^2} = t^2 \ \lambda_i \exp(-\rho_i t),
&\frac{\partial^2 \hat R(t)}{\partial \rho_i \partial \lambda_i} = -  t \ \exp(-\rho_i t),
\nonumber \\
&\frac{\partial^2 \hat R(t)}{\partial \lambda_i \partial \lambda_j} = 0,
&\frac{\partial^2 \hat R(t)}{\partial \rho_i \partial \rho_j} = 0,
\quad \frac{\partial^2 \hat R(t)}{\partial \lambda_i \partial \rho_j} = 0
\quad \text{ if } i\neq j.
\nonumber
\end{align}

\subsubsection{General Multi-exponential curve}\label{appendix:calib_ME_gene}

If we relax the condition $\nu = 1-\sum_{i=1}^p \lambda_i$ so that $\hat R(0)$ can be greater than unity, we obtain
\begin{equation}
\hat R(t) = \overline \nu + \underset{i=1}{\overset{p}{\sum}} \overline \lambda_i \exp(-\rho_i t),
\nonumber
\end{equation}
with $\overline{\nu} \geq 0$, $\overline{\lambda}_i \geq 0$. The parameter $\pi = (\overline \nu, \overline \lambda_1, \cdots, \overline \lambda_p,\rho_1, \cdots, \rho_p)$ is then $(2p+1)$-dimensional. The gradient and Hessian are given by
\begin{align}
&\frac{\partial \hat R(t)}{\partial \overline \nu} = 1, \
&\frac{\partial^2 \hat R(t)}{\partial \overline \nu^2} = 0, \
\frac{\partial^2 \hat R(t)}{\partial \overline \nu \partial \overline \lambda_i} = 0, \
\frac{\partial^2 \hat R(t)}{\partial \overline \nu \partial \rho_i} = 0,
\nonumber \\
&\frac{\partial \hat R(t)}{\partial \overline \lambda_i} = \exp(-\rho_i t),
&\frac{\partial \hat R(t)}{\partial \rho_i} = - t \ \overline \lambda_i \exp(-\rho_i t),
\nonumber \\
&\frac{\partial^2 \hat R(t)}{\partial \rho_i^2} = t^2 \ \overline \lambda_i \exp(-\rho_i t),
&\frac{\partial^2 \hat R(t)}{\partial \rho_i \partial \overline \lambda_i} = -  t \ \exp(-\rho_i t),
\nonumber \\
&\frac{\partial^2 \hat R(t)}{\partial \overline \lambda_i \partial \overline \lambda_j} = 0,
&\frac{\partial^2 \hat R(t)}{\partial \rho_i \partial \rho_j} = 0,
\quad \frac{\partial^2 \hat R(t)}{\partial \overline \lambda_i \partial \rho_j} = 0
\quad \text{ if } i\neq j.
\nonumber
\end{align}

\section{Maximum Likelihood Estimation for the Hawkes intensity}\label{section:calib_intensity}

The estimation of the Hawkes parameters, as presented in Section~\ref{section:calib_method_int}, resorts to Maximum Likelihood Estimation. The use of the MLE for Hawkes processes is well known, see for instance Ozaki~\cite{Ozaki}, and has been recently considered by Da Fonseca and Zaatour~\cite{DFZ} in a similar financial framework. In this section, we give the formula of the log-likelihood for Hawkes processes, and we derive its gradient and Hessian matrix which are necessary to use Newton-Raphson's algorithm.

We define the jump processes $J^+_t = \sum_{0<\tau<t} \indi{\Delta N_t>0}$ and $J^-_t =\sum_{0<\tau<t} \indi{\Delta N_t<0}$, i.e. $J^+$ (resp. $J^-$) makes a unit jump when $N^+$ (resp. $N^-$) jumps.
Say that we observe the realization of the process on the time interval $[0,T]$, and that we want to maximize its log-likelihood on $[t_0,T]$, with $t_0 \in [0,T)$.
Conditionally to  $(\kappa^\pm_t)_{t \in [0,T]}$, the log-likelihood of a trajectory $(J^\pm_t)_{t \in [t_0,T]}$ on the time interval $[t_0,T]$  is
(see~\cite{DVJ}, Section III Proposition 7.2)
\begin{equation}\label{likelihood_Hawkes2}
\ln \mathcal{L}(J^\pm|\kappa^\pm) = \int_{t_0}^T \ln(\kappa^\pm_{t^-}) \ \textup{d}J^\pm_t
 \ - \ \int_{t_0}^T \kappa^\pm_t \ \textup{d}t
\ + \ T.
\end{equation}
Moreover, conditionally to  $(\kappa^+_t, \kappa^-_t)_{t \in [0,T]}$, the global log-likelihood of the model is
\begin{equation}\label{likelihood_Hawkes}
\ln \mathcal{L}(J|\kappa) = \ln \mathcal{L}(J^+|\kappa^+) + \ln \mathcal{L}(J^-|\kappa^-).
\end{equation}
We now compute $\ln \mathcal{L}(J^+|\kappa^+)$.
Since we do not know the history of the process before time $t=0$, it is impossible to compute $\kappa^+_t$ exactly using equation~\eqref{kappa_gen} since it requires to know all the jumps. However, a reasonable approximation is to choose $t_0 \in (0,T)$ such that
\begin{equation}
\forall u \geq t_0, \ K(u) \ll 1,
\nonumber
\end{equation}
which yields
\begin{equation}
\kappa^+_t
\approx
\kappa_\infty
+ \sum_{0<\tau<t}K(t-\tau)\left[\indi{\Delta N_t>0} \phis(\Delta N_t /m_1) +\indi{\Delta N_t<0} \phic(-\Delta N_t /m_1) \right]
\label{kappa_approx_t0}
\end{equation}
for $t \in [t_0,T]$. Let us assume in the sequel of this section that $t_0$ is such that~\eqref{kappa_approx_t0} can be considered as an equality.

We define $\tau_0=0$ and $\tau_i, \ i \geq 1$ the ordered combined jump times of $N^+$ and $N^-$ on $[0,T]$, and $\chi(t) = \max\{i\geq0, \tau_i\leq t\}$ for $t \in [0,T]$. We also define for $i\geq1$
\begin{equation}
\theta^+_i = \phis(\Delta N^+_{\tau_i}/m_1) k^+_i + \phic(\Delta N^-_{\tau_i}/m_1) k^-_i,
\nonumber
\end{equation}
where $k^+_i = 1$ if $\tau_i$ is a jump time of $N^+$, $k^+_i = 0$ otherwise, and $k^-_i$ is defined similarly with $N^-$. One has for $t \in [t_0,T]$
\begin{equation}
\kappa^+_t = \kappa_\infty
+ \sum_{j=1}^{\chi(t)} \theta^+_j K(t-\tau_j).
\nonumber
\end{equation}
%
Distinguishing the jumps before and after~$t_0$, we get
\begin{equation}
\int_{t_0}^T \kappa^+_t \ \textup{d}t
\ = \
\kappa_\infty (T-t_0)
\ + \ \sum_{j=1}^{\chi(t_0)} \theta^+_j \left[ \underline{K}(T-\tau_j) - \underline{K}(t_0-\tau_j) \right]
\ + \ \sum_{j=\chi(t_0)+1}^{\chi(T)} \theta^+_j \left[ \underline{K}(T-\tau_j) - \underline{K}(0) \right],
\label{integ_kpp}
\end{equation}
where $\underline{K}$ is the antiderivative of $K$.  Let us turn to the other term of the log-likelihood.
We set $A^+_1 = 0$ and for $i \geq 2$
\begin{equation}
A^+_i = \overset{i-1}{\underset{j=1}{\sum}} \theta^+_j K(\tau_i-\tau_j),
\nonumber
\end{equation}
and we have
\begin{equation}
\int_{t_0}^T \ln(\kappa^+_{t^-}) \ \textup{d}J^+_t
\ = \ \overset{\chi(T)}{\underset{i=\chi(t_0)+1}{\sum}} 
k_i^+ \ln\big(\kappa_\infty + A^+_i \big).
\label{likeli_jump_part}
\end{equation}

We have the explicit expression of the log-likelihood $\ln \mathcal{L}(J^+|\kappa^+)$ from~\eqref{likelihood_Hawkes}, \eqref{likelihood_Hawkes2}, \eqref{integ_kpp} and~\eqref{likeli_jump_part}. Thus, it can be evaluated on a discrete set of points, for instance to estimate one or several parameters with a grid search. Now, to maximize the likelihood using Newton-Raphson's algorithm, one must also determine the gradient and Hessian matrix of $\ln \mathcal{L}(J^+|\kappa^+)$.

For given parameterizations of $\phis, \phic$ and $K$, we note $\pi$ an arbitrary parameter, and we have
\begin{align}
\frac{\partial \ln \mathcal{L}(J^+|\kappa^+)}{\partial \kappa_\infty}
&= \ \overset{\chi(T)}{\underset{i=\chi(t_0)+1}{\sum}}\frac{k_i^+}{\kappa_\infty + A^+_i}
\ - \ (T-t_0),
\nonumber \\
\frac{\partial \ln \mathcal{L}(J^+|\kappa^+)}{\partial \pi}
&= \ \overset{\chi(T)}{\underset{i=\chi(t_0)+1}{\sum}} \frac{k_i^+ \partial_\pi A_i^+}{\kappa_\infty + A^+_i}
\nonumber \\
& \qquad - \  \sum_{j=1}^{\chi(t_0)} \partial_\pi \{ \theta^+_j \left[ \underline{K}(T-\tau_j) - \underline{K}(t_0-\tau_j) \right]\}
\ - \  \sum_{j=\chi(t_0)+1}^{\chi(T)} \partial_\pi \{ \theta^+_j \left[ \underline{K}(T-\tau_j) - \underline{K}(0) \right]\},
\nonumber
\end{align}
which yields the gradient of the log-likelihood. For the Hessian matrix, let us note $\pi, \pi'$ two parameters (distinct or not) of $\phis, \phic$ or $K$. We have
\begin{align}
\frac{\partial^2 \ln \mathcal{L}(J^+|\kappa^+)}{\partial \kappa_\infty^2}
&= \ - \ \overset{\chi(T)}{\underset{i=\chi(t_0)+1}{\sum}} \frac{k_i^+}{[\kappa_\infty + A^+_i]^2},
\qquad
\frac{\partial^2 \ln \mathcal{L}(J^+|\kappa^+)}{\partial \kappa_\infty \partial \pi}
\ = \ - \ \overset{\chi(T)}{\underset{i=\chi(t_0)+1}{\sum}} \frac{k_i^+ \partial_\pi A_i^+}{[\kappa_\infty + A^+_i]^2},
\nonumber \\
\frac{\partial^2 \ln \mathcal{L}(J^+|\kappa^+)}{\partial \pi \partial \pi'}
&= \overset{\chi(T)}{\underset{i=\chi(t_0)+1}{\sum}} k_i^+
\left( \frac{\partial^2_{\pi \pi'} A_i^+}{\kappa_\infty + A^+_i}
- \frac{\partial_\pi A_i^+ \ \partial_{\pi'} A_i^+}{[\kappa_\infty + A^+_i]^2}
\right)
\nonumber \\
& \qquad - \  \sum_{j=1}^{\chi(t_0)} \partial^2_{\pi \pi'} \{ \theta^+_j \left[ \underline{K}(T-\tau_j) - \underline{K}(t_0-\tau_j) \right]\}
\ - \  \sum_{j=\chi(t_0)+1}^{\chi(T)} \partial^2_{\pi \pi'} \{ \theta^+_j \left[ \underline{K}(T-\tau_j) - \underline{K}(0) \right]\}.
\nonumber
\end{align}
As soon as $\underline{K}$ is known and $\phis, \phic, K, \underline{K}$ are twice differentiable w.r.t. the parameterization, it is straightforward to deduce the analytical expressions of the gradient and Hessian matrix of the log-likelihood from the preceding equations.

\section{Optimal execution with a multi-exponential Hawkes kernel}\label{appendix:opt_strat_multi}

\subsection{Proof of Theorem~\ref{thm_MIHMext}}\label{appendix:proof_MIHMext}
First, let us remark that $\E\left[\int_0^T W_tdX_t -W_TX_T \right]=0$, and we can assume without loss of generality that $\sigma=0$. 
We decompose the price process as follows. We introduce $dS^N_t=\frac\nu q   \dd N_t$, $\dd D_t^{N,i} = -\rho_i  D_t^{N,i}  \dd t  +  \frac{\lambda_i} q  \dd N_t$,  $dS^X_t= \frac\nu q    \dd X_t$ and $\dd D_t^{X,i} = -\rho_i  D_t^{X,i}  \dd t  +  \frac{\lambda_i} q   \dd X_t$, with $S^N_0=S_0$, $D_0^{N,i}=D^i_0$, $S^X_0=D_0^{X,i}=0$. We have
$$P_t=P^X_t+P^N_t, \text{ with } P^N_t=S^N_t+\sum_{i=1}^p D^{N,i}_t, \ P^X_t=S^X_t+\sum_{i=1}^p D^{X,i}_t.$$
Then, we can write the cost~\eqref{costX} as 
$$C(X)=\int_{[0,T)} P^N_u \textup{d}X_u -  P^N_T X_T+ \bar{C}(X),$$
where $\bar{C}(X)=\int_{[0,T)} P^X_u \textup{d}X_u +  \frac1{2q} \underset{\tau \in \mathcal{D}_X \cap [0,T)}{\sum} (\Delta X_\tau)^2   -  P^X_T X_T  +  \frac1{2q} \ X_T^2$. We note that $\bar{C}(X)$ is a deterministic function of~$X$ and is precisely the cost function considered in~\cite{AS_SICON}. Besides, it satisfies $\bar{C}(cX)=c^2\bar{C}(X)$ for $c\in \R$. By the same argument as in the proof of Theorem~2.1 in~\cite{AB_DynHawkes}, we get that there is no PMS if, and only if $P_t$ is a $(\cF_t)$-martingale when $X_t=0$ for any~$t$.

We now consider that  $X \equiv 0$ and write the martingale condition for~$P$ under the Hawkes model~\eqref{def_kpi}, \eqref{def_kmi} and~\eqref{price_kappa_mkv}. We have 
\begin{equation}
\dd P_t = \dd S_t + \dd D_t +\sigma \ \dd W_t
\ = \
\frac1q \ \dd N_t - \sum_{i=1}^p \rho_i D_t^i \ \dd t + \sigma \ \dd W_t
\ = \
\frac1q \ \dd \tilde{N}_t  + \sigma \ \dd W_t +  \ \dd t \sum_{i=1}^p A_t^i,
\nonumber
\end{equation}
where
\begin{equation}
A_t^i = \frac{m_1}q \delta_t^i - \rho_i D_t^i, \ \delta^i_t={\kappa^+_t}^{(i)} - {\kappa^-_t}^{(i)},
\nonumber
\end{equation}
and $\tilde{N}_t = N_t - m_1 \int_0^t \delta_u \dd u$ is a martingale. Then, $(P_t)$ is a martingale if and only if almost surely and $\dd t$-almost everywhere, $\sum_{i=1}^p A_t^i = 0 $. We have
\begin{equation}
\dd A_t^i = - \rho_i A_t^i \ \dd t + \frac{m_1}q w_i \ \dd I_t - \frac{\lambda_i \rho_i} q \ \dd N_t,
\nonumber
\end{equation}
with 
\begin{equation}\label{def_I}I_t = \int_0^t \left[(\phis-\phic)(\textup{d}N^+_u/m_1) - (\phis-\phic)(\textup{d}N^-_u/m_1)\right].
\end{equation}
In particular,  $\dd A_t^i = - \rho_i A_t^i \dd t$ between two consecutive jumps $\tau$ and $\tau'$ of~$N$. Therefore, we have $\sum_{i=1}^p A_t^i=\sum_{i=1}^p A_\tau^ie^{-\rho_i (t-\tau)}=0$ for $t\in[\tau,\tau')$ and therefore 
$A_\tau^i = 0$ for all $i$ (the equality for $t=\tau+k(\tau'-\tau)/p, k\in\{0,\dots,p-1\}$ gives a Vandermonde system). Thus, we necessarily have $A_t^i=0$ for $t\ge 0$ for any $i$. Then, 
 $\dd A_t^i=0$ gives
\begin{equation}
\frac{m_1}q w_i \ [(\phis-\phic)(\dd N^+_t/m_1) - (\phis-\phic)(\dd N^-_t/m_1)] \ = \ \frac{\lambda_i \rho_i} q \ [\dd N^+_t - \dd N^-_t]
\nonumber
\end{equation}
for all $t \geq 0$ and all $i \in \{1,\cdots,p\}$. Thus, $\phis-\phic$ must be linear on the support of the law $\mu$ of the jumps of $N^\pm$, and besides, we must have
$\forall i, \quad (\ios-\ioc) w_i = \lambda_i \rho_i$. This precisely gives~\eqref{cond_MIHM_ext}. Conversely, it is clear that~\eqref{cond_MIHM_ext} ensures that $P$ is a martingale by the same calculations. 

\subsection{Proof of Theorem~\ref{thm_opt_strat}}\label{appendix:proof_opt_strat}

As in Section~\eqref{appendix:proof_MIHMext}, we assume without loss of generality  that $\sigma=0$. We first introduce some notations to present the main results on the optimal execution. We define $\delta^i_t={\kappa^+_t}^{(i)} - {\kappa^-_t}^{(i)}$ and $\Sigma^i_t={\kappa^+_t}^{(i)} + {\kappa^-_t}^{(i)}$. From~\eqref{def_kpi}, \eqref{def_kmi} and~\eqref{expo_mixture_forK}, we have

\begin{equation}
\textup{d}\delta_t^{i}  \ = \ - \beta_i \ \delta_t^{i} \ \textup{d}t
\ + \ w_i \ \text{d}I_t
\quad , \quad
\textup{d}\Sigma_t^{i}  \ = \ - \beta_i \ (\Sigma_t^{i} - 2 \kappa_{\infty}/p) \ \textup{d}t
\ + \ w_i \ \text{d}\overline{I}_t,
\label{dyn_deltasigma}
\end{equation}
for all $i \in \{1,\cdots,p\}$, where $\overline{I}_t = \int_0^t \left[(\phis+\phic)(\textup{d}N^+_u/m_1) + (\phis+\phic)(\textup{d}N^-_u/m_1)\right]$ and $I_t$ is defined by~\eqref{def_I}.
%
%

We now proceed exactly as in~\cite{AB_DynHawkes}, Appendix~B, and only give here the main lines and use similar notations. We assume without loss of generality $q=1$. For $t\in[0,T]$, $x,d,z\in \R$ and $\delta,\Sigma \in \R^p$, we denote by $\mathcal{C}(t,x,d,z,\delta,\Sigma)$ the minimal cost to liquidate $X_t=x$ over the time interval $[t,T]$ when $D_t=d$, $S_t=z$, $\delta_t=\delta$ and $\Sigma_t=\Sigma$. We look for a function that has the following form
\begin{eqnarray}
\mathcal{C}(t,x,d,z,\delta,\Sigma) & = & a(T-t) (d-(1-\epsibis)x)^2 
\ + \ \frac12 (z - \epsibis x)^2 \ + \ (d-(1-\epsibis)x)(z - \epsibis x) \ - \ \frac{(d+z)^2}2 
\nonumber \\
& & \quad + \  (d-(1-\epsibis)x) \ \sum_{i=1}^p b_i(T-t) \ \delta_i
\ + \ \sum_{i=1}^p \sum_{i=1}^p c_{i,j}(T-t) \ \delta_i \delta_j
\nonumber \\
& & \quad + \ \sum_{i=1}^p e_i(T-t) \ \Sigma_i \ + \ g(T-t),
\label{eqn:cost_form_Hawkes}
\end{eqnarray}
with $a,b_i,c_{i,j},e_i,g: \mathbb{R_+} \rightarrow \mathbb{R}$ continuously differentiable functions. We have the limit condition $\mathcal{C}(T,x,d,z,\delta,\Sigma) \ = \ -(d+z)x + x^2/2 = \frac12(d+z-x)^2 - (d+z)^2/2$, which is the cost of a trade of signed volume $-x$. We thus have 
$$a(0) = \frac12, \ b_i(0) = c_{i,j}(0) = e(0) = g(0) = 0.$$
For an arbitrary strategy~$X$, we define $\Pi_t(X)=\int_0^t P_u dX_u +\frac{1}{2} \sum_{0\le \tau < t} (\Delta X_\tau)^2+\mathcal{C}(t,X_t,D_t,S_t\delta_t,\Sigma_t)$. This is the cost of the strategy which is equal to $X$ up to time $t$ and is then optimal. Therefore, $(\Pi_t(X),t\in[0,T])$ has to be a submartingale and is a martingale if, and only if, $X$ is optimal. We define
\begin{align}
\textup{d}A^X_t  = & \left[ Z(t,X_t,D_t,S_t,\delta_t,\Sigma_t)     +   \partial_t \mathcal{C}  
    - \rho D_t \partial_d \mathcal{C}       -  \sum_{i=1}^p \beta_i  \delta_t^{i}  \partial_{\delta_i} \mathcal{C}  
 -  \sum_{i=1}^p \beta_i  (\Sigma_t^{i}-2\kappa_\infty/p)  \partial_{\Sigma_i} \mathcal{C}  \right] \textup{d}t, \label{eqn:dA1_Hawkes}
\end{align}
where the derivatives of~$\mathcal{C}$ are taken in $(t,X_t,D_t,S_t\delta_t,\Sigma_t)$ and  
$
Z(t,x,d,z,\delta,\Sigma) := 
$
\begin{align}
\left(\frac12 \sum_{i=1}^p [\Sigma_i + \delta_i]\right) & 
\mathbb{E} \big[ \mathcal{C}(t,x,d + (1-\nu) V, z+\nu V,\delta+\varphi_{s-c}(\frac{V}{m_1})w,\Sigma+\varphi_{s+c}(\frac{V}{m_1})w) 
-  \mathcal{C}(t,x,d,z,\delta,\Sigma) \big]
\nonumber \\
+ \ \left( \frac12 \sum_{i=1}^p [\Sigma_i - \delta_i] \right) &  
\mathbb{E} \big[ \mathcal{C}(t,x,d - (1-\nu)V, z-\nu V,\delta-\varphi_{s-c}(\frac{V}{m_1})w,\Sigma+\varphi_{s+c}(\frac{V}{m_1})w) 
-  \mathcal{C}(t,x,d,z,\delta,\Sigma) \big],
\nonumber
\end{align}
with $V \sim \mu$, $\varphi_{s-c}=\phis-\phic$ and $\varphi_{s+c}=\phis+\phic$. The process $A^X_t$ is continuous and such that $\Pi_t(X)-A_t^X$ is a martingale. Given the quadratic nature of the problem, we search a process $A^X$ of the form
\begin{equation}
\textup{d}A^X_t \ = \ \frac\rho{1-\epsibis} \textup{d}t \times
 \Big[
j(T-t) (D_t-(1-\epsibis)X_t) \ - \ D_t \ + \ \sum_{i=1}^p k_i(T-t) \ \delta_t^{i}
\Big]^2.
\label{eqn:dA2_Hawkes}
\end{equation}
We now introduce the variables $y=d - (1-\epsibis)x$ and $\xi = z - \epsibis x$ and work with $(y,d,\xi,\delta,\Sigma)$ instead of $(x,d,z,\delta,\Sigma)$. From~\eqref{eqn:cost_form_Hawkes} and the definition of~$Z$, we have
\begin{align}
& \partial_t \mathcal{C}(t,x,d,z,\delta,\Sigma)  =  - \dot a \ y^2
\ - \ y \sum \dot b_i \ \delta_i \ - \ \sum  \sum \dot c_{i,j} \ \delta_i \delta_j  \ - \ \sum \dot e_i \ \Sigma_i \ - \ \dot g
\nonumber , \\
& - \rho d \ \partial_d \mathcal{C}(t,x,d,z,\delta,\Sigma)  =  - \left(2 \rho a + \frac{\rho \epsibis}{1-\epsibis}\right) \ d y
\ + \ \frac\rho{1-\epsibis} \ d^2 \ - \ \rho d  \sum b_i \delta_i
\nonumber ,\\
& - \beta_i \delta_i \ \partial_{\delta_i} \mathcal{C}(t,x,d,z,\delta,\Sigma) 
 =  - \beta_i b_i \ \delta_i y \ - \ \beta_i \delta_i \left[ 2 c_{i,i} \delta_i + \sum_{j\neq i}  c_{i,j} \delta_j \right]
\nonumber ,\\
& - \beta_i (\Sigma_i - 2 \kappa_\infty/p) \ \partial_{\Sigma_i} \mathcal{C}(t,x,d,z,\delta,\Sigma)  = 
- \beta_i e_i \ \Sigma_i \ + \ 2 \beta_i \kappa_\infty e_i /p,
\nonumber \\
& Z(t,x,d,z,\delta,\Sigma)  =  \left( m_1 \times \left[2 (1-\nu) a + \nu +\frac\epsibis{1-\epsibis}\right] 
\ + \ \sum_{k=1}^p \alpha_k b_k \right) \ y \sum_{i=1}^p \delta_i 
\ - \ \frac{m_1}{1-\epsibis} \ d \sum_{i=1}^p \delta_i 
\nonumber \\
& \hspace{3cm} + \sum_{i=1}^p \sum_{j=1}^p \left[ (1-\nu) m_1 b_i \ + \ 2 \sum_{k=1}^p c_{i,k} \alpha_k \right] \ \delta_i \delta_j
\nonumber \\
& \ + \ \sum_{i=1}^p \Bigg( m_2 \times \left[ (1-\nu)^2 a + \nu(1-\nu/2) - \frac12 \right] + (1-\nu) \sum_{k=1}^p \tilde{\alpha}_k b_k
+  \hat{\alpha} \sum_{k=1}^p \sum_{l=1}^p c_{k,l} w_kw_l 
+ \sum_{k=1}^p (\alpha_k +2 w_k \ioc ) e_k \Bigg) \ \Sigma_i
\nonumber ,
\end{align}
with $\tilde{\alpha}=\E[V\times(\phis-\phic)(V/m_1)]$, $ \hat{\alpha}=\E[(\phis-\phic)^2(V/m_1)]$. We now  identify each term of equations (\ref{eqn:dA1_Hawkes}) and (\ref{eqn:dA2_Hawkes}).

\textbf{(Eq. $dy$):} $\quad
 - \left(2 \rho a + \frac{\rho \epsibis}{1-\epsibis}\right) = - \frac{2\rho}{1-\epsibis} j$, \hspace{1cm}  \textbf{(Eq. $y^2$):} $
 - \dot a \ = \ \frac\rho{1-\epsibis}j^2$.

These two equations are the same as in~\cite{AB_DynHawkes} and  give
\begin{equation}\label{def_j}
j(u) = \frac1{2+\rho u}
\text{ and  } a(u) = \frac1{1-\epsibis} \left( \frac1{2+\rho u} - \frac\epsibis2 \right).
\end{equation}

\textbf{(Eq. $\delta_i y$):} $\quad
 - \ \dot b_i \ - \ \beta_i b_i \ + \ \sum_{k=1}^p \alpha_k b_k 
\ + \ m_1 \times \left[2 (1-\nu) a + \nu +\frac\epsibis{1-\epsibis}\right]
\ = \ \frac{2 \rho}{1-\epsibis} j k_i
$.

\textbf{(Eq. $\delta_i d$):} $\quad
- \ \rho b_i \ - \ \frac{m_1}{1-\epsibis} \ = \ -\frac{2\rho}{1-\epsibis} k_i
$,

which yields
$
k_i(u) \ = \ \frac{1-\epsibis}2 \ b_i(u) \ + \ \frac{m_1}{2\rho}
$.
Plugging this in (Eq. $\delta_i y$), we have
$
\dot b_i = -\beta_i b_i + \sum_{k=1}^p \alpha_k b_k
- \frac{2\rho}{1-\epsibis} j \left( \frac{1-\epsibis}2 b_i + \frac{m_1}{2\rho}\right)
+ m_1 \left[2 (1-\nu) a + \nu +\frac\epsibis{1-\epsibis}\right]
$, and since $j/(1-\epsibis) \ = \ a + \epsibis/[2(1-\epsibis)]$,
we have
$
\dot b_i(u) \ = \
-\beta_i b_i(u) + \sum_{k=1}^p \alpha_k b_k(u) - \frac\rho{2+\rho u} b_i(u)
\ + \ \frac{m_1}{1-\epsibis} \times \frac{1+\nu \rho u}{2+\rho u}
$. We rewrite it as

\begin{eqnarray}
\dot b(u) & = &
\left[ - H - \frac\rho{2+\rho u} I_p \right] b(u)
\ + \ \frac{m_1}{1-\epsibis} \times \frac{1+\nu \rho u}{2+\rho u} 
\times (1,\cdots,1)\tp
\label{eqn:b_general}, \\
\nonumber
\end{eqnarray}
where $I_p \in \R^{p\times p}$ is the identity matrix and $H \in \R^{p\times p}$ is given by~\eqref{def_H}. To solve equation (\ref{eqn:b_general}), we search a solution of the form $
%
b(u) =   \frac1{2+\rho u} \times [\exp(-u H) \ . \ \tilde{b}(u)]
%
$
for $u \geq 0$.
This yields
\begin{equation}
\frac1{2+\rho u} \times [\exp(-u H) \ . \ \dot{\tilde{b}}(u)]  \ = \ 
\frac{m_1}{1-\epsibis} \times \frac{1+\nu \rho u}{2+\rho u} \times (1,\cdots,1)\tp,
\nonumber
\end{equation}
thus
\begin{equation}
\dot{\tilde{b}}(u)  \ = \ 
\frac{m_1}{1-\epsibis} \times (1+\nu \rho u) \times [\exp(u H) \ . \ (1,\cdots,1)\tp].
\nonumber
\end{equation}
From the definition~\eqref{eqn:def_omega_zeta}, we have $
\exp(-u H) . \left[\int_0^u (1+\nu \rho s) \times \exp(s H)  \text{d}s \right] \ = \ u \zeta(u H) + \nu \rho u^2 \omega(u H)$ for $u\ge 0$. Since $\tilde{b}(0) = 2 b(0) = 0$, we obtain
\begin{equation}
b(u) \ = \ \frac{m_1 u}{1-\epsibis} \times \frac1{2+\rho u} \times [\{\zeta(u H)+ \nu \rho u \ \omega(u H)\} \ . \ (1,\cdots,1)\tp].
\label{eqn:b_simplified_aneqb}
\end{equation}
%
%
Equation (Eq: $\delta_id$) then gives the vector function $k(u)$
\begin{equation}
k(u) \ = \ \frac{m_1}{2\rho} \times \left\{ I_p
+ \frac{\rho u}{2+\rho u} \times [\zeta(u H)+ \nu \rho u \ \omega(u H)]
\right\} \ . \ (1,\cdots,1)\tp.
\label{eqn:k_aneqb}
\end{equation}
Thus, the functions $j$ and $k$ involved in~\eqref{eqn:dA2_Hawkes} are explicit, which guarantees that the optimal strategy is obtained as a closed formula. 

The remaining functions $c_{i,j}$, $e_i$ and $g$ do not play any role to determine the optimal strategy. By identifying the terms in $\delta_i\delta_j$, $\Sigma_i$ and the constant term, we check that they solve a system of linear ODEs. They are thus uniquely determined and well-defined on $\R_+$, and the cost function~$\mathcal{C}$ is well-defined. Thes ODEs are also important to run the verification argument, i.e. to check that $\mathcal{C}$ is indeed the optimal cost function and that the strategy $X^*$ described below is the optimal one.

We now determine the strategy $X^*$ such that $\Pi(X^*)$ is a martingale, or equivalently such that $A^{X^*}$ is constant. Equations (\ref{eqn:dA2_Hawkes}) and (\ref{def_j}) yield
\begin{eqnarray}
\textup{d}A^X_t & = & \frac\rho{1-\epsibis} \textup{d}t \times
 \left[
\frac{D_t-(1-\epsibis)X_t}{2+\rho(T-t)} \ - \ D_t \ + \ \sum_{i=1}^p k_i(T-t) \ \delta_t^{i}
\right]^2
\nonumber \\
& = & \frac{\rho/(1-\epsibis)}{[2+\rho(T-t)]^2} \ \textup{d}t \times
 \bigg[
(1-\epsibis) X_t \ + \ [1+\rho(T-t)] \ D_t
\ - \ [2+\rho(T-t)] \  \sum_{i=1}^p k_i(T-t) \ \delta_t^{i}
\bigg]^2.
\nonumber
\end{eqnarray}
Thus, $A^{X^*}$ is constant on $(0,T)$ if, and only if
\begin{equation}
\text{a.s.} \ , \ \textup{d}t \ \text{-a.e. on} \ (0,T) \ , \quad 
(1-\epsibis) X^*_t = - \ [1+\rho(T-t)] \ D_t
\ + \ [2+\rho(T-t)] \  \sum_{i=1}^p k_i(T-t) \ \delta_t^{i}.
\label{eqn:martingale_cond_Hawkes}
\end{equation}
 This equation characterizes the optimal strategy.
In particular, we obtain its initial jump $\Delta X_0^*$ at time $t=0$
\begin{equation}\label{initial_trade}
(1-\epsibis) \Delta X_0^* \ = \ 
- \frac{[1+\rho T] q D_0
+ x_0}
{2+\rho T}
\ + \ \frac{m_1}{2\rho} \times \left[(1,\cdots,1) \ . \
 \left\{ I_p + \frac{\rho T}{2+\rho T} \times [\zeta(T H)+ \nu \rho T \ \omega(T H)] \right\} \ . \ \delta_0 \right].
\nonumber
\end{equation}
where $\delta_0 = (\delta_0^{1}, \cdots, \delta_0^{p})\tp \in \R^p$.

\end{document}